\RequirePackage{fix-cm}
\documentclass[a4,12pt]{article}
\usepackage{graphicx}
\usepackage{times,longtable,color}
\usepackage{amsmath}
\usepackage{amsfonts}
\usepackage{algorithm}
\usepackage{algpseudocode}
\usepackage{graphicx}
\usepackage{booktabs}
\usepackage{array}
\usepackage{paralist}
\usepackage{verbatim}
\usepackage{subfig}
\usepackage{calc}
\usepackage{varwidth}
\usepackage{url}
\usepackage{fullpage}
\usepackage{hyperref}

\newcommand{\eps}{\varepsilon}
\newcommand{\smallestred}{\mathfrak{r}}
\newcommand{\p}{\mathfrak{p}}

\newcommand{\f}{\mathfrak{f}}
\newcommand{\A}{\mathcal{A}}
\newcommand{\B}{\mathcal{B}}

\newcommand{\Q}{\mathcal{Q}}
\newcommand{\R}{\mathcal{R}}
\newcommand{\D}{\mbox{\textsc{Redspace}}}
\newcommand{\U}{\mathcal{U}}

\newcommand{\opt}{\mbox{\textsc{opt}}}

\newcommand{\N}{\mathcal{N}}
\newcommand{\M}{\mathcal{M}}
\newcommand{\SonofH}{\mbox{\textsc{Son Of Harmonic}}}
\newcommand{\SuperH}{\mbox{\textsc{Super Harmonic}}}
\newcommand{\MarkItems}{\mbox{\textsc{Mark}} \mbox{\textsc{and}} \mbox{\textsc{Color}}}
\newcommand{\EHarm}{\mbox{\textsc{Extreme}} \mbox{\textsc{Harmonic}}}
\newcommand{\Pack}{\mbox{\textsc{Pack}}}
\newcommand{\PackS}{\mbox{\textsc{PackSimple}}}
\newcommand{\Hpp}{\mbox{\textsc{Harmonic++}}}
\newcommand{\af}{{\sc Any Fit}}
\newcommand{\nf}{{\sc Next Fit}}

\newcommand{\ff}{{\sc First Fit}}
\newcommand{\bfit}{{\sc Best Fit}}
\newcommand{\harm}{{\sc Harmonic}}

\newcommand{\rTypes}{N}
\newcommand{\rItems}{n}

\newcommand{\reditems}{n_{\mathrm{red}}^i}
\newcommand{\allitems}{n^i}
\newcommand{\bonus}{n_{\mathrm{bonus}}}

\newcommand{\leaves}{\mathrm{leaves}}
\newcommand{\needs}{\mathrm{needs}}
\newcommand{\redfit}{\mathrm{redfit}}
\newcommand{\bluefit}{\mathrm{bluefit}}
\newcommand{\blue}{x}%\mathrm{blue}}

\newcommand{\redspace}{\mathrm{redspace}}
\newcommand{\redfrac}{\mathrm{red}}
\newcommand{\Tiny}{\textsc{Tiny}}
\newcommand{\UnmixedRed}{\textsc{UnmixedRed}}

\newcommand{\finalratio}{1.5813}	
\newcommand{\finallb}{1.5762}
\newcommand{\superhratio}{1.58880}
\newcommand{\newsuperhratio}{1.5884}

\algrenewcomment[1]{\hfill // \emph{#1}}
\algnewcommand{\LineComment}[1]{\State // \emph{#1}}
\algnewcommand\algorithmicforeach{\textbf{for each}}
\algdef{S}[FOR]{ForEach}[1]{\algorithmicforeach\ #1\ \algorithmicdo}

\newcommand{\qed}{}%\hfill$\square$}
\newtheorem{invariant}{Invariant}
\newtheorem{definition}{Definition}

\newtheorem{aproperty}{Property}
\newtheorem{pproperty}{Packing Property}
\newtheorem{lemma}{Lemma}
\newtheorem{corollary}{Corollary}
\newtheorem{theorem}{Theorem}
\newenvironment{proof}{\noindent\textbf{Proof\ \ }}{\hfill$\square$}

\newlength{\LPlhbox}
\newcommand{\LPblocktag}[2]{\settowidth{\LPlhbox}{(#1)}%
	\parbox{\LPlhbox}{\begin{align}\tag{#1}#2\end{align}}%
	\hspace*{\fill}}

\begin{document}

\title{Beating the Harmonic lower bound for online bin packing\thanks{A preliminary version of this paper appeared in 43rd Int. Coll. on Automata, Languages and Programming (ICALP 2016), p. 41:1-41:14.}%\thanks{Grants or other notes
}


\author{Sandy Heydrich\footnote{Fraunhofer-Institut f\"ur Techno- und Wirtschaftsmathematik ITWM, Kaiserslautern, Germany. \url{heydrich@mpi-inf.mpg.de} Work performed at Max Planck Institute for Informatics and Graduate School of Computer Science, Saarland Informatics Campus, Germany, and supported by the Google Europe PhD Fellowship. }         \and
        Rob van Stee\footnote{Department of Mathematics, University of Siegen, 57068 Siegen, Germany, \url{rob.vanstee@uni-siegen.de}.
        Work done at the University of Leicester, LE1 7RH University Road, Leicester, UK.} %etc.
}

\iffalse 
\institute{S. Heydrich \at
              Max Planck Institute for Informatics and Graduate School of Computer Science, Saarland Informatics Campus, Germany\\
              Tel.: +49681-9325-1028\\
              \email{heydrich@mpi-inf.mpg.de}           %  \\
           \and
           R. van Stee \at
              Department of Mathematics, University of Siegen, 57068 Siegen, Germany\\
              \email{rob.vanstee@uni-siegen.de}.
              Work done at the University of Leicester, LE1 7RH Universitz Road, Leicester, UK.
}
\fi 
\date{Received: date / Accepted: date}

\maketitle

\begin{abstract}
In the online bin packing problem, items of sizes in $(0,1]$ arrive online to
be packed into bins of size 1. The goal is to minimize the number of used bins.
In this paper, we present an online bin packing algorithm with asymptotic 
competitive ratio of \finalratio. %, which constitutes 
This is the first improvement %over the algorithm {\Hpp} 
in fifteen years and reduces the gap to the lower bound by %roughly 
15\%. 
Within the well-known {\SuperH} framework, no competitive ratio below 1.58333 can be achieved.

We make two crucial changes to that framework. First, 
some of our algorithm's decisions depend on \emph{exact sizes} of items, instead of only their types.
In particular, for each item with size in $(1/3,1/2]$, we use its exact size to determine if it can be packed together with an item of size greater than $1/2$. Second, we add constraints to the linear programs considered by Seiden, in order to better lower bound the optimal solution. These extra constraints are based on marks that we give to items based on how they are packed by our algorithm.
We show that for each input, there exists a single weighting function that can upper bound the competitive ratio on it.

We use this idea to simplify the analysis of {\SuperH}, and show that the algorithm {\Hpp} is in fact \superhratio-competitive (Seiden proved 1.58889), and that {\newsuperhratio} can be achieved within the {\SuperH} framework.
Finally, we give a lower bound of {\finallb} for our new framework.
\end{abstract}

\section{Introduction}
In the online bin packing problem, a sequence of \emph{items} with sizes in the
interval $(0,1]$ arrive one by one and need to be packed into \emph{bins}, so that
each bin contains items of total size at most 1. Each item must be irrevocably
assigned to a bin before
the next item becomes available. The algorithm has no knowledge about future items.
There is an unlimited supply of bins available, and the goal is to minimize the 
total number of used bins (bins that receive at least one item).

Bin packing is a classical and well-studied problem in combinatorial optimization. 
Extensive research has gone into developing approximation
algorithms for this problem, e.g. ~\cite{CoGaJo97,GaGrUl72,FerLue81,KarKar82,Rothvoss13,GoeRot14}. Such algorithms have provably good competitive ratio for any possible
input and work in polynomial time. In fact, the bin packing problem was one 
of the first for which approximation algorithms were designed \cite{Johnso73}. 

For bin packing, we are typically interested in the long-term behavior of
algorithms: how good is the algorithm for large inputs, relative to the optimal solution? 
%If we simply compare to the optimal solution, the worst ratio 
This ratio is often determined by %some 
very small inputs. 
To avoid such pathological
instances, the {\em asymptotic competitive ratio} was introduced,
which we now define. For a given input sequence $\sigma$,
let $\A(\sigma)$ be the number of bins used by algorithm~$\A$ on $\sigma$.
The {\em asymptotic competitive ratio} for an algorithm~$\A$ is defined to be
\begin{equation}
\label{eq:apr}
\R^{\infty}_{\A} = \limsup_{n\to\infty}\sup_\sigma\left\{\left. \frac{{\A} (\sigma)}
{ \opt(\sigma)} \right| \opt(\sigma) = n\right\}.
\end{equation}
From now on, we only consider the asymptotic competitive ratio unless otherwise stated. 
%Furthermore, for our analysis we always assume that we have fixed one specific optimal solution for a given input.
For a given input, we typically consider a fixed optimal solution for the analysis.

Lee and Lee~\cite{LeeLee85} presented an algorithm called \harm, 
which partitions the interval
$(0,1]$ into $m>1$ intervals $(1/2,1], (1/3,1/2],\dots,(0,1/m]$.
The type of an item is defined as the index of the interval which contains its size.
Each type of items is packed into separate bins ($i$ items per bin for type $i=1,\dots,m-1$; type $m$ items are packed
using {\nf} in dedicated bins).
For any $\varepsilon>0$, there is a number $m$ such that the \harm\ algorithm
that uses $m$ types has a competitive ratio of at most
$(1+\varepsilon)\Pi_{\infty}$~\cite{LeeLee85}, where $\Pi_{\infty} \approx 1.69103$ for $m \mapsto \infty$.

\iffalse 
{\SuperH} algorithms classify items based on an interval partition of $(0,1]$ and give each item a color as it arrives, red or blue.
For each type $j$, the fraction of red items is some constant denoted by $\redfrac_j$. 
Blue items are packed as in \harm, i.e., for each item type $j$, every bin with blue
items contains a maximal number of blue items. (This may leave some space for
smaller red items of different types.) Red items are packed in bins which
are only partially filled. The idea is that hopefully, later blue items
of other types will arrive that can be placed into the bins with red items. \fi 

\begin{definition}
\label{def:1}
We use the following adjectives for ranges of item sizes.
\emph{Huge} means $(2/3,1]$, \emph{large} means $(1/2,2/3]$, \emph{medium} means $(1/3,1/2]$, and \emph{small} means $(0,1/3]$.
%A \emph{large} item is an item with size in $(1/2,1]$. A \emph{medium} item is an
%item with size in $(1/3,1/2]$. A \emph{small} item is an item with size at most $1/3$.
\end{definition}

If we consider the bins packed by \harm, then
it is apparent that in bins with large items, nearly half the space can remain
unused. It is better to use this space for items of other types.
After a sequence of papers which used this idea to develop ever better algorithms~\cite{LeeLee85,RaBrLL89,Richey91},
Seiden~\cite{Seiden02} presented a general framework called {\SuperH} which captures all of these algorithms. We describe it in some detail, since we reuse many concepts,
and in order to describe our modifications in a clear way.

\paragraph{The {\SuperH} framework~\cite{Seiden02}}

The fundamental idea of all \SuperH{} algorithms is to first classify items by
size, and then pack an item according to its type (as opposed to
letting the exact size influence packing decisions).
For the classification of items, we use numbers $t_1 = 1 \ge t_2 \ge \dots \ge t_\rTypes > 0$ to 
partition the interval
$(0,1]$ into subintervals $I_1,\dots,I_{\rTypes}$.
($\rTypes$ is a parameter of the algorithm.)
 We define $I_j=(t_{j+1},t_j]$
for $j=1,\dots,\rTypes-1$ and $I_{\rTypes}=(0,t_{\rTypes}]$.
We denote the type of an item $\p$ by $t(\p)$, and its size by $s(\p)$.
An item $\p$ has type $j$ if $s(\p)\in I_j$. 
A type $j$ item has size at most $t_j$. 

Each item receives a color when it arrives, red or blue; an algorithm in the {\SuperH} framework defines parameters $\redfrac_j \in [0,1]$ for each type $j$, which denotes the fraction of items of type $j$ that are colored red.\footnote{This parameter was called $\alpha_j$ by Seiden; we have made many changes to the (somewhat ad hoc) notation.}
Blue items of type $j$ are packed using \nf{}. We use each bin until 
exactly $\bluefit_j := \lfloor 1/t_j \rfloor$ items are packed into it.
For each bin, smaller red items may be packed into the space 
of size $1-\bluefit_jt_{j}$ that remains unused. 
Red items are also packed using \nf{}, using a fixed amount of the available space in a bin. 
This space is chosen in advance from a fixed set $\D=\{\redspace_i\}_{i=1}^K$ of spaces, where $\redspace_1\le\dots\le\redspace_K$. 
If red items of type $j$ are packed into a space of size $\redspace_i$, we pack $\redfit_j:=\lfloor \redspace_i/t_j\rfloor$ red items into each bin.
%This determines the number of red items of type $j$ that are packed together in one bin, which is denoted by $\redfit_j$. 
In the space not used by red items, the algorithm may pack blue items.
There may be several types that the algorithm can pack into a bin together with red items of type $j$.
Each bin will contain items of at most two different types.
If a bin contains items of two types, it is called mixed. If it contains items of only one type, but items of another type may be packed into this bin later, it is called unmixed. 
A bin that will always contain items of one type is called pure blue. %, since the items in it will all be blue.
 A {\SuperH} algorithm tries to minimize the number of unmixed bins,
and to place red and blue items in mixed bins whenever possible. 
Seiden~\cite{Seiden02} showed that the {\SuperH} algorithm {\Hpp}, which uses 70
intervals for its classification and has about 40 manually set
parameters, achieves a competitive ratio of at most 1.58889.

The algorithm {\Hpp} always packs only one red item in a bin, and Seiden exploits this fact in his analysis. However, a very minor technical change is sufficient to make his analysis more general. Since Seiden does mention the possibility of packing more than one red item in a bin, only deciding against it because he could not find good settings for the parameters, we do not see this as a new idea of our algorithm. By allowing more than one red item in a bin in the {\SuperH} framework, a competitive ratio of 1.5884 can be achieved.

Ramanan et al.~\cite{RaBrLL89} gave a lower bound of $19/12\approx1.58333$
for all {\SuperH} algorithms.
It is based on \emph{critical bins} (formally defined later) like the one shown in
Fig. \ref{fig:ramanan}, which contain a \emph{medium} item (size in $(1/3,1/2]$) and
a \emph{large} item (size in $(1/2,2/3]$). 
Both of these items arrive many times, and although they fit pairwise
into bins, the algorithm does not combine them like this. 
In contrast, the optimal solution consists exclusively of critical bins.
%No matter how fine the item 
%classification of an algorithm, pairs of items such as these, that the algorithm does
%not pack together into one bin, can always be found. 

\paragraph{Our contribution}
We avoid the lower bound construction of Ramanan et al.~\cite{RaBrLL89} by defining
the algorithm so that it combines medium and large items \emph{whenever} they 
fit together in a single bin. 
Essentially, we use {\af} to combine such items into bins (under certain
conditions specified below). This is a generalization of the well-known
algorithms {\ff} and {\bfit}~\cite{Ullman71,GaGrUl72}, which have been used in similar 
contexts before~\cite{babdss15,BaChKK04}. 
For all other items, we essentially leave the structure of {\SuperH} intact, although a number of technical changes are made, as we describe next.
Each bin will still contain items from at most two types, and if there are two types in a bin, then the items of one type are colored blue and the others are colored red.

\begin{figure}%[t!]
	\begin{center}
		\includegraphics[width=0.85\textwidth]{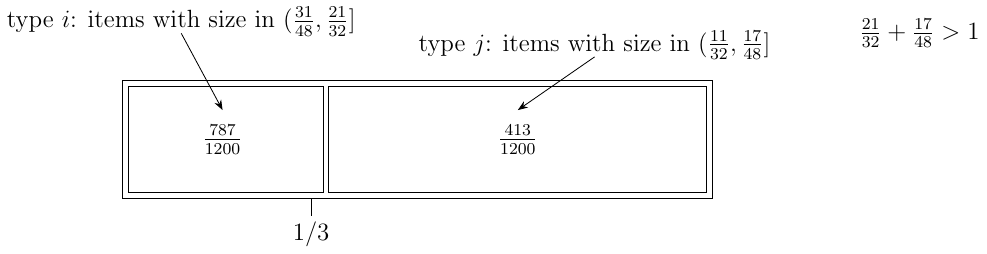}
		\caption{\label{fig:ramanan}A critical bin.
			The item sizes are chosen such that a given interval classification algorithm 
			(in this case, {\Hpp}) does not pack these items together. For any \SuperH{} algorithm, such sizes can be found.
			The central idea of our new algorithm is that we limit the number of times that these critical bins can occur in the optimal solution.
			This is how we beat the ratio of 1.58333.
		}
	\end{center}
\end{figure}

We extend the definitions of huge, large, medium and small items (Definition \ref{def:1}) to types in the natural way.\footnote{There will not be any types that contain $1/2$, $1/3$ or $2/3$ as an inner point in their interval.}
In order to benefit from using \af, we need to ensure that for each medium type, as much as possible, the \emph{smallest} items are colored red.
Otherwise, we run into the same problems as {\SuperH}, see Fig.~\ref{fig:red-too-large} for an example.
Our plan is therefore to initially give each medium item \emph{no color} and pack it alone in a bin.
After several items of some type $j$ have arrived, we color the smallest one red and the others blue.
The next arriving blue items of this type (so-called \emph{late} items) will be packed into the bins with single blue items.
%The bins with blue items are then used to pack the next items of this type, which will also be colored blue. 
(See Fig. \ref{illus}.) %We call these \emph{late} items.
In this way, at least the blue items which are \emph{first} in their bins (the \emph{early} items) are
not smaller than the smallest red item. We thus have a lower bound on the size of half of the blue items (the early ones).

However, postponing coloring decisions like this is not always possible or even desirable. In fact there are exactly two cases where this will not be done upon
arrival of a new medium item $\p$.
\begin{enumerate}
	\item 
	If a bin with suitable small red items (say, of some type $t$)
	is available, and it is time to color $\p$ blue,
	we will pack $\p$ into that bin and color it blue, regardless of the precise size of $\p$.
	In this case, in our analysis we will carefully consider how many small
	items of type $t$ the input contains; knowing that there must be some.
	This implies that in the optimal solution, not \emph{all} the bins can be critical.
	Moreover, our algorithm packs these small items very well, using almost the entire space in the bin.
	\item
	If a bin with a large item is available, and $\p$ fits into such a bin, we will
	pack $\p$ in one such bin as a red item regardless of which color it was supposed to get. 
	This is the best case overall, since finding combinations like this was exactly our goal! 
	This helps to avoid the worst case instances for {\SuperH} (Fig.~\ref{fig:all-medium-red}).
	However, there is a technical problem with this, which we discuss below.
\end{enumerate}

Overall, we have three different cases: medium items are packed alone initially (in which case we have a guarantee about the sizes of some of the blue items), medium items are combined with smaller red items (in which case these small items exist and must be packed in the optimal solution), or medium items are combined with larger blue items (which is exactly our goal).

The main technical challenge is to quantify these different advantages into one overall analysis. 
In order to do this, we introduce---in addition to and separate from the coloring---a marking of the medium items. The marking indicates whether the blue or red items of a given mark are in mixed or unmixed bins. This will bound the number of critical bins (Fig.~\ref{fig:ramanan}) that can exist in the optimal solution, leading to better lower bounds for the optimal solution value than Seiden~\cite{Seiden02} used.

\begin{figure}
	\begin{center}
		\subfloat
		[Packing produced by \SuperH{}.]
		{\includegraphics[width=0.45\textwidth]{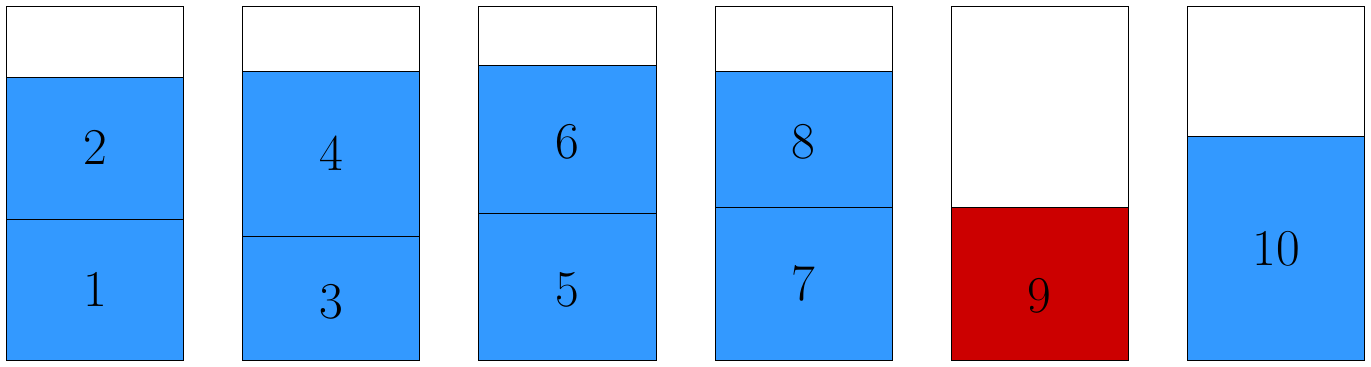}}
		\hspace{7mm}
		\subfloat
		[Optimal packing.]
		{\includegraphics[width=0.45\textwidth]{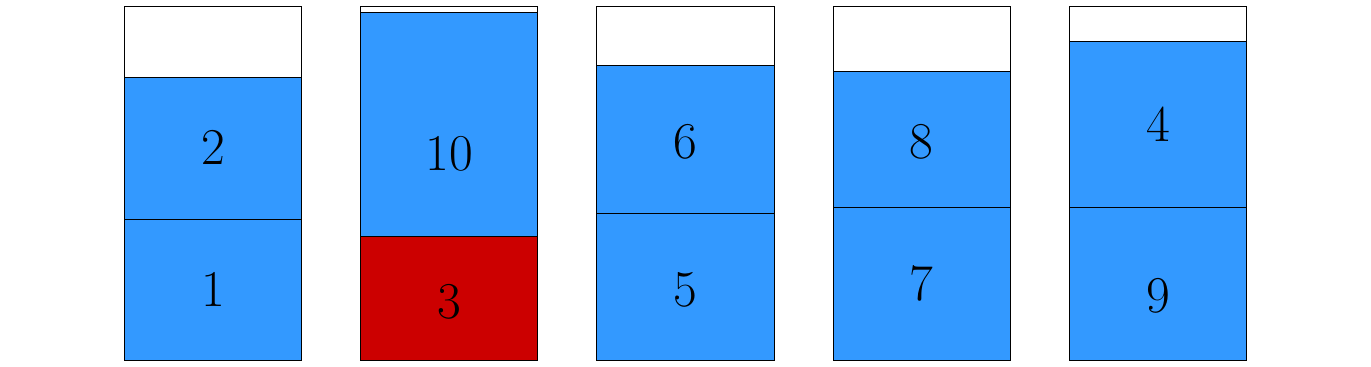}}
	\end{center}
	\caption{Packing of a sequence of medium items (all of the same type $i$) followed by one large item. The items arrive in the order indicated by the numbers. \SuperH{} needs more bins than the optimal solution, as the red medium item is too large to be combined with the large item. 
		(We assume $\redfrac_i=1/9$ here.) \label{fig:red-too-large}}
\end{figure}

\begin{figure}
	\begin{center}
		\subfloat
		[Pack items one per bin without coloring them.]
		{\includegraphics[width=0.45\textwidth]{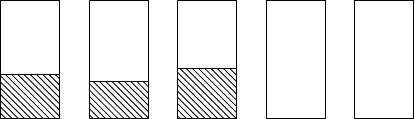}}
		\hspace{7mm}
		\subfloat
		[The fifth item arrives: time to fix the colors.]
		{\includegraphics[width=0.45\textwidth]{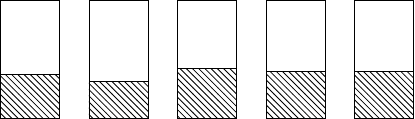}}
		\\
		\subfloat
		[The smallest item becomes red.]
		{\includegraphics[width=0.45\textwidth]{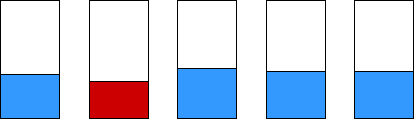}}
		\hspace{7mm}
		\subfloat
		[Additional blue items of the same type are added.]
		{\includegraphics[width=0.45\textwidth]{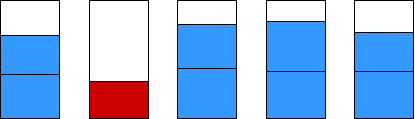}}
	\end{center}
	\caption{\label{illus}Illustration of the coloring in \EHarm{}. Hatched items are uncolored. In this example, $\redfrac_i=1/9$, where $i$ is the type of all items depicted in this example.
		Note that the ratio of $1/9$ does not hold (for the bins shown)
		at the time that the colors are fixed: $1/5$ of the items are red at this point.
		The ratio $1/9$ is achieved when all bins with blue items contain two blue items.
		The blue items which arrive in step (d) are called \emph{late} items.}
\end{figure}

\begin{figure}
	\centering
	\subfloat
	[Packing produced by \SuperH{}.]
	{\includegraphics[width=0.45\textwidth]{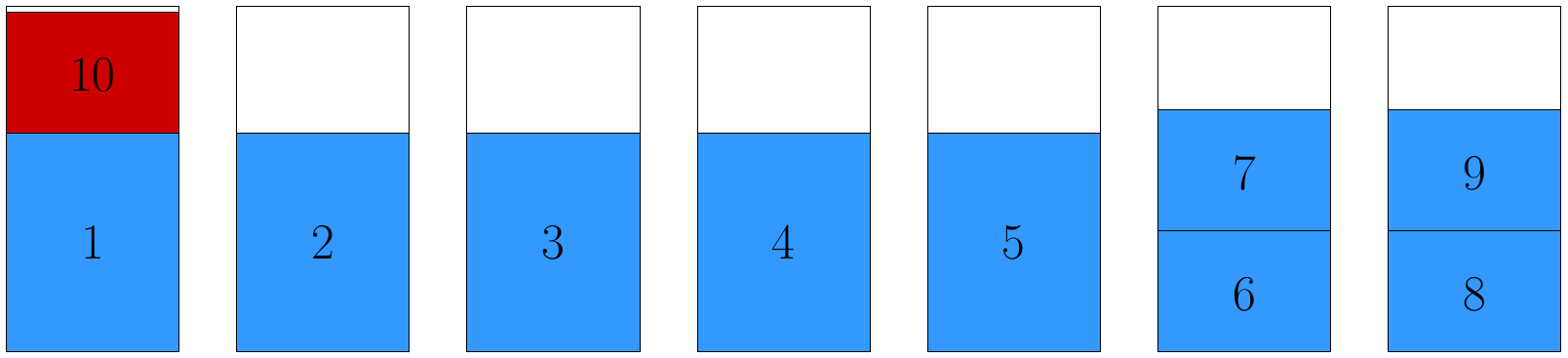}}
	\hspace{7mm}
	\subfloat
	[Optimal packing and packing produced by \EHarm{}.]
	{\includegraphics[width=0.45\textwidth]{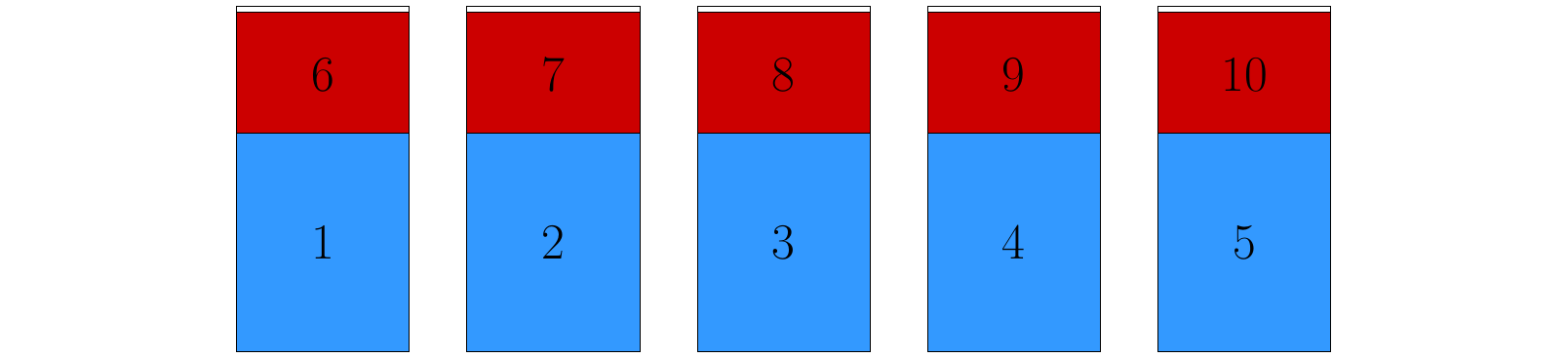}}
	\caption{An example illustrating why it helps to occasionally color more than a $\redfrac_i$-fraction of the items of a medium type $i$ red. First, five large items arrive (numbered 1-5), then five medium items (numbered 6-10) of type $i$. We assume that $\redfrac_i=1/5$ here. \label{fig:all-medium-red}}
\end{figure}

Maintaining the fraction $\redfrac_j$ of red items for all marks separately is necessary for the analysis. 
As we have seen however, if many large items arrive first, we must pack medium items with them whenever possible,
even if this violates the ratio $\redfrac_j$. 
If there are more than $\redfrac_j$ medium items of some type $j$ when the input ends,
we call those items \emph{bonus items}. Each bonus item is packed in a bin with a large item.
After the input ends, we will (virtually) make some of those large items \emph{smaller} so that they get 
type $j$ as well (see Fig.~\ref{fig:post-processing}(a)).
We then change the colors of the bonus items to ensure the proper fraction
$\redfrac_j$ of red medium items.
Hence we \emph{modify the input}, but we only do this for the analysis and only once all the items have been packed.
Clearly the number of bins in the optimal solution can only decrease as a result of making some items smaller.

However,
there could be small red items (say, of type $t$) in separate bins that could have been packed in bins with two medium type $j$ items, 
had such bins been available at the time when the small red items arrived. Creating such bins after the input ends generates a packing that 
%could not have been produced by our algorithm (Fig.~\ref{fig:post-processing}(b)). 
is not covered by our analysis (as this analysis assumes that such compatible items are packed together in one bin, not two; see Fig.~\ref{subfig:postproc-b}).
To avoid this, we \emph{do not allow} small items to be packed into new bins as red items as long as bins with large and medium items exist that may later be modified. Instead, in such a case, we \textbf{count} a single medium item in such a bin as a number of red small items of type $t$, and pack the incoming item of type $t$ as a blue item (Fig.~\ref{fig:post-processing}(c)).
This ensures (as we will show) that if suitable bins with blue items are available,
red items of type $t$ are always packed in them, rather than in new bins.

At this point, we stress that our algorithm does not actually modify the input while it is packing it in any way. The only thing that changes is the internal accounting of the algorithm (in such a way that it thinks it has packed less total size than it actually has). We will show a number of properties of the packing that the algorithm produces, and we crucially show that all of these properties are maintained whenever a bonus item is counted as several small items, which is the only point in which the accounting of the algorithm is changed relative to the actual input. Thus we do not follow the perhaps more common approach of showing that the algorithm would have performed the same on the modified input; we believe that approach cannot be applied here, as we do not see how to define arrival times for the small items that are created.

\begin{figure}[t]
	\begin{center}
	\subfloat
		[If bonus items remain after the algorithm terminates, we transform some large items to medium items, re-establishing the correct ratio of red items for the medium items (in this example, this ratio is $1/5$).]
		{\includegraphics[width=\textwidth]{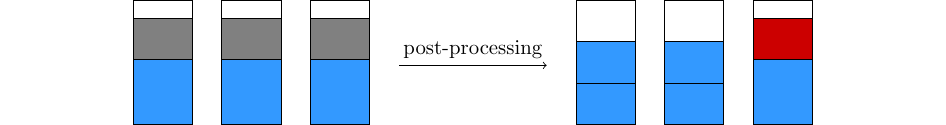}\label{subfig:postproc-a}}
	\end{center}
	\subfloat
		[This situation \textbf{must not occur} in our algorithm: We shrink a large item packed with a bonus item but there are uncombined red small items compatible with the bonus item. ]
		{\includegraphics[width=\textwidth]{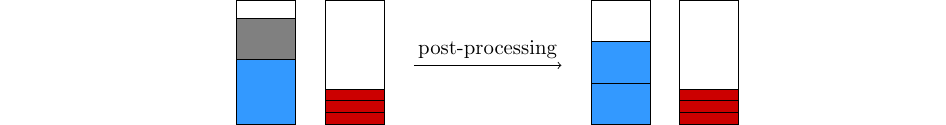}\label{subfig:postproc-b}}
	\\
	\subfloat
		[In order to prevent the situation in Fig. \ref{subfig:postproc-b}, we (virtually) resize and split the bonus item into small items when other small red items arrive. The new item becomes blue instead. Later, more small blue items can be packed with it.]
		{\includegraphics[width=\textwidth]{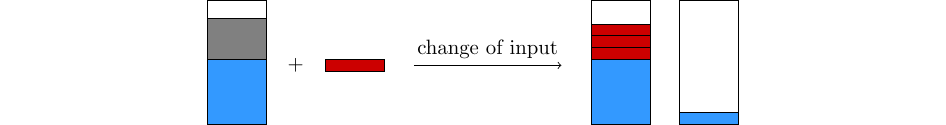}\label{subfig:postproc-c}}
	\\
	\caption{Post-processing and change of the input for the analysis. Gray items denote bonus items.
	\label{fig:post-processing}}
\end{figure}

Like Seiden~\cite{Seiden02} and many other authors~\cite{Ullman71,LeeLee85,RaBrLL89}, we use weighting functions to analyze the competitive ratio of our algorithm. 
A weighting function defines a weight for each item, depending on its type (and mark, in our case). 
By analyzing these, 
Seiden ended up with a set of mathematical programs that upper bounded
the competitive ratio of {\SuperH} algorithms. These
represented a kind of knapsack problems where each item has two
different weights.
Seiden used heuristics to solve these problems in reasonable time.

We instead split each mathematical program into two standard linear programs, and we add new constraints limiting how many critical bins there can be in the optimal solution, which can be deduced from the marks of the items. We solve the linear programs by creating a separation oracle for the dual, which solves a standard knapsack problem (with just one weight per item), making the results much easier to verify. The final weighting function we find depends on the input but does not depend on the marks anymore.
The two dual programs of each pair of linear programs are symmetric, so it is sufficient to give a solution for one of them.
%In effect, at the end we construct a single weighting function (not based on marks anymore) to upper bound the competitive ratio
%for each possible input. The function we construct depends on the input.

We implemented a computer program which quickly solves the knapsack problems and also does the other necessary work, including the automated setting of many parameters like item sizes and values $\redfrac_i$. As a result, our algorithm  {\SonofH} requires far less manual settings than \Hpp. We also provide a verifier program that checks the feasibility of these solutions; this verifier program should be easy to check by a reader. In addition, we also output the set of knapsack problems directly to allow independent verification.

This approach can also be applied to the original {\SuperH} framework. Surprisingly, we find that the algorithm {\Hpp} is in fact \superhratio-competitive,
using only one weighting function per input. 
A benefit of using our approach is that this result becomes more easily verifiable as well. 
Furthermore, we were able to improve and simplify the parameters of {\Hpp} to achieve a competitive ratio of {\newsuperhratio} within the {\SuperH} framework.

Our second main contribution is a new lower bound for all algorithms of this kind.
The fundamental property of all these algorithms is that they color a fixed fraction
of all items red (for each type). 
We show that no such algorithm can be better
than {\finallb}-competitive. Thus fundamentally different ideas will be needed to
get much closer to the lower bound of 1.54037, which we believe is closer to the
true competitive ratio of this problem.

\paragraph{Related Results}
The online bin packing problem was first investigated by
Ullman~\cite{Ullman71}.
He showed that the \ff\ algorithm has competitive ratio
$\frac{17}{10}$. This result was then published in~\cite{GaGrUl72}.
Johnson~\cite{Johnso74} showed that the \nf\ algorithm has competitive ratio 2.
Yao showed that {\sc Revised First Fit} has competitive ratio $\frac53$,
and further showed that no online algorithm has competitive ratio less
than $\frac32$~\cite{Yao80A}.
Brown and Liang independently improved this lower bound to
1.53635~\cite{Brown79,Liang80}.
The lower bound stood for a long time at $1.54014$,
due to van Vliet~\cite{Vliet92},
until it was improved to $\frac{248}{161}= 1.54037$ by Balogh et al.~\cite{BBG12}.

An improved upper bound of 1.5873 by Balogh et al.~\cite{arxiv} partially builds on our ideas but uses a different analysis. The comments about our paper in the arxiv version refer to a previous version and are outdated, as the issue it mentions has long been fixed.

The \emph{offline} version, where all the items are given in advance, 
is well-known to be NP-hard~\cite{GarJoh79}. 
This version has also received a great deal of attention,
for a survey see~\cite{CoGaJo97}.

\section{The {\EHarm} framework}
First of all, to facilitate the comparison to the new framework, we give a formal definition of
the {\SuperH} framework in Algorithm \ref{alg:sh} and \ref{alg:pack1}. It uses the following definitions.
Let $\allitems$ count the total number of items of type $i$, and $\reditems$ count the number of red items of type $i$.

\begin{algorithm}
	\caption{\label{alg:sh}How the \SuperH{} framework packs a single item $\p$ of type $i\le \rTypes-1$.
		At the beginning, we set $\reditems \gets 0$ and $n^i \gets 0$ for $1\le i \le \rTypes-1$.}
	\label{SuperH}
	\begin{algorithmic}[1]
		\State{$n^i \gets n^i + 1$}
		\If{$\reditems < \lfloor\redfrac_i n^i \rfloor$ \label{sh:centralif}}\Comment{pack a red item}
		\State{$\PackS(\p,\ $red$)$}
		\State{$\reditems \gets \reditems + 1$}
		\Else\Comment{pack a blue item}
		\State{$\PackS(\p,\ $blue$)$}
		\EndIf
	\end{algorithmic}
\end{algorithm}

\begin{algorithm}
	\caption{The algorithm $\PackS(\p,c)$ for packing an item $\p$ of type $i$ with color $c\in\{$blue, red$\}$.\label{alg:pack1}}
	\begin{algorithmic}[1]
		\State{Try the following types of bins to place $\p$ with color $c$ in this order:}
		\State{ \hspace{3mm} \textbullet\ 
			a pure blue, mixed, or unmixed $c$-open bin with items of type $i$ and 
			color $c$ \label{pack:open1}}
		\State{ \hspace{3mm} \textbullet\ 
			an unmixed bin that is \emph{compatible} with $\p$ (the bin becomes mixed) \label{pack:compatible1}}
		\State{ \hspace{3mm} \textbullet\ 
			a new unmixed bin (or pure blue bin, if $\leaves(i)=0$ and $c=$ blue)
			\label{pack:c11}}
	\end{algorithmic}
\end{algorithm}

%\paragraph{The functions $\leaves$ and $\needs$}
A {\SuperH} algorithm uses a function $\leaves:\{1,$ $\dots,\rTypes\}\to\{0,\dots,K\}$ to map each item type to an index of a space in $\D$, indicating
how much space for red items it leaves unused in bins with blue items of this type. Here $\leaves(j)=0$ means that no space is left for red items.
The algorithm also uses a function $\needs:\{1,\dots,\rTypes\}\to\{0,\dots,K\}$ to map
how much space (given by an index of $\D$) red items of each type require.
We define $\needs(i)=0$ if and only if $\redfrac_i=0$ (i.e., there are no red items of this type).

%Recall that pure blue bins contain only blue items. 
For each type $i$ such that $\leaves(i)=0$, 
the items of this type are packed in pure blue bins, that contain only blue items (only one type per bin).
An unmixed bin is called unmixed blue or unmixed red depending on the color of the items in it.

A mixed bin with blue items of type $i$ and red items of type $j$ satisfies the following properties:
$\leaves(i)>0,\redfrac_j>0,\redspace_{\needs(j)}\le \redspace_{\leaves(i)}$.
Note that the last inequality holds if and only if $\needs(j)\le \leaves(i)$.
The blue items will use space at most $1-\redspace_{\leaves(i)}$ and the red items will use
space at most $\redspace_{\needs(j)} \le \redspace_{\leaves(i)}$.
%
%\begin{definition}
%	\label{def:compatible}
	An unmixed blue bin with blue items of type $j$ is \emph{compatible} with a red item of type $i$
	if $\needs(i)\le \leaves(j)$.
	An unmixed red bin with red items of type $j$ is {compatible} with a blue item of type $i$ if $\needs(j)\le \leaves(i)$.
%\end{definition}
%In both cases, the condition means that the blue items and the red items together would use at most 1 space in the bin (the blue items leave enough space for the red items and vice versa).

\begin{definition}
	A bin is \emph{red-open} if it contains some red items but can still receive additional red items. 
	We define \emph{blue-open} analogously. 
	A bin is \emph{open} if it is red-open or blue-open. 
\end{definition}

Red-open bins with red items of type $j$ contain at least one and at most $\redfit_j-1$ red items.
Blue-open bins can be pure blue.
Red-open and blue-open bins can be mixed or unmixed.
Mixed bins can be red-open and blue-open at the same time.
A bin with $\bluefit_i$ items of type $i$ but no red items is not considered open, 
even though red items might still be packed into it later.

For {\EHarm}, we extend the definition of compatible bins.
As noted in the Introduction, some items will not receive a color when they arrive, but only later.
The goal of having uncolored items is to try and make sure that relatively small items of each medium type become red in the end
(to make it easier to combine them with large items).

\begin{definition}
\label{def:comp}
An unmixed bin is \emph{red-compatible} with a newly arriving item $\p$ if
\begin{enumerate}
\item
the bin contains blue or uncolored items\footnote{\label{fn4}We will see later that if an item has no color, it is the only item in its bin (Property \ref{prop:nocolor}).} of type $i$, %, %$\p$ will be colored red, 
$\p$ is small
and $\leaves(i)\ge \needs(t(\p))$, \emph{or}
\item
the bin contains a (blue) large item of size $x$, $\p$ is medium and $s(\p)\le 1-x$.
\end{enumerate}
An unmixed bin is \emph{blue-compatible} with a newly arriving item $\p$ if
\begin{enumerate}
\item
the bin contains red items\footnotemark[\value{footnote}] of type $j$, %$\p$ will be colored blue, 
$\p$ is medium or small
and $\leaves(t(\p))\ge \needs(j)$, \emph{or}
\item
the bin contains one red or uncolored medium item of size $x$, $\p$ is large and $s(\p)\le 1-x$.
\end{enumerate}
\end{definition}

It follows that for checking whether a large item and a medium item can be combined in a bin, we ignore the values $\leaves(i)$ and $\needs(j)$ and use only the relevant parts 2 of
Definition \ref{def:comp}.

Like {\SuperH} algorithms, an {\EHarm} algorithm first tries to pack a red (blue) 
item into a red-open (blue-open) bin with items of the same type and color; then it tries to find an unmixed compatible bin; if all else fails, it opens a new bin. 
Note that the definition of compatible has been extended compared to {\SuperH},
but we still pack blue items with red items of another type and vice versa;
there will be no bins with blue (or red) items of two different types.
The new framework is formally described in Algorithms \ref{alg:eh} and \ref{alg:pack}.
Items of type $\rTypes$ are packed using \nf{} as before.
We discuss the changes from {\SuperH} one by one. 
All the changes stem from our much more careful packing of medium items.
The algorithm {\MarkItems} called in line \ref{callmark} of {\EHarm} will be presented in Section \ref{sec:mark-and-color}.
This algorithm will take care of assigning marks and colors to the items. In particular, this will take care of fixing the color of medium items as described in Figure \ref{illus}.

As can be seen in {\Pack} (lines \ref{pack:open}, \ref{pack:c1} and \ref{pack:color}), medium items that are packed into new bins are initially packed \textbf{one} per bin and not given a color.
We wait until enough of these items have arrived, and then color the smallest one red using
{\MarkItems} (Fig. \ref{illus}). 
Note that $\reditems$ is increased in line \ref{eh:fixbonus4} even though the item might not receive a color at this time.
This means that the value $\reditems$ does not alway accurately reflect how many red items there currently are.
We will show that this is not an issue for the analysis (it will be accurate up to a constant).

\begin{algorithm}
	\caption{\label{alg:eh}How the \EHarm{} framework packs a single item $\p$ of type $i\le \rTypes-1$.
		At the beginning, we set $\reditems \gets 0$, $\bonus^i \gets0$
		and $n^i \gets 0$ for $1\le i \le \rTypes-1$.}
	\label{ExtremeH}
	\begin{algorithmic}[1]
		\State{$n^i \gets n^i + 1$\label{eh:ni+1}}
		\If{$\reditems < \lfloor\redfrac_i n^i \rfloor$ \label{eh:centralif}}\Comment{pack a red item}
		\If{$\bonus^i>0$ or $\needs(i)\le 1/3 \land \exists j: \bonus^j>0$
			\label{eh:fixbonus} }
		\LineComment{special case: replace bonus item instead and pack the new item as blue; see Fig. \ref{subfig:postproc-c}}
		\If{$\bonus^i>0$}
		\State{Let $\mathfrak b$ be a bonus item of type $i$}\Comment{in this case, $\redfit_i=1$}
		\Else
		\State{Let $\mathfrak b$ be a bonus item of some type $j$ with $\bonus^j>0$}\Comment{here $i$ is a small type}
		%		$\needs(i)\le \leaves(j)$; 
		\EndIf
		\State{$\bonus^{t(\mathfrak b)} \gets \bonus^{t(\mathfrak b)}-1$}
		\State{Label $\mathfrak b$ as type $i$} \Comment{count $\mathfrak b$ as type $i$ item(s) and color it/them red}\label{eh:fixbonus1}
		\State{$n^i \gets n^i+\redfit_i$}\Comment{$\mathfrak b$ might have been of type $i$ already, then $\redfit_i=1$}\label{eh:countasi} 
		\State{$\reditems \gets \reditems+\redfit_i$
			\label{eh:fixbonus2}}
		\State{$\Pack(\p,\ $blue$)$\label{eh:fixbonus3}} \Comment{since we now have $\reditems \ge \lfloor\redfrac_i n^i \rfloor$ again}
		\Else
		\State $\Pack(\p,\ $red$)$
		\State{$\reditems \gets \reditems+1$ \label{eh:fixbonus4}} \Comment{The item is red or uncolored}
		\EndIf
		\Else \Comment{pack a blue item}
		\If {$\p$ is medium, $\redfrac_i>0$, and there exists a bin $B$ that is red-compatible with $\p$  	\label{eh:early1}}
		\State{Place $\p$ in $B$ and label it as bonus item. \Comment{special case: bonus item}\label{eh:bonus}}
		\State{$n^i \gets n^i-1$\label{eh:ni-1}\Comment{we do not count this item for type $i$}}
		\State{$\bonus^i \gets \bonus^i+1$\label{eh:early2}} \Comment{Note that $B$ contains a large item}
		
		\Else	\State{$\Pack(\p,\ $blue$)$} \Comment{The item is blue or uncolored}
		\EndIf
		\EndIf
		\State{Update the marks and colors using {\MarkItems} (Section \ref{sec:mark-and-color}).\label{callmark}}
	\end{algorithmic}
\end{algorithm}

\begin{algorithm}
	\caption{The algorithm $\Pack(\p,c)$ for packing an item $\p$ of type $i$ with color $c\in\{$blue, red$\}$.\label{alg:pack}}
	\begin{algorithmic}[1]
		\State{Try the following types of bins to place $\p$ with (planned) color $c$ in this order:}
		\State{ \hspace{3mm} \textbullet\ 
			a pure blue, mixed, or unmixed $c$-open bin with items of type $i$ and 
			color $c$ \label{pack:open}}
		\State{ \hspace{3mm} \textbullet\ 
			a $c$-compatible unmixed bin
			(the bin becomes mixed, with fixed colors of its items) \label{pack:compatible}}
		\State{ \hspace{3mm} \textbullet\ 
			a new unmixed bin (or pure blue bin, if $\leaves(i)=0$ and $c=$ blue)
			\label{pack:c1}}
		\State{If $p$ was packed into a new bin, $\p$ is medium and $\redfrac_i>0$, give $\p$ \textbf{no} color, else give it the color $c$. 	%If $\p$ got the color red, $\reditems \gets \reditems + 1$.
			\label{pack:color}}
	\end{algorithmic}
\end{algorithm}

When an item arrives, in many cases, we cannot postpone assigning it a color, since a $c$-open or $c$-compatible bin is already available (see
lines \ref{pack:open}--\ref{pack:compatible} of $\Pack(p,c)$).
Additionally, if we are about to color an item blue because currently $\reditems \ge \lfloor\redfrac_i n^i \rfloor$, we check whether a suitable large item has
arrived earlier. We deal with this case in lines \ref{eh:early1}--\ref{eh:early2}
of {\EHarm}.
In this special case, we \emph{ignore} the value $\redfrac_i$.
We pack the medium item with the large item as if it were red (no further item will be packed into this bin), but we \emph{do not count} it towards the total number of existing medium items of its type; instead we label it a \textbf{bonus item}. Bonus items do not have a mark or color,
but this can change later during processing in the following two cases.
%
%	Relevance??
%
%Bins with bonus items are considered to be mixed bins; they are not red-open or blue-open since they already contain both a medium and a large item, so no other medium or large item could be added.

%If we create bonus items of type $i$, there are too many items of type $i$ that are packed as red items (more than a $\redfrac_i$ fraction). There are several ways that this can be fixed later on. 
\begin{enumerate}
\item
Additional items of type $i$ arrive which are packed as blue items. If enough of them arrive (so that it is time to color an item red again), we first check in line \ref{eh:fixbonus} of {\EHarm} if there is a bonus item of type $i$ that we could color red instead. If there is, we will do so, and pack the new item as a blue item.
\item
An item of some type $j$ and size at most $1/3$ arrives, that should be colored red.
In this case, for our accounting, we view the bonus item as $\redfit_j$ red items of type $j$, and adjust the counts accordingly in lines \ref{eh:fixbonus1}--\ref{eh:fixbonus2}
of \EHarm.\footnote{Note that the meanings of $i$
and $j$ are switched in the description of the algorithm for reasons of presentation.}
The new item of type $j$ is packed as a blue item in line \ref{eh:fixbonus3} of 
\EHarm.\footnote{Strictly speaking, we only need this whole procedure if type $j$ is compatible with the bonus item, to avoid the case in Figure \ref{subfig:postproc-b}.
Instead, we do it for all small items for simplicity.}
\end{enumerate}

It can be seen that blue items of size at most $1/3$ are packed as in {\SuperH}.
For red items of size at most $1/3$, we deal with existing bonus items in lines \ref{eh:fixbonus1}--\ref{eh:fixbonus2} of \EHarm, and in line \ref{pack:compatible} of $\Pack(\p,c)$, an existing medium item may be colored red or blue (the opposite of the parameter $c$).
Otherwise, the packing proceeds as in {\SuperH} for these items as well.

\subsection{Properties of {\EHarm} algorithms}

All {\EHarm} algorithms are required to satisfy the following properties.
The first two easy properties also hold for {\SuperH} and the third and fourth property hold for \Hpp{} (but not necessarily for all \SuperH{} algorithms).
Let $\eps=t_\rTypes$.
\begin{aproperty}[Lemma 2.1 in Seiden~\cite{Seiden02}]
	\label{prop:nextfit}
	Each bin containing items of type $N$, apart from possibly the last one, contains items of total size at least $1-\eps$.
\end{aproperty}
\begin{aproperty}\label{prop:leaves-needs}
	For any type $i$, if $\needs(i)>0$, then 
	${\leaves(i)} < \needs(i)$. %unless $\leaves(i)=\needs(i)=0$.
\end{aproperty}
\begin{proof}
If $\needs(i) > 0$, then $\redfrac_i>0$.
If $\leaves(i)\ge\needs(i)>0$, an additional item of type $i$ could be placed in the space
$\redspace_{\needs(i)}\le \redspace_{\leaves(i)}$, which means we could fit $\bluefit_i+1$ blue items of type $i$ into one bin,
contradicting the definition of $\bluefit_i$.
\end{proof}

\begin{aproperty}
	\label{prop:redspace}
	If $j$ is a small type with $\redfrac_j>0$, $\redspace_{\needs(j)} \le 1/3$.
	If $i$ is a medium type, then
	$\redspace_{\leaves(i)} < 1/3 < \redspace_{\needs(i)}$. % and $t_i\in\D$. %, and if additionally $\redfrac_i>0$, then $1/3 < \redspace_{\needs(i)}$. 
%	Conversely, if $x>1/3$ and $\forall i:t_i\not=x$, then $x\notin \D$.
	If $i$ is a large or huge type, then $\redfrac_i=0$, so $\reditems=0$ at all times.
\end{aproperty}
\begin{aproperty}
	\label{prop:tidelta}
	For $x>1/3$, we have $x\in\D$ if and only if $\exists i: x=t_i$.
\end{aproperty}

\begin{aproperty}
	\label{prop:redfrac}
	We have $\redfrac_i < 1/3$ for all types $i$.
\end{aproperty}

\begin{aproperty}
	\label{prop:type12N}
	We have $t_1=1, t_2=2/3, t_3=1/2,$ and $\redfrac_1=\redfrac_2=\redfrac_N=0$.
	All type 1 items (i.e., huge items) and type $N$ items are packed in pure blue bins.
	We have $1/3\in\D$.
\end{aproperty}

This property implies that $\leaves(1)=0$ and $\redspace_{\leaves(2)}=1/3$.
This means that an unmixed bin with a large item is never red-compatible
with a medium item via Condition 1 of Definition \ref{def:comp}
(so only Condition 2 is relevant for this combination).
This furthermore implies that for a medium item of type $i$, the precise value
$\needs(i)$ is irrelevant for the algorithm (only the fact that $\redspace_{\needs(i)}>1/3$ is relevant).
It will nevertheless be useful for the analysis to have $t_i\in\D$
as required by Property \ref{prop:tidelta}.

\begin{aproperty}\label{prop:q1unique}
	Let $t_i, t_{i+1}$ be two consecutive medium type thresholds of the algorithm. Then $t_i-t_{i+1}<t_N<1/100$.
\end{aproperty}

\begin{aproperty}
	For each type $i$ and color $c$, at any time, there is at most one $c$-open bin that contains items of type $i$ and no other type.
	For each pair of types and color $c$, at any time, there is at most one $c$-open bin with items of those types.
\end{aproperty}
\begin{proof}
	All bonus items are in mixed bins; such a bin remains mixed if its bonus item gets labeled with a different type.
	Consider an item of type $i$ and color $c$.
By the order in which {\Pack} tries to place items into bins, we only open a new unmixed or pure blue bin of type $i$ if no $c$-open bin is available, so the first claim holds. 

Now consider a pair of types. Say the blue items are of type $i$ and the red items are of type $j$. 
The only cases in which a mixed bin with such items is created are the following:
\begin{itemize}
\item A red item of type $j$ is placed into an unmixed bin $B$ with blue items of type $i$. In this case, there was no existing 
red-open mixed or unmixed bin with red items of type $j$.
\item A blue item of type $i$ is placed into an unmixed bin $B$ with red items of type $j$.
In this case, there was no existing blue-open mixed, unmixed or pure blue bin with blue items of type $i$.
\item A bin receives a bonus item in line \ref{eh:bonus} of {\EHarm} and is now considered mixed.
\item A bonus item gets counted as items of type $j$ in lines \ref{eh:fixbonus1}-\ref{eh:fixbonus2}.
\end{itemize}

At the beginning, there are zero open bins with items of type $i$ and type $j$.
Such bins are only created via one of the cases listed above. 
In the last two cases, no open bins are created (note that only one medium and one large item can be packed together in a bin).
In the second case, $B$ is the only red-open bin with these types (if $\redfit_i>1$, that is), and no
red items of type $j$ are packed into unmixed bins with blue items until $B$ contains $\redfit_i$ type $i$
items by line \ref{pack:open} of Algorithm \ref{alg:pack}. In the first case, similarly, no new blue item of type $i$ will be packed into unmixed bins with red items
as long as $B$ remains blue-open.
\qed\end{proof}

\begin{aproperty}
\label{prop:redok}
At all times, for each type $i$, $\reditems\ge\lfloor\redfrac_i\allitems\rfloor-1$.
For each medium type $i$, % if we disregard the actions taken by \MarkItems,
$\reditems\le \lfloor\redfrac_i\allitems\rfloor$.
For each small type $i$, $\reditems\le \lfloor\redfrac_i\allitems\rfloor+\redfit_i$.
\end{aproperty}
\begin{proof}
The first bound follows from the condition in line \ref{eh:centralif} of {\EHarm} and because $\allitems$ increases by at most $1$ in between two consecutive times that this condition is tested, unless a bonus item is labeled in lines \ref{eh:fixbonus1}--\ref{eh:fixbonus2} of {\EHarm}; but in that case, the fraction of red items of type $i$ only increases, because
$\allitems$ and $\reditems$ increase by the same amount.

The upper bounds follow because for each medium type $i$, $\reditems$ increases by at most 1 when $\reditems<\lfloor\redfrac_i\allitems\rfloor$ and a new item of this type arrives: either in line \ref{eh:fixbonus2} ($\redfit_i=1$ for medium items) or in
line \ref{eh:fixbonus4}. Furthermore, if $\reditems=\lfloor\redfrac_i\allitems\rfloor$, $\reditems$ is not increased anymore.
If a bonus item is created, $\allitems$ and $\reditems$ are unchanged (lines \ref{eh:ni+1} and \ref{eh:ni-1}).
For small items, $\reditems$ increases by at most $\redfit_i$ in one iteration (line \ref{eh:fixbonus2}), and this only happens if the ratio is too low (line \ref{eh:centralif}).
\qed\end{proof}

Recall that $\reditems$ is not always the true number of medium red items of type $i$, as some of these may not have a color yet.
For a small type $i$, the value $\reditems$ may also not be accurate, because it may include some bonus items. We will fix this in postprocessing, where we replace the bonus items by items of type $i$ to facilitate the analysis.

\begin{aproperty}
At all times, for each type $i$ that is not medium, $\bonus^i=0$.
\end{aproperty}

\begin{aproperty}
\label{prop:nocolor}
Each bin with an uncolored item contains only that item.
\end{aproperty}
\begin{proof}
By line \ref{pack:compatible} of the {\Pack} method, as soon as a bin becomes mixed, the colors of its items are fixed. 
By line \ref{pack:open} of the {\Pack} method, an unmixed bin with an uncolored item does not receive a second item of the same type.
\qed
\end{proof}

In particular, no bin which contains an uncolored item is a mixed bin.
The following important invariant generalizes a result for {\SuperH}
(which is %not formally proved in Seiden~\cite{Seiden02}, but is quite 
easy to see for that algorithm).

\begin{invariant}
\label{inv:unmixed}
If there exists an unmixed bin with red items of type $j$,
then for any type $i$ such that $\needs(j)\le \leaves(i)$, 
there is no bin with a bonus item of type $i$  
and no unmixed bin with blue items of type $i$.
\end{invariant}
\begin{proof}
%The fact that there exists no unmixed bin with a blue item of type $i$ follows
%from . 
As long as an unmixed red bin
with items of some type $j$ exists, no unmixed blue bin with items of type $i$ for which
$\needs(j)\le \leaves(i)$ can be opened and vice versa (line \ref{pack:compatible} of \Pack).

Now assume for a contradiction that there is an unmixed red bin with red items of type $j$ (denote the first item in this bin by $\f$) and a bin with a bonus item $\mathfrak b$ of type $i$.
Assume $\mathfrak b$ arrived before $\f$. Consider the point in time where $\f$ arrived. After deciding that $\f$ should be colored red in line \ref{eh:centralif} of \EHarm, we would have found that the second part of the condition in line \ref{eh:fixbonus}  of {\EHarm} is true, and as a consequence would have made $\mathfrak b$ no longer be bonus, a contradiction to our assumption.

Now assume that $\f$ arrived before $\mathfrak b$. In this case, either $\f$ or the large item $\mathfrak L$ that is packed with $\mathfrak b$ arrived first. (Note that $\mathfrak b$ definitely arrived after $\mathfrak L$, or it would not have been made bonus.)
Now $s(\mathfrak L)<2/3$ since $\mathfrak L$ was packed with the medium item $\mathfrak b$. But $1/3> \redspace_{\leaves(i)}\ge \redspace_{\needs(j)}$
by Property \ref{prop:redspace} and the assumption of the lemma.
Hence, regardless of which item among $\f$ and $\mathfrak L$ arrived first, the algorithm does not pack them in different unmixed bins; the second arriving item would be packed at the latest by line \ref{pack:compatible} of {\Pack}.
\qed\end{proof}

\subsection{Marking the items}
\label{sec:mark-and-color}

\begin{definition}
A \emph{critical bin} for an {\EHarm} algorithm is a bin used in the optimal solution that contains a pair of items,  one of a medium type $j$ ($t_j\in(1/3,1/2]$) and one of a large type $i$ ($t_i\in(1/2,2/3]$) such that $t_j+t_i>1$ but $t_{j+1} + t_{i+1}<1$.
\end{definition}
An example was given in Fig.~\ref{fig:ramanan}.
By marking the medium items, we keep track of how many red and blue items of a given type $j$ are in mixed bins.
Blue medium items in mixed bins imply the existence of compatible small items in the input
(which need to be packed somewhere in the optimal solution).
Red medium items in mixed bins means that the algorithm managed to combine at least some pairs of medium and large items
together into bins. In both cases, we have avoided the situation where the offline packing consists \emph{only} of
critical bins, whereas the online algorithm did not create \emph{any} bins which contain a large and a medium item.
%This will allow us to give a stronger lower bound on the optimal solution than Seiden~\cite{Seiden02} used.
%In particular, it will no longer be the case that \emph{all} bins in the optimal solution can be critical.
%
\begin{figure}[t]
	\begin{center}
	\subfloat
		[Items get mark $\R$: uncolored items and a red item in a mixed bin. The bins with blue $\R$-items will receive an additional blue item of the same type before any new bin is opened for this type.]
		{\includegraphics[width=\textwidth]{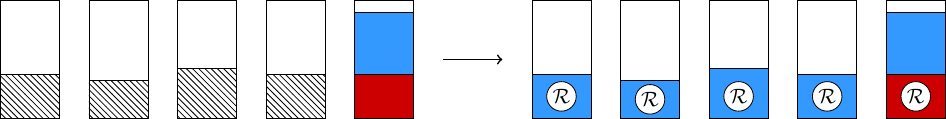}}
	\\
	\subfloat
		[Items get mark $\B$: a single uncolored item and blue items (in pairs) in mixed bins.]
		{\includegraphics[width=\textwidth]{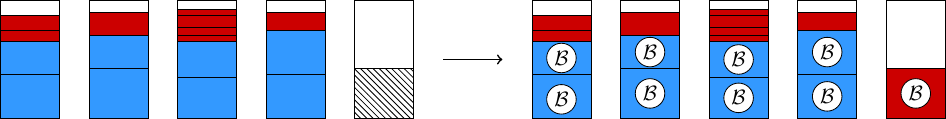}}
	\\
	\subfloat
		[Items get mark $\N$: a set of uncolored items. The bins with blue $\N$-items will receive an additional blue item of the same type before any new bin is opened for this type. See Fig. \ref{illus}.]
		{\includegraphics[width=\textwidth]{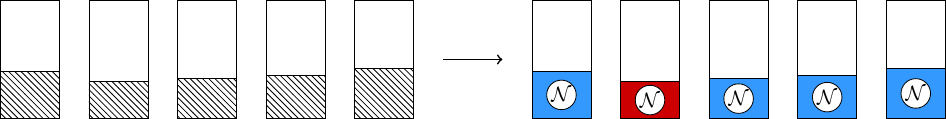}}
	\end{center}
	\caption{Marking the items.
	For simplicity, we have taken $\redfrac_i=1/9$ here (where $i$ is the type of the medium items).
	\label{fig:marking}}
\end{figure}
We use three different marks, which together cover all the cases. Our marking is illustrated in Fig. \ref{fig:marking}. 
\begin{description}
\item[$\R$] For any medium type $j$, a fraction $\redfrac_j$ of the items marked $\R$
are red, and all of these
\textbf{red} items are packed into mixed bins (i.e., together with a large item).
\item[$\B$] For any medium type $j$, a fraction $\redfrac_j$ of the items marked $\B$
are red, and the \textbf{blue} items are packed into mixed bins (i.e., together with small red items).
\item[$\N$] For any medium type $j$, a fraction $\redfrac_j$ of the items marked $\N$
are red, and \textbf{none} of the red and blue items marked $\N$ are packed into mixed bins.
\end{description}

The algorithm {\MarkItems} is defined in Algorithm \ref{alg:mark}. 
%Recall that $\allitems$ counts the total number of items of type $i$, and $\reditems$ counts the number of red items of type $i$. 
For a given type $i$ and set $\M \in \{\N,\B,\R\}$,
denote the number of red items by $\reditems(\M)$,
%the number of blue items by $\blueitems(\M)$, 
and the total number of items by $\allitems(\M)$. 
Algorithm \ref{alg:mark} is run every time after an item has been packed, 
and for every medium type $i$ for which $\redfrac_i>0$ separately.
It divides the medium items into three sets $\N, \B$ and $\R$ (see Fig. \ref{fig:marking}). 
Once assigned, an item remains in a set
until the end of the input (after which it may be reassigned, see Section \ref{sec:final}). In many cases, the algorithm will have nothing to do, as none of the conditions hold.
Therefore, some items will remain temporarily unmarked, in a set $\U$. 
The set $\U$ does not contain the bonus items (in fact none of the sets does).

\begin{algorithm}[!t]
\caption{\label{alg:mark}
The algorithm {\MarkItems} as applied to \emph{medium} items of type $i$ for which $\redfrac_i>0$.}
\begin{algorithmic}[1]
\If {there is an unmarked blue item $\p_1$ in a bin with a marked blue item $\p_2$}
\State Give $\p_1$ the same mark $\M$ as $\p_2$. \label{mark:1}
\State $\allitems(\M) \gets \allitems(\M)+1$ \label{mark:2}
\EndIf
\State{Let $\blue_\R$ be the minimum integer value such that
$\lfloor (\allitems(\R)+ 2\blue_\R+1)\redfrac_i\rfloor > \reditems(\R).$}
\If{there exist $\blue_\R$ uncolored non-bonus items and one unmarked
 red or bonus item in a mixed bin\label{mark:newr1}}
\State{Mark these $x_\R+1$ items $\R$. If there is a choice of items in mixed bins, use a bonus item\par if possible
	and color it red. Color the (other) uncolored items blue.\par
	If a bonus item is used, $\bonus^{i} \gets \bonus^{i}-1$.
\label{mark:newr2}}
\State{$\allitems(\R) \gets \allitems(\R)+\blue_\R+1,$ 
$\reditems(\R) \gets \reditems(\R)+1$ \label{mark:newr3}}
\EndIf
\State{Let $\blue_\B$ be the minimum integer value such that
$\lfloor (\allitems(\B) + \blue_\B +1)\redfrac_i \rfloor > \reditems(\B)$.}
\If
{there exists an uncolored non-bonus item and a set of mixed bins with 
two unmarked blue items each, which contains a number
$\blue'_\B\in\{\blue_\B, \blue_\B+1\}$ of blue items in total, \label{mark:newblue1}} 
\State{Mark these $\blue'_\B$ items $\B$ and color the uncolored item red. \label{mark:newblue2}}
\State{$\allitems(\B) \gets \allitems(\B)+\blue'_\B+1$,  $\reditems(\B) \gets \reditems(\B)+1$}
\EndIf
\State{Let $\blue_{\N}$ be the minimum integer value such that 
$\lfloor(\allitems(\N)+2\blue_{\N}+1)\redfrac_i \rfloor > \reditems(\N)$.}
\If
{there exist $\blue_{\N}+1$ uncolored non-bonus items
\label{mark:check}}
\State{Mark the $\blue_{\N}$ largest uncolored items and the single smallest uncolored item $\p$\par with the mark $\N$. Color $\p$ red and the other $\blue_{\N}$ items blue.\label{mark:newind1}}
\State{$\allitems(\N) \gets \allitems(\N)+\blue_{\N}+1$, $\reditems(\N) \gets \reditems(\N)+1$\label{mark:newind5}}
\EndIf
\end{algorithmic}
\end{algorithm}

%\paragraph{Properties of the marking} 

%To begin with, l
Line \ref{mark:newind1} of {\MarkItems} ensures the following property, which was the point of postponing the coloring.
Recall that early items are blue $\N$-items which did not get their color immediately and were
packed one per bin (each late item is packed in a bin that already contains an early item).

\begin{aproperty}
	\label{prop:first}
	Each early $\N$-item is at least as large as the red $\N$-item that
	received its mark in the same iteration of \MarkItems.
\end{aproperty}

After all items have arrived and after some post-processing, we will have %that for $\M\in\{\N, \B\}$
\begin{equation}
\label{eq:rightfraction}
|\reditems(\M) - \allitems(\M)\cdot \redfrac_i |  = O(1)\mbox{ for }\M\in\{\N, \B, \R\}\mbox{ and }
\mbox{ each medium type }i.
\end{equation}
%It may not hold exactly for $\R$, but this turns out to be unnecessary in any event.
Each item will be marked according to the set to which it (initially) belongs.
We will see that the values $\blue_\R, \blue_\B$ and $\blue_{\N}$ in {\MarkItems} are calculated in such a way that
$\reditems(\M) = \lfloor \allitems(\M)\cdot \redfrac_i \rfloor$ holds just before any assignment
to $\M\in\{\N, \B, \R\}$.
%(\ref{eq:rightfraction}) is maintained.
%Since we run {\MarkItems} after every single arrival of an item, i
%By Property \ref{prop:redok} and since the marks cover all the cases, this implies that there can be only constantly many unmarked items, which can then be ignored.
%These items can be ignored for the analysis. We prove this formally in the next section.
The proof is straightforward. %, but we need to be precise with the bound on $\blue_\N$ for later.
%In particular, Corollary \ref{cor} will turn out to be important.

Note that {\MarkItems} never changes the values $\reditems$ and $\allitems$.
As we saw, the value $\reditems$ may be inaccurate for some types in any event.
This will be fixed for small types in post-processing, whereas for medium types we will prove (\ref{eq:rightfraction}).
Of course, {\MarkItems} does change values
$\reditems(\M)$ and $\allitems(\M)$ for $\M\in\{\R,\N,\B\}$ in order to record how many
items with each mark there are (and these values \emph{will} be accurate).

%All {\EHarm} algorithms are required to satisfy the following property, which
%we need for the proofs in this section.

\begin{lemma}
\label{lem:marksNR}
Let $\M\in\{\N,\R\}$.
Just before assignments of new items to $\M$ in lines \ref{mark:newr2}-\ref{mark:newr3} or lines \ref{mark:newind1}-\ref{mark:newind5},
for each medium type $i$ such that $\redfrac_i>0$,
we have $\reditems(\M) = \lfloor \redfrac_i \allitems(\M)\rfloor$
and $%1/(2\redfrac_i)-3/2 < 
\blue_\M < 1/(2\redfrac_i)+1/2$. Generally, we have 
$\reditems(\M) \in [\lfloor \redfrac_i\allitems(\M)\rfloor,$ $\lfloor \redfrac_i\allitems(\M)\rfloor+1]$. 
\end{lemma}
\begin{proof}
Call the assignment of new items to $\M$ due to lines \ref{mark:newr2}-\ref{mark:newr3} or lines \ref{mark:newind1}-\ref{mark:newind5} \emph{early assignments}. 
% and assignments to $\M$ due to line \ref{mark:1} \emph{late assignments}.

At the beginning, we have $\reditems(\M)=\allitems(\M)=0$.
Thus the lemma holds at this time.
When an early assignment takes place, 
$\allitems(\M)$ increases by $\blue_\M+1$, and $\reditems(\M)$ by 1.
By minimality of $x_\M$, just before any early assignment we have
\begin{align}
\label{xbound1}
&&\lfloor (\allitems(\M)+ 2(\blue_\M-1)+1)\redfrac_i\rfloor & \le \reditems(\M)\\
\nonumber
\Rightarrow &&
(\allitems(\M)+ 2(\blue_\M-1)+1)\redfrac_i &< \reditems(\M) +1\\
\nonumber
%\label{xbound2}
\Rightarrow &&
(\allitems(\M)+ 2 \blue_\M+1)\redfrac_i &< \reditems(\M) +1 +2\redfrac_i\\
\nonumber
\Rightarrow&&
\lfloor (\allitems(\M)+ 2\blue_\M +1)\redfrac_i\rfloor&\le \reditems(\M)+1\\
\nonumber
\Rightarrow&&
\lfloor (\allitems(\M)+ 2\blue_\M +1)\redfrac_i\rfloor&= \reditems(\M)+1
\end{align}
where we have used Property \ref{prop:redfrac} and integrality in the penultimate line and
the definition of $x_\M$ in the last line.
This immediately implies that right after an early assignment to $\M$,
\begin{equation}
\lfloor(\allitems(\M)+\blue_\M) \redfrac_i\rfloor = \reditems(\M). \label{eq:mark-after-assign}
\end{equation}
There are then $x_\M$ bins with one early blue medium item of type $i$. 
{\EHarm} will put
the next arriving blue items of this type into these $x_\M$ bins (one additional item per bin) before opening any new bins.
All of these late blue items are assigned to $\M$ and $\allitems(\M)$ is increased accordingly in lines \ref{mark:1}--\ref{mark:2},
%In total, this increases $\blueitems(\M)$ by exactly the value of $\blue_\M$ that was used in the last assignment of items to $\M$, 
so eventually $\reditems(\M) = \lfloor \redfrac_i \allitems(\M)\rfloor$.

After that, $\allitems(\M)$ and $\reditems(\M)$ remain unchanged until the next early assignment of items to $\M$. Hence before an early assignment of items to $\M$, the first claimed equality holds.

This equality together with (\ref{xbound1}) 
%(which also holds immediately before initially assigning new items to $\M$) 
gives $\lfloor \allitems(\M)\redfrac_i + (2(\blue_\M-1)+1)\redfrac_i \rfloor \le \lfloor \allitems(\M)\redfrac_i \rfloor$
which implies $(2(\blue_\M-1)+1)\redfrac_i<1$ %by (\ref{xbound1}) 
and thus $x_\M < 1/(2\redfrac_i)+1/2 < 1/\redfrac_i$ since $\redfrac_i<1$.
This, together with (\ref{eq:mark-after-assign}), implies $\reditems(\M)<\lfloor (\allitems(\M) \rfloor +1$ (note that $\reditems(\M)$ is largest relative to $\allitems(\M)$ right after an early assignment to $\M$, i.e., when (\ref{eq:mark-after-assign}) holds).
\iffalse
It furthermore follows from (\ref{xbound2}) that after any early assignment
has taken place and the $x_\M$ bins with one blue item have received their
second blue item, we have 
\begin{equation}
\allitems(\M)\redfrac_i < \reditems(\M) +2\redfrac_i \label{lem1:before-assignment}
\end{equation}
since by this time $\allitems(\M)$ has increased by $2x_\M+1$ and
$\reditems(\M)$ has increased by 1 compared to the situation in (\ref{xbound2}).
This also holds before the first early assignment takes place.
By definition of $\blue_\M$, we have $(\allitems(\M)+2\blue_\M+1)\redfrac_i \ge \reditems(\M)+1$. This, together with (\ref{lem1:before-assignment}), implies
$(2\blue_\M+1)\redfrac_i > 1-2\redfrac_i$ and thus $\blue_\M > 1/(2\redfrac_i)-3/2$.
%The bound on $\blue_\M$ follows from the minimality of $\blue_\M$ which implies that 
%$\redfrac_i(2\blue_\M+1)<2$ by Property \ref{prop:redfrac}. This bound implies the upper bound on $\reditems(\M)$.
\fi 
\qed\end{proof}

\iffalse 
\begin{corollary}
	\label{cor}
	For a medium type $i$ with $\redfrac_i>0$, let $x_i = \max\{a\in\mathbb N| a <  1/(2\redfrac_i)+1/2 \}.$
	Then $\blue_\N\in\{x_i,x_i-1\}$ in all iterations of Algorithm \ref{alg:mark} for type $i$.
\end{corollary}
\begin{proof}
	By Lemma \ref{lem:marksNR}, $\blue_\N$ is strictly contained in an open interval of length 2.
\end{proof}
\fi 

\begin{corollary}
	\label{cor:uncolored}
	After each execution of \MarkItems{} and for each medium type  $i$ such that $\redfrac_i>0$, $\allitems(\U)\le 1/\redfrac_i$.
\end{corollary}
\begin{proof}
	We have $\blue_\N < 1/\redfrac_i$ and $\blue_\R < 1/\redfrac_i$ by Lemma \ref{lem:marksNR}  since $\redfrac_i<1$, so at the latest when
	$1/\redfrac_i+1$ uncolored non-bonus items exist, they are marked and colored.
\qed\end{proof}

\begin{lemma}
\label{lem:markB}
At all times and for each medium type $i$ such that $\redfrac_i>0$, $\reditems(\B) = \lfloor \redfrac_i\allitems(\B)\rfloor$
and $\blue_\B < 1/\redfrac_i$. 
\end{lemma}
\begin{proof}
We use similar calculations to the proof of Lemma \ref{lem:marksNR}.
At the beginning, all counters are zero. When {\MarkItems} is about to assign items to $\B$, 
we have $\lfloor \redfrac_i(\allitems(\B) + \blue_\B +1)\rfloor = \lfloor \redfrac_i(\allitems(\B) + \blue_\B +2)\rfloor =\reditems(\B)+1$
by definition of $x_\B$ and Property \ref{prop:redfrac}.
%since $\redfrac_i<1/2$.
This immediately implies that after each assignment,
%\emph{if} we assign $\blue_\B+1$ items to $\B$ (including the red one),
we have $\lfloor \redfrac_i\allitems(\B) \rfloor = \reditems(\B)$. 
%Since $\blue_\B$ was chosen minimally and by Property \ref{prop:redfrac}, 
%the desired bounds hold even if we assign one additional blue item to 
%$\B$. 
By minimality of $x_\B$, we also conclude $\lfloor \redfrac_i(\allitems(\B) + \blue_\B )\rfloor = \reditems(\B)$, so  $\blue_\B \cdot \redfrac_i < 1$.
\qed\end{proof}

\section{Post-processing}
\label{sec:post}
%After the algorithm has packed all items, we perform some post-processing operations 
%that will make the output more structured. 
Since we consider only the asymptotic competitive ratio in this paper, 
it is sufficient to prove that a certain ratio holds for 
all but a constant number of bins: such bins are counted in the additive constant.
We will perform a sequence (of constant length) of removals of bins in this section.
We will also change the marks of some items to better reflect the actual output, 
fix the type and color of any remaining bonus items and reduce the sizes of some items
to match the values used by {\EHarm} in its accounting (see line \ref{eh:countasi} of Algorithm \ref{alg:eh}).

%We will remove some items,  For each of these three steps, we will argue why we make these changes and why they do not impair the validity of our argumentation. 

%\paragraph{Removing some items}
%\label{sec:removing}
To begin with, we remove the at most $\sum_{i: \redfrac_i>0} 1/\redfrac_i$ bins with unmarked medium items
(Corollary \ref{cor:uncolored}), but not the bonus items.
We also remove (at most $\rTypes-1$) blue-open pure blue bins,
as well as the single bin with items of type $\rTypes$ of total size at most $1-\eps$, if it exists (see Property \ref{prop:nextfit}).
Additionally, we remove any bins with a \emph{single} blue $\N$- or $\R$-item, as well as all bins that were assigned to $\N$ and $\R$ at the same
time as such bins (i.e., during one execution of lines \ref{mark:newblue2} or \ref{mark:newind1} of {\MarkItems}).
This is at most $\sum_{i: \redfrac_i>0} (1+1/\redfrac_i)$ bins
by Lemma \ref{lem:marksNR}.
Overall we have removed at most a constant number of bins so far.
The packing now has the following property.
This and subsequent Packing Properties will continue to hold during post-processing and will be the basis of our proof of the competitive ratio.

\begin{pproperty}
\label{pp:twoblue}
All medium non-bonus items are marked.
Each blue item in $\N,\R$ and $\B$ is packed in a bin that contains two blue items,
and $|\reditems(\M) - \allitems(\M)\cdot \redfrac_i  | = O(1)\mbox{ for }\M\in\{\N, \B, \R\}$ and 
each medium type $i$.
All bins with blue $\B$-items or red $\R$-items are mixed. 
All bins with items of type $N$ are at least $1-\eps$ full.
\end{pproperty}
\begin{proof}
Lines \ref{mark:newblue2} or \ref{mark:newind1} of {\MarkItems} are only executed if all blue items that were
assigned to $\N$ or $\R$ in a previous run of {\MarkItems}
are already packed into bins with two blue items, since 
Algorithm {\Pack} prefers to pack a new blue item into an existing blue-open bin.
Thus, when we remove all bins with single $\N$- or $\R$-items, this is only constantly many bins.
The blue $\B$-items are packed two per bin by the rules of \MarkItems.
The equality then follows from Lemmas \ref{lem:marksNR} and \ref{lem:markB}. 
The penultimate line follows from the way {\MarkItems} selects the items to mark. The last line follows from Property \ref{prop:nextfit}.
\qed\end{proof}

\paragraph{Final marking}
\label{sec:final}

\begin{figure}[t!]
	\begin{center}
		\includegraphics[width=0.8\textwidth]{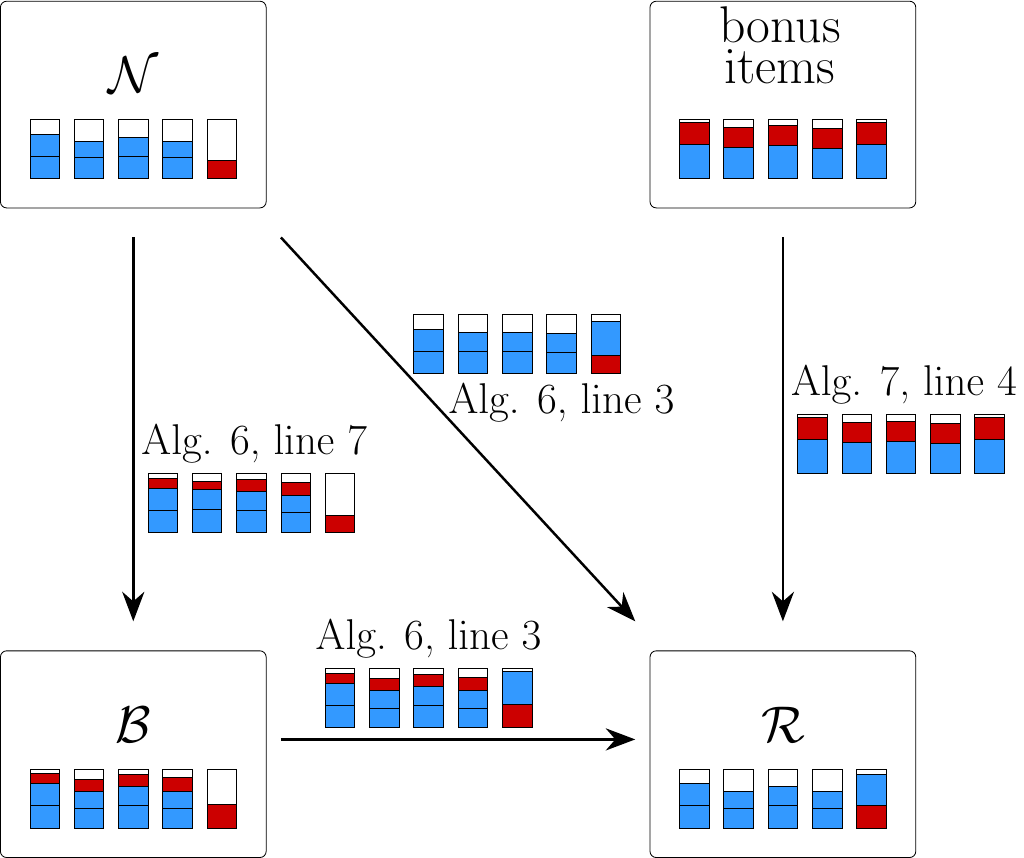}
		\caption{\label{fig:final}Reassigning marks after the input is complete and changing some items to get rid of bonus items.
		Items are sorted into their correct sets whenever possible, updating the marks
		that they received while the algorithm was running. Some item sizes are reduced (!).
		The bins next to the arrows indicate what sets of bins are being reassigned.
		}
	\end{center}
\end{figure}

An overview of our changes of marks and sizes
is given in Fig. \ref{fig:final}. 
%
%Recall that $\U$ is the set of unmarked items.
We will change marks of some items to $\R$ or $\B$ if
such marks are appropriate. To do this, we run Algorithm \ref{alg:final}
for every medium type $i$ separately.
Note that seemingly wrongly marked items like the ones we look for in
Algorithm \ref{alg:final} can indeed exist because while the algorithm is
running we only mark
each item once, when it is assigned to a set; other items could arrive later
and be packed with it, invalidating its mark.
Packing Property \ref{pp:twoblue} is not affected by Algorithm \ref{alg:final},
since we change marks in the correct proportions, and we only add items to $\R$ and $\B$
that satisfy Packing Property \ref{pp:twoblue}. 

Instead of the process described in Algorithm \ref{alg:final}, an easier approach might seem to be
the following. For changing marks from $\N$ to $\R$, we could simply take the group of bins containing 
the $\N$-items that received their mark at the same time as the red $\N$-item in the mixed bin.
The problem with this approach is that not all these groups have the same size in general, since $\blue_\N$ may vary. %(see Corollary \ref{cor:uncolored}).
%Unless $1/\redfrac_i$ is an odd integer, the groups for type $i$ have two different sizes, namely $x_i$ and $x_i-1$, where $x_i = \max\{a\in\mathbb N| a <  1/(2\redfrac_i)+1/2 \}.$
%Viewed in isolation, the larger groups have too few red items (less than a $\redfrac_i$ factor), while the smaller groups have too many. Thus, by transferring groups of items to $\R$
%that got their marks in the same iteration of {\MarkItems}, 
This means the ratio $\redfrac_i$ would possibly not
be maintained for $\R$ (and then also not for $\N$).
%Note however that the difference in size between large and small groups is just 1.

\begin{algorithm}[t]
\caption{\label{alg:final}
Final marking for items of type $i$ in {\EHarm} algorithms. Again we only
consider items of medium type $i$.}
\begin{algorithmic}[1]
\State{Sort the bins with two blue $\N$-items in order of increasing size of the early $\N$-items in these bins.}\label{final:sort}
\For{$\M=\{\N,\B\}$}
\State{Let $T$ be the largest integer value such that there exist}
	\begin{itemize}
		\item $\lfloor \redfrac_i T \rfloor$ red $\M$-items in mixed bins (one per bin) and
		\item $(T - \lfloor \redfrac_i   T \rfloor)/2$ bins with (two) blue $\M$-items (so $(T - \lfloor \redfrac_i   T \rfloor)/2\in\mathbb N$)
	\end{itemize}
\State{Assign the $\lfloor \redfrac_i T \rfloor$ largest red $\M$-items in mixed bins and the blue $\M$-items in the first\par
	$(T - \lfloor \redfrac_i   T \rfloor)/2$ bins in the sorted order to $\R$ \label{final:assign-to-r}}
\EndFor
\State{Let $T$ be the largest integer value such that there exist}
	\begin{itemize}
		\item $\lfloor \redfrac_i T \rfloor$ red $\N$-items 
		\item $(T - \lfloor \redfrac_i   T \rfloor)/2$ mixed bins with (two) blue $\N$-items
	\end{itemize}
	\State{Assign the $\lfloor \redfrac_i T \rfloor$ largest red $\N$-items and the blue $\N$-items in the first
		$(T - \lfloor \redfrac_i   T \rfloor)/2$ mixed bins in the sorted order
	 that were not yet reassigned to $\R$, to $\B$}
\end{algorithmic}
\end{algorithm}

\begin{pproperty}
\label{pp:allesOK}
%Each blue item in $\N,\R$ and $\B$ is packed in a bin that contains two blue items.
No bins with items in $\N$ are mixed. 
No bins with \emph{red} items in $\B$ are mixed.
\end{pproperty}

\begin{lemma}
	\label{lem:pp12}
	After running Algorithm \ref{alg:final}, only constantly many bins need to be removed in order to ensure that Packing Property \ref{pp:allesOK} holds.
	Packing Property (\ref{pp:twoblue}) is maintained.
\end{lemma}

\begin{proof}
Let us fix a medium type $i$.
After the first loop is finished, there can be at most constantly many red $\N$-items and $\B$-items in mixed bins, 
since these sets of items are both colored with the correct proportion of
red items by Packing Property \ref{pp:twoblue} and we move a maximal subset of items with the correct proportion to $\R$. 
After Algorithm \ref{alg:final} completes, there are at most constantly many blue $\N$-items in mixed bins
for the same reason.
We can remove all of these bins at the end if needed. This does not affect Packing Property \ref{pp:twoblue}.
\qed\end{proof}

The following lemma helps us to bound the optimal solution later.

\begin{lemma}
\label{lem:smallestred1}
Let the smallest medium red item of type $i$ in $\N$ be $\smallestred_i$. 
It is packed alone in a bin.
At most %$(\redfrac_i + (1-\redfrac_i)/2) \allitems(\N)-2$ 
$\frac{1-\redfrac_i}2 (\allitems(\R)+\allitems(\N))
+O(1) $
items in $\N$ have size less than $s(\smallestred_i)$.
\end{lemma}
\begin{proof}
%\begin{figure}
%\centering
%\includegraphics[width=.8\textwidth]{image18}
%\caption{Assume the medium items given in this figure are marked $\N$ after the termination of the input (before post-processing). The numbers within the items denote their size. The medium items of group 1 were assigned to $\N$ at the same time, and the medium items of group 2 were assigned to $\N$ at the same time. Say item $\smallestred_i$ as defined in the proof of Lemma \ref{lem:smallestred1} is the red item in group 1. We know that the early blue items in this group are at least as large as $\smallestred_i$ by Property \ref{prop:first}. However, the red item of group 2 and thus also the early blue items in group 2 may be smaller than $\smallestred_i$. In this case, Algorithm \ref{alg:final} makes sure that when the red item of group 2 is moved to $\R$, we select bins with \emph{the smallest early blue items} for $\R$ as well; in this example, group 2 could be moved to $\R$, effectively removing blue early items smaller than $\smallestred_i$ from set $\N$. \label{fig:smallestred}}
%\end{figure}
Item $\smallestred_i$ is packed alone by Packing Property \ref{pp:allesOK}.
Each red $\N$-item of type $i$ has size at least $s(\smallestred_i)$ by definition of $\smallestred_i$.
Furthermore, each \emph{early} blue $\N$-item of type $i$ has size at least $s(\p)$, where
$\p$ is the red $\N$-item that got its mark at the same time (Property \ref{prop:first}).
However, it is possible that the bin containing $\p$ received an additional (large, blue) item later.
In that case, after post-processing, the item $\p$ does not have mark $\N$ anymore,
so it is not considered when determining $\smallestred_i$, and may in fact be smaller than $\smallestred_i$.
In Algorithm \ref{alg:final}, we therefore take care to always select the bins with
the smallest early blue $\N$-items (line \ref{final:sort}). %; see also Figure \ref{fig:smallestred}.

We now give an upper bound for the number of early blue $\N$-items that can be smaller than $\smallestred_i$ but still have mark $\N$ after Algorithm \ref{alg:final} completes.
%We transfer the items from $\N$ to $\R$ in the right ratio of red to blue, and we select the bins with the smallest possible early items.
%
%By Lemma \ref{lem:marksNR}, we have
%$(2(\blue_\N-1)+1)\redfrac_i<1$ in each iteration of {\MarkItems}.
%In particular, for the number $x_i$ defined in Corollary \ref{cor:uncolored}, we have
%$(2(\blue_i-1)+1)\redfrac_i<1$.
%The number of items of type $i$ that are eventually packed in the bins containing
%a set of $\N$-items that got their mark in the same iteration of {\MarkItems} is exactly
%$(2\blue_\N+1)\redfrac_i$.  
%The number of early items in these bins is $\blue_\N$.
Let $z = \lfloor \redfrac_i T \rfloor$ be the number of red $\N$-items in mixed bins
that receive the mark $\R$ in line \ref{final:assign-to-r} of Algorithm \ref{alg:final}.
Then the total number of early items that got their mark $\N$ at the same time as these
$z$ items is upper bounded by $z/(2\redfrac_i) + z/2$ by Lemma \ref{lem:marksNR}.
We transfer in total $(T-z)/2$ early items from $\N$ to $\R$.
The number of early items that do not get transferred and are potentially smaller than $\smallestred_i$
is therefore at most  $z/(2\redfrac_i) + z -T/2 \le z$ since $T\ge z/\redfrac_i$.
Clearly, since we move $z$ red $\N$-items to $\R$,
$z$ is at most $\redfrac_i\allitems(\R)$ afterwards.

Finally, we give an upper bound for the number of late blue $\N$-items.
There are $\allitems(\N) - \reditems(\N) =
\allitems(\N) - \lfloor \redfrac_i \allitems(\N)\rfloor$ blue items in $\N$ (using Lemma \ref{lem:marksNR}).
Half of them are packed into existing bins (i.e., as late items).
We have $\frac12 (\allitems(\N) - \lfloor \redfrac_i \allitems(\N)\rfloor) \le \frac{1-\redfrac_i}2 \allitems(\N) +\frac12$.

Since $\redfrac_i < \frac{1-\redfrac_i}2$ by Property \ref{prop:redfrac}, the lemma follows.
%The bound in the lemma consists of the $\lfloor \redfrac_i \allitems(\N)\rfloor$ red $\N$-items
%and the $\lfloor(1-\redfrac_i)\allitems(\N) / 2\rfloor$ blue items that are first in their bins
%(there are $\allitems(\N) - \lfloor \redfrac_i \allitems(\N)\rfloor 
%\ge (1-\redfrac_i)\allitems(\N)$ blue items in $\N$).
\qed\end{proof}

\paragraph{A modification of the input}
\label{sec:mod}
In line \ref{eh:bonus} of {\EHarm}, bonus items are created. These are medium items which are packed as red items (each such item is in a bin with a large blue item) but violate the ratio $\redfrac_i$. 
Some of them may still be bonus when the algorithm has finished. 
Also, some of them may be
labeled with a different type than the type they belong to according to their size. 
We call such items \emph{reduced} items. Note that {\EHarm} treated each reduced item as small red items in its accounting (but had in fact packed the larger bonus item).
All reduced items are in mixed bins. They are not counted as bonus items.

\begin{algorithm}[ht]
\caption{\label{alg:mod}
Modifying the input after packing all items}
\begin{algorithmic}[1]
\State{Let the number of bonus items of type $i$ be $T$. \Comment{These are not reduced items}}
	\State{Color $\lfloor \frac{2\redfrac_i T}{1+\redfrac_i} \rfloor$ of these items red and the others blue. Mark them all $\R$.}\label{mod:doit-next}
	\State{Reduce the size of blue large items in the bins with (now) blue medium items of type $i$ to $t_i$.}\label{mod:shrink-large}
	\State{Mark all of these items $\R$ as well.\label{mod:mark-r}}
	\State{$\allitems(\R)\gets \allitems(\R)+2T - \lfloor \frac{2\redfrac_i T}{1+\redfrac_i} \rfloor$.}
	\State{$\reditems(\R)\gets \reditems(\R)+\lfloor\frac{2\redfrac_i T}{1+\redfrac_i} \rfloor.$}\label{mod:13}
\ForEach{reduced item $\p$}
	\State{Let $j$ be the type with which $\p$ is labeled.\label{mod:reduced1}}
	\State{Split up $\p$ into $\redfit_j$ red items of size $s(\p)/\redfit_j$.\label{mod:split}}
	\State{Reduce the size of the newly created items until they belong to type $j$.\label{mod:reduced2}}
\EndFor
\end{algorithmic}
\end{algorithm}

After {\EHarm} has finished, and the steps previously described in this sections have been applied,
we modify the packing that it outputs as described in Algorithm \ref{alg:mod}. 
Again we run this algorithm for every medium type $i$.
The post-processing is illustrated in Fig. \ref{fig:post-processing}; the process in lines \ref{mod:doit-next}--\ref{mod:13} is illustrated in Fig. \ref{subfig:postproc-a}, the process in lines \ref{mod:reduced1}--\ref{mod:reduced2} in Fig. \ref{subfig:postproc-c}.
\begin{lemma}
\label{lem:noincrease}
Denote the set of items in a given packing $P$ by $\sigma$.
Denote the set of items after applying Algorithm \ref{alg:mod} to the packing $P$ by $\sigma'$. Then $P$ induces a valid packing for $\sigma'$, and
$
\opt(\sigma')\le \opt(\sigma).
$
\end{lemma}
\begin{proof}
In line \ref{mod:shrink-large} of Algorithm \ref{alg:mod}, items are only made smaller. In line \ref{mod:split}, a medium item of type $i$ is split into $\redfit_j$ items of some type $j$. 
The condition for an item to be labeled with type $j$ in line \ref{eh:fixbonus} of 
{\EHarm} is that $j$ is a small type.

By definition of $\redfit_j$ and $\redspace_{\needs(j)}$, 
we have that $\redfit_j$ items of type $j$ have
total size at most $\redspace_{\needs(j)}$. Since $j$ is a small type, this value is less than $1/3$ by Property \ref{prop:redspace}.
This means the newly created items occupy less space than the medium item that they replace. Hence, in both cases we do not increase the amount of occupied space in any bin.

The inequality follows by choosing $P$ to be an optimal packing for $\sigma$.
\qed\end{proof}

\begin{lemma}
Lemma \ref{lem:smallestred1} still holds after executing Algorithm \ref{alg:mod}.
\end{lemma}
\begin{proof}
Algorithm \ref{alg:mod} only creates new $\R$-items. Therefore, the number of ``problematic items'' that we want to upper bound, that is, the number of $\N$-items of size less than $s(\smallestred_i)$, does not increase. As we only increase $n^i(\R)$ in Algorithm \ref{alg:mod}, the upper bound in Lemma \ref{lem:smallestred1} is not decreased.
\end{proof}

\begin{theorem}
	\label{thm:red-alpha}
	For a given input $\sigma$, denote the result of all the post-processing done in this section by $\sigma'=\{\p_1,\dots,\p_n\}$. 
	Packing Properties \ref{pp:twoblue} and \ref{pp:allesOK} as well as Invariant \ref{inv:unmixed}
	still hold after post-processing.
	For any type $i$, at the end we have $|\reditems - \redfrac_i\allitems|=O(1)$,
	where $\reditems$ counts the (correct) total number of red items of type $i$ after postprocessing.
	There are no bonus items, and  the optimal cost to pack the input did not increase in post-processing.
\end{theorem}
\begin{proof}
	Let $P_0$ be the packing of $\sigma$ that is output by $\A$.
	Let $P_1$ be the packing after running Algorithm \ref{alg:final}, and let the items packed into $P_1$ be $\sigma_1$.
	Packing Properties  \ref{pp:twoblue} and \ref{pp:allesOK} hold for $P_1$  by Lemma \ref{lem:pp12}. Since in Algorithm \ref{alg:mod} we colored the bonus items in the right proportions,
	the ratio $\redfrac_i$ of red items holds for each medium type $i$ up to a constant number of items by Packing Property \ref{pp:twoblue}.
	For a small type $i$, we have $|\reditems - \redfrac_i\allitems|=O(1)$ by Property \ref{prop:redok}.
	The only effect of post-processing is that afterwards, $\reditems$ counts the actual number of red items of type $i$ in $\sigma_1$.
	(Some of these red items replace bonus items in $\sigma_1$, but the algorithm already counted them in the value $\reditems$.)
%	(for small items, Algorithm \ref{alg:mod} may create an additional discrepancy of $\redfit_i$, which is also constant).
	All bonus items were removed, and Packing Properties  \ref{pp:twoblue} and \ref{pp:allesOK} remain unaffected.
	
	Consider some medium type $i$. Invariant \ref{inv:unmixed} is not affected by any change of marks or removal of bins.
	The effect of lines \ref{mod:doit-next}--\ref{mod:13} of Algorithm \ref{alg:mod} is that some bins with bonus items
	of type $i$ are replaced with unmixed bins with blue type $i$ items.
	This does not affect the validity of Invariant \ref{inv:unmixed}.	
	
	To get from $P_0$ to $P_1$, we only removed some bins (and changed marks, which are irrelevant for the optimal solution).
	Hence $\opt(\sigma_1)\le\opt(\sigma)$. We can now apply Lemma \ref{lem:noincrease} to the optimal packing for $\sigma_1$ to get the final claim.
\end{proof}

\section{Weights}
Let $\A$ be an {\EHarm} algorithm.
For analyzing the competitive ratio of $\A$, we will use the well-known technique of weighting functions. The idea of this technique is the following. We assign weights to each item such that the number of bins that our algorithm uses in order to pack a specific input is equal (up to an additive constant) to the sum of the weights of all items in this input. Then, we determine the average weight that can be packed in a bin in the optimal solution. This average weight for a single bin gives us an upper bound on the competitive ratio. In order to use this technique, we now define a set of weighting functions.

Fix an input sequence $\sigma.$
Denote the result of post-processing $\sigma$ by $\sigma'=\{\p_1,\dots,\p_n\}$. 
Let $P$ be the packing of $\sigma$ that is output by $\A$.
Let $P'$ be the packing of $\sigma'$ induced by $P$ (Lemma \ref{lem:noincrease}).

From this point on, our analysis is purely based on the structural properties of the packing $P'$ that we established in Theorem \ref{thm:red-alpha}. We view $\sigma'$ only as a set of items and not as a list. We prove in Theorem \ref{thm:ub-eh} below that this is justified. In particular, we do \emph{not} make any statement about $\A(\sigma')$, since the post-processing done in Algorithm \ref{alg:final} means that some items (e.g., the ones introduced in lines \ref{mod:reduced1}-\ref{mod:reduced2}) do not have clearly defined arrival times, and it is not obvious how to define arrival times for them in order to ensure that $\A(\sigma') = \A(\sigma)$.

The \emph{class}
of an item of type $t$ is $\leaves(t)$, if it is blue, and $\needs(t)$ if it is 
red.
The class of an item $\p$ indicates how much space is reserved for red items in the bin containing $\p$
(both if $\p$ is red and if $\p$ is blue), namely $\redspace_i$ space if the class is $i$.

\begin{lemma}
	\label{lem:allthesame}
	For $k\in\{1,\dots,K\}$, red items of class $k$ are either all medium or all small.
	If they are medium, they are of the unique type $t$ such that $k=\needs(t)$.
\end{lemma}
\begin{proof}
	By Property \ref{prop:redspace}, if for a red item of type $t$ we have
	$\redspace_{\needs(t)}>1/3$, then it is a medium item; in this case, 
	type $t$ is the only type such that
	$\needs(t)=k$ since each medium type is in a different class by Property \ref{prop:tidelta}
	and for each small item $\p$ we have $\redspace_{\needs(t(\p))}\le1/3$.
\end{proof}

%\begin{definition}
%	Let the \emph{class} of an unmixed bin be the class of the items it contains (an index of a space in $\D$, given by the function $\needs$ for unmixed red bins and by the function $\leaves$ for unmixed blue bins).
%\end{definition}
%The class of an unmixed bin is well-defined, as each unmixed bin contains items of only one type.
%In unmixed (red or blue) bins of class $i$, $\redspace_i$ of each bin is reserved for red items.
%(Note that this holds both for unmixed red and for unmixed blue bins.)
%Recall that the class of a bin indicates how much space is reserved for red items in this bin
%().
The class of a bin with red items is the class of those red items.
%The class of an unmixed bin with blue items is the class of those blue items.
This is well-defined, as each bin contains red items of only one type.

\begin{definition}
	\label{def:6}
	Let $k$ be the minimum class of any unmixed red bin.
%	If all red items are in mixed bins, we define $k=K+1$.
%	If $k<K+1$ and $\redspace_k>1/3$,
	Let $\smallestred$ be a smallest item in the unmixed red bins of class $k$. %This item is a medium item.
	If all red items are in mixed bins, we define $k=K+1$ (and $\smallestred$ is left undefined).
\end{definition}

If $k\in\{1,\dots,K\}$, then by this definition we have $k=\needs(t(\smallestred))$. 
If $\redspace_k\le 1/3$, there may be several red items in one bin,
as in {\SuperH} (in {\Hpp}, there is always at most one red item per bin).
Also, there can be several types $t$ such that $k=\needs(t)$.
If $\redspace_k > 1/3$, there is only one type $t$ such that $k=\needs(t)$, and this is a medium type; it is only in this case that we need to consider the item $\smallestred$ and in particular its exact size.

We follow Seiden's proof, adapting it to take the marks into account.
%for the case where $s(\smallestred)>1/3$. 
%For the case {that $\smallestred$ is small} (and the case $k=K+1$), we can apply the analysis by Seiden. (We still get a better result overall because our improvements allow us to set some parameters, in particular values $\redfrac_i$, differently.)
%
In order to define the weight functions, it is convenient to introduce some additional types.
Note that the algorithm does not depend on the weight functions in any way.
It is also unaware of the added type thresholds.
First of all, for each $i$ such that $1/3<\redspace_i<1/2$,
we add a threshold $1-\redspace_i$ between $t_2=2/3$ and $t_3=1/2$
(see Property \ref{prop:type12N}).
For a type $t$ with upper bound $1-\redspace_i$ we define
$\leaves(t)=i$. We furthermore add a threshold $1-s(\smallestred)$ in case
$\smallestred$ is medium.
This splits an existing type into two types.
For the new type $t^1$ with upper bound $1-s(\smallestred)$, we define
$\leaves(t^1) = k$, where $k = \needs(t(\smallestred))$.
For the new type $t^2$ with lower bound $1-s(\smallestred)$, we define
$\leaves(t^2) = k - 1$.
To maintain consistency with the rest of the paper, 
we add negative indices for the types to maintain $t_3=1/2$.
That is, if there are $a$ values in $\D$ in the range $(1/3,1/2)$,
the corresponding values $1-\redspace_i$ and the threshold $1-s(\smallestred)$ 
(if $\smallestred$ is medium)
are stored in ascending order in the values $t_2,t_1,\dots,t_{2-a}$, and $t_{1-a}=2/3$, $t_{-a}=1$.

For large items, the value of $\leaves$ is only used by the algorithm to check whether \emph{small} items can be combined with them.
Moreover, for small items, the only relevant piece of information is that at least $1/3$ of space is left by large items.
An {\EHarm} algorithm defines $\leaves(2)$ such that $\redspace_{\leaves(2)}=1/3$ (and then ignores this value when considering to pack a medium item with a large item). The additional types simply make the function $\leaves$ more accurate, in particular with the threshold $1-s(\smallestred)$, which the algorithm does not know.
It can be seen that the definition of $k$ (and $\smallestred$) is not affected by these new types, as only types of large (i.e., blue) items are changed, and $k$ and $\smallestred$ are defined based on unmixed red bins.

%
%	Relevance??
%
%A medium red item in an unmixed bin can be an $\N$-item or a $\B$-item, 
%but no $\R$-item.
The weights of an item $\p$ will depend on $s(\smallestred)$,
the class of the red and blue items of type $t(\p)$ relative to $k$, and the mark of $\p$.
%Since the weight functions depend on $k$ and on the marks of some items,
This means we essentially define them for every possible input sequence separately.
The value of $k$ and $s(\smallestred)$ (and the marks) become clear by running the algorithm.
We do not write the dependence on $\sigma$ explicitly since we have fixed $\sigma$ in this section.

The two weight functions of an item of size $x$, type $t$ and mark $\M$ are given by Table \ref{tab:weights}. 
Recall that $\eps=t_N$.
Regarding $w_k(\p)$, non-medium items have no mark  and are handled under the case $\M\not=R$.
(Unmarked \emph{medium} items were removed in the previous section). 
Note that $w_k(\p)$ does not depend on $s(\smallestred)$ or the added types, as $\redfrac_t=0$ and $\needs(t)=0$ for all items larger than $1/2$. In contrast, $v_{k,s(\smallestred)}(\p)$ depends on $s(\smallestred)$, as the value of $\leaves(t)$ changes at the threshold
$1-s(\smallestred)$ if $\smallestred$ is medium as described above.
%We say that these are the two weight functions \emph{of class $k$}.

\begin{table}[h]
	\caption{Weighting functions of class $k$ for an item $\p$ of size $x$, type $t$ and mark $\M$.}
	\label{tab:weights}
	\centering
	\begin{tabular}{ll|rl}
		\multicolumn{2}{c|} {
			$w_k(\p)=w_k(x,t,\M)$} & 
		\multicolumn{2}{c} {
			$v_{k,s(\smallestred)}(\p)=v_{k,s(\smallestred)}(x,t)$}\\
		\hline
		$\frac{1-\redfrac_t}{\bluefit_t} + \frac{\redfrac_t}{\redfit_t}$ & if  $t<\rTypes$, $\needs(t)> k$
		& $\frac{1-\redfrac_t}{\bluefit_t} + \frac{\redfrac_t}{\redfit_t}$ & if $t<\rTypes$, $\leaves(t)< k$\\
		& or $\needs(t)=0$&&\\
		$\frac{1-\redfrac_t}{\bluefit_t} + \frac{\redfrac_t}{\redfit_t}$ & if $t<\rTypes$, $\needs(t)= k, \M\not=\R$
		& $\frac{\redfrac_t}{\redfit_t}$     & if $t<\rTypes$, 		$\leaves(t)\ge k$\\
		$\frac{1-\redfrac_t}{\bluefit_t}$   % !!!!!!!!!!!!!!
		                                   & if $t<\rTypes$,		$\needs(t)= k, \M=\R$             
		&  $\frac{1}{1-\eps}x$               & if $t=\rTypes$ \\
		$\frac{1-\redfrac_t}{\bluefit_t}$    & if $t<\rTypes$, $0<\needs(t)< k$\\
		$\frac{1}{1-\eps}x$ & if $t=\rTypes$ &
	\end{tabular}
\end{table}

Note that $w_k$ counts all blue items, and $v_{k,s(\smallestred)}$ counts all red items. By definition of $k$ and Packing Property \ref{pp:allesOK}, we have $\M=\R$ for all items with type $t$ such that $\needs(t) < k$.
For simplicity, we ignore the markings for any 
type $t$ with $\needs(t) > k$, essentially assuming that there are no items of
such types that are marked $\R$. It is clear that this assumption can
only increase the weight of any item. %We can furthermore observe that 
%large items have weight 1 under weighting function $w_k$ and weight $1$ or $0$ under weighting function $v_k$ as $\redfrac_t=0$ and $\bluefit_t=1$ for these items.

Define $v_k(\p) = v_{k,t_{t(\smallestred)}}(\p)$. Note that for any item $\p$, we have $v_k(\p) \ge v_{k,s(\smallestred)}(\p)$ since
$1-s(\smallestred) \ge 1-t_{t(\smallestred)}$, and this is the point at which the $\leaves$ function drops below $k$.
%Our claim is the following.

\begin{theorem}
	\label{thm:ub-eh}
	For any input $\sigma$ and {\EHarm} algorithm $\A$, defining $k$ as above we have
	\begin{align}	
%		\label{eq:ub1}	
%		\A(\sigma) &\le \sum_{i=1}^\rItems w_{K+1}(\p_i) +O(1) && \mbox{if }k=K+1\\
		\label{eq:ub-eh}
		\A(\sigma) &\le \min\left\{ 
		\sum_{i=1}^\rItems w_k(\p_i),\sum_{i=1}^\rItems v_{k}(\p_i)
		\right\} + O(1) %&& \mbox{if }k\in\{ 1,\dots,K \}		 
	\end{align}
\end{theorem}
\begin{proof}
	Our goal is to upper bound $\A(\sigma)$ by the weights of the items $\p_1,\dots,\p_n$, which
	are the items in $\sigma'$. 
	We will show that the number of bins in the packing $P'$ is upper bounded by the first term in 
	(\ref{eq:ub-eh}), with the additive constant $O(1)$ corresponding to the bins removed in post-processing.
	We follow the line of the corresponding proof in Seiden~\cite{Seiden02}.
	
	Let $\Tiny$ be the total size of the items of type $\rTypes$ in $\sigma'$.
	Let $\UnmixedRed$ be the number of unmixed red bins in $P'$.
	Let $B_i$ and $R_i$ be the number of bins in $P'$
	containing blue items of class $i$ and type less than $N$, and red items of class $i$, respectively. 
	Note that this means that mixed bins are counted twice.
	
	If $\UnmixedRed = 0$, every red item is placed in a bin with one or more blue items, and $k=K+1$. In this case,
	the total number of bins in $P'$ is exactly the total number of bins containing blue items.
	Each bin containing items of type $\rTypes$ contains at least a total size of $1-\eps$
	due to Packing Property \ref{pp:twoblue}.
	The bins used to pack $\Tiny$ are pure blue and $\sum_{t(\p_i)=N} w_{K+1}(\p_i)=
	\sum_{t(\p_i)=N} v_{K+1}(\p_i) \ge  
	\Tiny/(1-\eps)$. 
	%In the summation $\sum_{i=0}^K B_i$, we count all the blue items and all the red items with types less than $\rTypes$.
	For each item $\p$ 	of type $t<\rTypes$, we have
	$w_{K+1}(\p) = \frac{1-\redfrac_t}{\bluefit_t} < v_{K+1}(\p)$. We see that $w_{K+1}$ counts all the bins with blue items, and
	$	\A(\sigma)\le \frac\Tiny{1-\eps} + \sum_{i=0}^K B_i\le \sum_{i=1}^\rItems w_{K+1}(\p_i)$
	(since $B_0$ does not include bins with items of type $N$).
	%which is at most %upper bounded by 
	%\[
	%$\frac1{1-\eps}\Tiny + \sum_{i=0}^K B_i.$
	%\]
	
	%This holds because each  and , and (\ref{eq:ub-eh}) follows.
	
	If $\UnmixedRed>0$, then $k=\needs(t(\smallestred))$, and there is an unmixed red bin of class $k$.
	By %Lemma \ref{lem:invar}, 
	Invariant \ref{inv:unmixed}, % still holds. 
	%
	% 	Diese Bemerkung scheint komplett irrelevant zu sein.
	%
	%, and there are no bonus items left. 
	%Therefore,
	all bins with a blue item of class $i \ge k$ must be \emph{mixed} bins.
	These are the bins which contain blue items of any type $j$ such that $\leaves(j) \ge k$;
	if $\smallestred$ is medium, this means exactly the large items with size at most $1-s(\smallestred)$.
	We conclude
	\begin{equation}
	\label{eq:1}
	\UnmixedRed \le \sum_{i=1}^{K}R_i - \sum_{i=k}^{K}B_i.
	\end{equation}
	Let $R_k(-\R)$ be the number of bins in $P'$ containing 
	red items of class $k$ that are not marked $\R$.
	If items of class $k$ are not medium, then $R_k(-\R)=R_k$.
	This is a well-defined condition by Lemma \ref{lem:allthesame}.
	Let $R_i^*$ be the number of \emph{unmixed} bins in $P'$ containing red items of class $i$.
	Since every red item with class less than $k$ (that is, red items of any type $j$ such that $\needs(j)<k$)
	is placed in a mixed bin by definition of $k$, we have
	\begin{equation}
	\label{eq:2}
	\UnmixedRed \le \sum_{i=k+1}^{K} R_i^* + R_k(-\R)
	\le \sum_{i=k+1}^{K} R_i + R_k(-\R).
	\end{equation}
	The first inequality holds because the red items marked $\R$ are in {mixed} bins by Packing Property \ref{pp:twoblue}.	
	(If $\smallestred$ is not medium, $R_k(-\R)=R_k$, so it also holds.)
	By combining (\ref{eq:1}) and (\ref{eq:2}), we have
%	\begin{equation}
%	\label{eq:3}
$	\UnmixedRed \le \min\left\{\sum_{i=k+1}^{K} R_i + R_k(-\R), 
	\sum_{i=1}^K R_i - \sum_{i=k}^K B_i
	\right\}.
$%	\end{equation}
	So if $\UnmixedRed > 0$,	the total number of bins in $P'$ is at most
	\begin{align}
	\nonumber
	& \quad \frac1{1-\eps} \Tiny + \UnmixedRed + \sum_{i=0}^K B_i + O(1)\\
	\label{eq:min}
	\le & \quad B_0 + \frac1{1-\eps} \Tiny %\nonumber \\& \quad\quad 
	+ 
	\min\left\{
	\sum_{i=k+1}^K R_i + R_k(-\R) + \sum_{i=1}^K B_i, 
	\sum_{i=1}^K R_i + \sum_{i=1}^{k-1} B_i
	\right\} 
	+ O(1).
	\end{align}
	Let $J$ be the set of types whose blue items are packed in pure blue bins, including type 1 and type $\rTypes$. 
	For each item $\p$ of type $t\not=\rTypes$, $t\in J$, we have $\leaves(t)=0<k$, 
	so $v_{k,s(\smallestred)}(\p) = \frac{1-\redfrac_t}{\bluefit_t}$. Furthermore, for all $t\neq N$ we have $w_k(\p) \ge \frac{1-\redfrac_t}{\bluefit_t}$.
	We conclude
	$ \sum_{j\in J} \sum_{t(\p_i)=j} w_k(\p_i) \ge
	\sum_{j\in J} \sum_{t(\p_i)=j} v_{k,s(\smallestred)}(\p_i) \ge B_0+\frac{\Tiny}{1-\eps}\ .
	$
	
	In the first term of the minimum in (\ref{eq:min}), we count all bins with blue items
	except the pure blue bins, all bins with red items of classes
	above $k$, and the bins with red items of class $k$ that are not marked $\R$.
	(If red items of class $k$ are small, this means all red items of this class.)
	%We conclude 
	%$ \sum_{i=k+1}^K R_i + R_k(-\R) + \sum_{i=1}^K B_i
	%\le 
	This term is therefore upper bounded by 
	$\sum_{j\notin J} \sum_{t(\p_i)=j} w_k(\p_i).
	$
	In the second term of the minimum in (\ref{eq:min}), we count all bins with red items, as well
	as bins with blue items of class at least 1 and at most $k-1$.
	%We clearly have
	%$ \sum_{i=1}^K R_i + \sum_{i=1}^{k-1} B_i 
	%\le
	The second term is therefore  upper bounded by  
	$\sum_{j\notin J} \sum_{t(\p_i)=j} v_{k,s(\smallestred)}(\p_i).
	$
	As noted above Theorem \ref{thm:ub-eh}, this is at most $\sum_{j\notin J} \sum_{t(\p_i)=j} v_{k}(\p_i).
	$
	\qed\end{proof}

\paragraph{Note} In his proof, Seiden~\cite{Seiden02} defines an item $e$ as the smallest red item in an indeterminate red group bin, and proceeds to argue using the class of $e$. This only works because there is one red item in each bin, so there could not be a larger red item of a smaller class that is in an indeterminate group bin. The proof structure above (defining first $k$ and then $\smallestred$) allows {\SuperH} algorithms to pack multiple red items in one bin as well.

Seiden expresses the upper bound as a maximum over $k$, even though for 
a fixed input sequence, the value of $k$ is of course fixed. 
While the resulting expression is correct, we prefer the easier and more direct formulation in Theorem \ref{thm:ub-eh} above.
%As noted, for small $\smallestred$, this Theorem holds for {\SuperH} as well
%(and it also holds for medium $\smallestred$ using the simpler weighting functions of Seiden, ignoring the marks).
%
\iffalse However, this may have been introduced 
to properly cover the case $E=0$: the inequalities proved for $E$ do not hold if $E=0$,
as can be seen for an output consisting mostly of blue items (the first inequality requires that
there are more bins with red items than with blue items in this case, which does not have to hold)
\fi 
%
%For {\Hpp}, the analysis works out because there is only one red item in each bin and the class happens to be 
%monotone in the size, 
%but this is not necessarily the case for all {\SuperH} algorithms.
%
%Finally, it can be seen that in Case 2 of the proof, 
%our weighting functions match those in Seiden~\cite{Seiden02}.

\section{The offline solution} \label{sec:offline-solution}
Having derived an upper bound for the total cost of an \EHarm{} algorithm
in Theorem \ref{thm:ub-eh}, in order to calculate the asymptotic competitive
ratio (\ref{eq:apr}), we now need to lower bound the optimal cost of a given input
after post-processing. This will again depend on what the value of $k$ is.
There are two main cases if $k\in\{1,\dots,K\}$: {$\smallestred$ is medium and $\smallestred$ is small}.
The case $k=K+1$ is much easier, because $w_{K+1}(\p)\le v_{K+1}(\p)$ for each item $\p$,
so $\sum_{i=1}^n w_{K+1}(\p)$
upper bounds the cost of $\A$ by Theorem \ref{thm:ub-eh}, 
and this sum does not depend on any marks of items.
We can therefore use a standard knapsack search
as in Seiden~\cite{Seiden02} (for this case) and other papers.

For $k\in\{1,\dots,K\}$, we will be interested in the weights of items for a fixed value of $k$.
It can be seen that in the range $(1/2,1]$, the function $v_{k}(\p)$ changes at most once
(viewed as a function of the size of $\p$), namely at the threshold $1-t_{t(\smallestred)}$,
where $\leaves(k)$ drops below $k$ if $\smallestred$ is medium.
On the other hand $w_k(\p)=1$ in the entire range $(1/2,1]$.
For a fixed value of $k<K+1$, we therefore reduce the number of types again as follows. 
Recall that $t_3=1/2$, and $\smallestred$ is determined by $k$.

\paragraph{Case 1: $\smallestred$ is medium}
We set $t_2=1-t_{t(\smallestred)}$, $t_1=2/3$ and $t_0=1$. We set $\leaves(2)=k$, 
$\leaves(1)<k$ such that $\redspace_{\leaves(1)}=1/3 < t_{t(\smallestred)}$, and $\leaves(0)=0$.

\paragraph{Case 2: $\smallestred$ is small}
We set $t_2=2/3$ and $t_1=1$ as in {\EHarm} itself (Property \ref{prop:type12N}).
We have $\redspace_{\leaves(2)}=1/3$, and $\leaves(1)=0$.

\vspace{10pt}
After these changes, Theorem \ref{thm:ub-eh} remains valid for any fixed $k$, as $w_k$ and $v_k$ remain
unchanged (given $k$).
This holds even though if $\smallestred$ is medium, the types do not match the types used by {\EHarm};
the important property is that they match the behavior of {\EHarm} for any fixed value of $k<K+1$.
%We will see later that items of type 0 may safely be ignored, so we can assume that the highest
%type threshold is $t_1$ as in the original definition.

We now define patterns for the two main cases.
Intuitively, a pattern describes the contents of a bin in the optimal offline solution. 
If $\smallestred$ is medium, a \emph{pattern of class k} is an integer tuple $q=(q_0,q_1,\dots,q_{\rTypes-1},(q_{t(\smallestred)}^\N,q_{t(\smallestred)}^\B,q_{t(\smallestred)}^\R))$ where $q_i\in \mathbb{N}\cup\{0\}, q_{t(\smallestred)}^\M \in \mathbb{N}\cup\{0\}$ for $\M \in \{\N,\B,\R\}$, $q_{t(\smallestred)}^\N+q_{t(\smallestred)}^\B+q_{t(\smallestred)}^\R = q_{t(\smallestred)}$
and
\begin{equation}
	\label{eq:fits}
	\sum_{i=0}^{\rTypes-1} q_it_{i+1}<1.
\end{equation}
The values $q_i$ describe how many items of type $i$ are present in the bin.
The value $q_{t(\smallestred)}^\M$ counts the number of items of type $t(\smallestred)$ and mark $\M$. 
It can be seen that any feasible packing of a bin can be described by a pattern: the only quantity that is not 
fixed by a pattern is the total size of the items of type $N$, which we will call sand.
However, by (\ref{eq:fits}), there can be at most $1-\sum_{i=0}^{\rTypes-1} q_it_{i+1}$ of sand in a bin
packed according to pattern $q$. Conversely, for each pattern, a set of items matching the pattern
that fit into a bin can be found by choosing the size of each item close enough (from above) to the
lower bound $t_{i+1}$ for its type; then (\ref{eq:fits}) guarantees the items will fit.

% (the remaining space is not necessarily in the range $(0,t_\rTypes]$).
If $\smallestred$ is small, we define a pattern of class $k$ as an integer tuple
$q=(q_1,\dots,q_{\rTypes-1})$ where $q_i \in \mathbb{N}\cup\{0\}$ and (\ref{eq:fits}) holds
using $q_0=0$
(note that the values $t_1$ and $t_2$ depend on whether $\smallestred$ is medium or small,
but the definition of $t(\smallestred)$ is consistent across these two cases).

There are only finitely many patterns for each value of $k$. 
Denote this set by $\Q_{k}$ for $k=1,\dots,K$.
If $\smallestred$ is small or $k=K+1$, $\Q_{k}$ is a fixed set, denoted by $\Q$.

For a given weight function $w$ of class $k$, 
we define $w(q)$ for some pattern $q$ as the sum of the weights of the non-sand items in it plus
$w(1-\sum_{i=0}^{\rTypes-1} q_it_{i+1},\rTypes,\emptyset)$. As noted, $1-\sum_{i=0}^{\rTypes-1} q_it_{i+1}$
is an upper bound for the amount of sand in a bin packed according to pattern $q$; this 
value is not necessarily in the range $(0,t_\rTypes]$. If $\smallestred$ is medium, $q_0=0$.
Pattern $q$ specifies all the information we need to
calculate $w(q)$, as $w$ does not depend on the precise size of non-sand items, and for class $k$
we know exactly how many items there are (if any) for each mark.

\iffalse 
following terms.
\begin{itemize}
	\item
	For $i=1,\dots,\rTypes-1, i\not=t(\smallestred)$, we have a term $q_i w(t_{i},i,\N)$ (the mark is irrelevant here:
	either $\smallestred$ is medium and then $t(\smallestred)$ is the unique type such that $\needs(t(\smallestred))=k$,
	or $\smallestred$ is not medium and then items of class $k$ have no mark).
	\item 
	For $i=t(\smallestred)$, if $s(\smallestred)>1/3$, we have a term $(q_{t(\smallestred)}^\B+q_{t(\smallestred)}^\N))w(t_{t(\smallestred)},t(\smallestred),\N) + q_{t(\smallestred)}^\R w(t_{t(\smallestred)},t(\smallestred),\R)$, else we have a term
	$q_{t(\smallestred)} w(t_{t(\smallestred)},t(\smallestred),\emptyset)$ (these items are not marked in this case)
	\item 
	For $i=\rTypes$, there is a term $w(1-\sum_{i=1}^{\rTypes-1} q_it_{i+1},\rTypes,\emptyset)$ (these items are not marked; the first argument is an upper bound for the volume left for sand. This 
\end{itemize}
\fi 

We can describe the solution of an offline algorithm for a given post-processed input $\sigma'$ by a distribution $\chi$ over
the patterns, where $\chi(q)$ indicates which fraction of the bins 
in the optimal solution are packed using pattern $q$. 
Theorem \ref{thm:red-alpha} shows that $\opt(\sigma')\le\opt(\sigma)$, where $\sigma$ refers to the original input and $\sigma'$ refers to the input after post-processing.

To show that \EHarm{} has competitive ratio at most $c$ for an input sequence
$\sigma$ with
{a particular value } $k<K+1$, by Theorem \ref{thm:ub-eh} it is sufficient to show that
\begin{align}
\frac{\min\left\{ 
	\sum_{i=1}^\rItems w_k(\p_i),\sum_{i=1}^\rItems v_k(\p_i)
	\right\}}{\opt(\sigma')} &= 
\min\left\{ 
\frac{\sum_{i=1}^\rItems w_k(\p_i)}{\opt(\sigma')},\frac{\sum_{i=1}^\rItems v_k(\p_i)}{\opt(\sigma')}
\right\}\\
& \le \min\left\{ 
{\sum_{q\in \Q_k} \chi(q) w_k(q)},{\sum_{q\in \Q_k} \chi(q) v_k(q)}
\right\}\le c
\label{eq:compratio}
\end{align}
for all such inputs $\sigma$, using that
$\sum_{i=1}^\rItems w(\p_i) \le {\opt(\sigma')}\sum_{q\in \Q_k} \chi(q) w(q)$ for $w\in\{w_k,v_k\}$,
as $w(q)$ uses an upper bound for the amount of sand but is otherwise exactly the sum of the weights
of the items in it.
%Here we have used that $w_k$ and $v_k$ do not depend on the size of $\smallestred$ but only on $k$.

As can be seen from this bound, the question now becomes: what is the distribution $\chi$ (the mix of patterns) that maximizes the minimum in (\ref{eq:compratio})? We begin by deriving some crucial constraints on $\chi$ for the important case that $\smallestred$ is medium.
This is the point where we start using the marks.
%
%Let $n^{i}(q)$ be the number of items %NOT just the N-items!
%of type $i$ in pattern $q$ for $i=1,\ldots,N-1$. 
%Let $n^{t(\smallestred)}_\M(q)$ be the number of items of type $t(\smallestred)$ and mark $\M$ in pattern $q$ and let $n^{t(\smallestred)}_{-\M}(q)$ be the number of items of type $t(\smallestred)$ and \emph{not} mark $\M$ in pattern $q$ ($\M \in \{\N,\B,\R\}$).
The notation $q_i(q)$ refers to entry $q_i$ in pattern $q$.
We use $q_{t(\smallestred)}^{-\B}(q)$ as shorthand for $q_{t(\smallestred)}^{\R}(q)+q_{t(\smallestred)}^{\N}(q)$.

We define three important patterns $q^1,q^2,q^3$.
For $i=1,2,3$, let $$q^i = (0,1,0,\ldots,0,1,0,\ldots,0,(e_i)),$$ where the second 1 is at position $t(\smallestred)$, and
$e_i$ is the $i$-th three-dimensional unit vector. 
These are the three possible patterns with an item of type $t(\smallestred)$ and an item larger than $1-s(\smallestred)$.
By Property \ref{prop:q1unique}, 
no non-sand item can be added to any of these patterns while maintaining $\sum_{i=0}^{\rTypes-1} q_it_{i+1}<1$.

%\paragraph{Critical bins with $\N$-items}
%Assume  $\smallestred$ is medium. 
%Let $q^1 = (0,1,0,\ldots,0,1,0,\ldots,0,(1,0,0))$, where the second 1 is at position $t(\smallestred)$.
%That is, $q^1$ contains an $\N$-item of type $t(\smallestred)$
%and an item 
\begin{lemma}
	\label{lem:q1fv}
	If {$\smallestred$ is medium},
	then %$n_{-\B}(q) \in \{0,1,2\}$ for all $q$, and 
	\[\chi(q^1) \le \frac{1-\redfrac_{t(\smallestred)}}{1+\redfrac_{t(\smallestred)}}\sum_{q\neq q^1} \chi(q) q_{t(\smallestred)}^{-\B}(q)\,.\]
\end{lemma}
\begin{proof}
	We ignore additive constants in this proof, as we will divide by 
	$\opt(\sigma')$ at the end to achieve our result.	
	The pattern $q^1$ contains an $\N$-item that is strictly smaller than
	$\smallestred$. 
	We apply Lemma \ref{lem:smallestred1} for $i=t(\smallestred)$ (ignoring the additive constant) to get
	\begin{align*}
	\chi(q^1)\opt(\sigma') &\le \frac{1-\redfrac_{t(\smallestred)}}{2} (n^{t(\smallestred)}(\R) + n^{t(\smallestred)}(\N))	\\
	&
	\le \frac{1-\redfrac_{t(\smallestred)}}{2} \left(\chi(q^1) + \sum_{q\neq q^1} \chi(q)
	q_{t(\smallestred)}^{-\B}(q)
	%n^{t(\smallestred)}_{-\B}(q)
	\right)\opt(\sigma'),
	\end{align*}
	and the bound in the lemma follows.
	\qed\end{proof}

%\paragraph{Critical bins with $\B$-items}
%Assume $\smallestred$ is medium. Let $q^2 = (0,1,0,\ldots,0,1,0,\ldots,0,(0,1,0))$, where the second 1 is at position $t(\smallestred)$. 
%That is, $q^2$ contains a $\B$-item of type $t(\smallestred)$ and an item larger than $1-s(\smallestred)$.
%As above, by Property \ref{prop:q1unique} no non-sand item can be added to $q^2$.

\begin{lemma}
	\label{lem:itsblue}
	In $q^2$, the $\B$-item $\p$ of type $t(\smallestred)$ is blue.
\end{lemma}
\begin{proof}
	{\EHarm} did not pack $\p$ alone in a bin as a red item, since it is smaller than $\smallestred$.
	But by Packing Property \ref{pp:allesOK}, $\p$ also was not packed in a mixed bin as a red $\B$-item.
\end{proof}

\begin{lemma}
	\label{lem:q2fv}
	If $\smallestred$ is medium, then
	\[ \frac1%{1-\redfrac_{t(\smallestred)}}
	{2}\chi(q^2) \le 
	\sum_{j:0<\needs(j)\le \leaves({t(\smallestred)})} \sum_{q} \frac{\redfrac_j}{\redfit_j}\chi(q)q_j(q)\ .\] 
\end{lemma}
\begin{proof}
	Again, we ignore additive constants.
	There are $\chi(q^2)\opt(\sigma')$ bins packed with pattern $q^2$, meaning that
 	$\sigma'$ contains at least $\chi(q^2)\opt(\sigma')$ blue $\B$-items of type $t(\smallestred)$
	by Lemma \ref{lem:itsblue}. So in the packing $P'$, there exist
	at least $\frac1%{1-\redfrac_{t(\smallestred)}}
	{2}\chi(q^2)\opt(\sigma')$ bins with two blue $\B$-items of type $t(\smallestred)$
	and red items. The red items are red-compatible with those $\B$-items.
	That is, each such red item is of a type $j$ such that $0<\needs(j)\le \leaves(t(\smallestred)).$
	
	The number of items of type $j$ in $\sigma'$ is given by
	$\sum_q \chi(q)q_j(q)\cdot\opt(\sigma')$.
	By Theorem \ref{thm:red-alpha}, the number of red items of type $j$ is $\redfrac_j\sum_q \chi(q)q_j(q)\cdot\opt(\sigma')$.
	We place $\redfit_j$ red items together in each bin.
	This means that the number of bins in $P$ with red items of type $j$
	is $\frac{\redfrac_j}{\redfit_j}\sum_q \chi(q)q_j(q)\cdot\opt(\sigma')$.
	Summing over all types $j$ with $0<\needs(j)\le \leaves(t(\smallestred))$, we find that
	\begin{align*}
	\frac1%{1-\redfrac_{t(\smallestred)}}
	{2}\chi(q^2)\opt(\sigma')
	&\le 
	(\mbox{number of bins in $P'$ with two blue $\B$-items of type } t(\smallestred) %\\ &\quad\quad\quad\quad\quad\quad 
	\mbox{and red items})\\
	&\le
	(\mbox{number of bins in $P'$ with red items that fit with items} %\\	&\quad\quad\quad\quad\quad\quad 
	\mbox{ of type } t(\smallestred))\\
	&= \left(\sum_{j:0<\needs(j)\le \leaves(t(\smallestred))} 
	\sum_q \frac{\redfrac_j}{\redfit_j}\chi(q)q_j(q)\right)\cdot\opt(\sigma').
	\end{align*}
	\qed\end{proof}

\subsection{Linear program}

Maximizing the minimum in (\ref{eq:compratio}) is the same as maximizing the first
term under the condition that it is not larger than the second term---except that
this condition might not be satisfiable, in which case we need to maximize the
second term. For each value of $k \in \{1,\ldots,K\}$, we will therefore consider two linear programs, and furthermore these linear programs will differ depending on whether $\smallestred$ is medium or small, so that in total we get four different LPs which we will call
$P_w^{k,\text{med}}$,$P_w^{k,\text{sml}}$, $P_v^{k,\text{med}}$ and $P_v^{k,\text{sml}}$ (we will use the notation $P_w^{k}$ ($P_v^k$) whenever we want to refer to both $P_w^{k,\text{med}}$ and $P_w^{k,\text{sml}}$ ($P_v^{k,\text{med}}$ and $P_v^{k,\text{sml}}$)). Let $\Q_k=\{q^1,\dots,q^{|\Q_k|}\}$
and let $\chi_i = \chi(q^i), w_{ik}=w_k(q^i), v_{ik}=v_k(q^i), n_{ij} = q_j(q^i), m_i = q_{t(\smallestred)}^{-\B}(q^i)$. 
If $\smallestred$ is medium, $P^{k,\text{med}}_w$ is the following linear program.

\LPblocktag{$P^{k,\text{med}}_w$}{\label{linprog:lp1}}%
\begin{minipage}{\linewidth-2cm}
	\begin{flalign}
	& \max & \sum_{i=1}^{|\Q_k|} \chi_i w_{ik} \\
	& \text{s.t.} & \chi_1 - \frac{1-\redfrac_{t(\smallestred)}}{1+\redfrac_{t(\smallestred)}}\sum_{i=3}^{|\Q_k|} \chi_im_i   & \leq  0                 &&  && \label{c1:q1fv}\\
	&             & \frac1%{1-\redfrac_{t(\smallestred)}}
	{2}\chi_2 - 
	\sum_{j:0<\needs(j)\le \leaves(t(\smallestred))} \sum_{i=3}^{|\Q_k|}  \frac{\redfrac_j}{\redfit_j}\chi_in_{ij} & \leq 0 &&  && \label{c1:q2fv}\\
	& & \sum_{i=3}^{|\Q_k|} \chi_i\left(w_{ik}-v_{ik}\right) & \le 0 && && \label{c1:wvfv}\\
	& & \sum_{i=1}^{|\Q_k|} \chi_i & \le 1 && && \label{c1:chifv}\\
	& & \chi(q) & \ge 0 && \forall q\in\Q_k && \label{c1:chi0fv}
	\end{flalign}~
\end{minipage}
%\begin{align}
%\max &&\sum_{i=1}^{|\Q_k|} \chi_i w_{ik} \\
%%\label{c1:q1fv}
%s.t.&& \chi_1 - \frac{1-\redfrac_{t(\smallestred)}}{1+\redfrac_{t(\smallestred)}}\sum_{i=2}^{|\Q_k|} \chi_im_i&\le0 \\
%%\label{c1:q2fv}
%&&\frac{1-\redfrac_{t(\smallestred)}}{2}\chi_2 - 
%\sum_{j:\needs(j)\le \leaves(t(\smallestred))} \sum_{i=3}^{|\Q_k|}  \frac{\redfrac_j}{\redfit_j}\chi_in_{ij}&\le0\\
%%\label{c1:wvfv}
%&&
%\sum_{i=3}^{|\Q_k|} \chi_i\left(w_{ik}-v_{ik}\right)&\le 0 \\
%%\label{c1:chifv}
%&&\sum_{i=1}^{|\Q_k|} \chi_i&\le1\\
%%\label{c1:chi0fv}
%&&\chi(q)&\ge0\ & \forall q\in\Q_k
%\end{align}
$P^{k,\text{med}}_w$ has a very large number of variables
but only four constraints (apart from the nonnegativity constraints).
Constraint (\ref{c1:q1fv}) is based on Lemma \ref{lem:q1fv}, where we have used that
$q^2$ does not contain any item marked $\N$ or $\R$, implying $m_2=0$. Constraint
(\ref{c1:q2fv}) is based on Lemma \ref{lem:q2fv}, using that $q^1$ and $q^2$ do not
contain non-sand items of size less than $1/3$, so $n_{1j}=0$ and $n_{2j}=0$ for all $j$ for which
$\needs(j)\le \leaves(t(\smallestred))$.\footnote{We also have $n_{3j}=0$, but we keep the term for $i=3$ in (\ref{c1:q2fv}) to make the dual easier to write down.} Constraint (\ref{c1:wvfv}) says
simply that the objective function must be at most $\sum_{i=1}^{|\Q_k|} \chi_i v_{ik}$ (using that $w_{ik}=v_{ik}$ for $i=1,2$, which we will prove in Lemma \ref{lem:10}):
if this does not hold, we should be solving the linear program $P_v^{k,\text{med}}$, which has
objective function $\sum_{i=1}^{|\Q_k|} \chi_i v_{ik}$, instead. The final constraints (\ref{c1:chifv}) 
and (\ref{c1:chi0fv}) say that $\chi$ is a distribution.

\begin{lemma}
	\label{lem:10}
	$v_{1k}=w_{1k} = w_{2k}=v_{2k}$.
\end{lemma}
\begin{proof}
	Recall that $q^1$ contains one $\N$-item of type $t(\smallestred)$, i.e. the same type as $\smallestred$, and one item larger than $1-s(\smallestred)$. Call the $\N$-item $\smallestred'$ and the large one $\mathfrak{L}$; note that $t(\mathfrak{L})=2$. We have that $w_k(q^1) = w_k(\smallestred') + w_k(\mathfrak{L}) + S$, where $S$ is an upper bound for the weight of the sand, and $v_k(q^1) = v_k(\smallestred') + v_k(\mathfrak{L}) + S$ (the maximum possible amount of sand and hence also its weight is equal in the two cases). As  $\redfrac_2 = 0$ ($\mathfrak{L}$ is larger than 1/2 and such items are never red), and $\mathfrak{L}$ is too large to be combined with $\smallestred$, 
	we have $w_k(\mathfrak{L}) = v_k(\mathfrak{L}) = 1/\bluefit_2=1$. 
	
	For $w_k(\smallestred')$, consider that $\smallestred'$ and $\smallestred$ have the same type, and as the mark of $\smallestred'$ is $\N$, we get $w_k(\smallestred') = \frac{1-\redfrac_{t(\smallestred)}}{\bluefit_{t(\smallestred)}} + \frac{\redfrac_{t(\smallestred)}}{\redfit_{t(\smallestred)}}$.
	For type $t(\smallestred)$, we have that $\leaves(t(\smallestred))<\needs(t(\smallestred))$ (Property \ref{prop:leaves-needs}). Therefore, $v_k(\smallestred') = \frac{1-\redfrac_{t(\smallestred)}}{\bluefit_{t(\smallestred)}} + \frac{\redfrac_{t(\smallestred)}}{\redfit_{t(\smallestred)}} = w_k(\smallestred')$.
	This shows that $v_{1k}=w_{1k}$.
	
	The pattern $q^2$ contains one $\B$-item of type $t(\smallestred)$ (denoted by $\smallestred''$) and one item larger than $1-\smallestred$ (again denoted by $\mathfrak{L}$). 
	We have $w_k(\smallestred'') = w_k(\smallestred')$ since the weight $w_k$ is the same for $\N$- and $\B$-items of the same class. As above, we find  $w_k(\mathfrak{L})
	=v_k(\mathfrak{L})=1$ and 
	$v_k(\smallestred'') = \frac{1-\redfrac_{t(\smallestred)}}{\bluefit_{t(\smallestred)}} + \frac{\redfrac_{t(\smallestred)}}{\redfit_{t(\smallestred)}} = w_k(\smallestred'')$.
	This shows $w_{2k}=v_{2k}$ and $w_{1k}=w_{2k}$.
	\qed\end{proof}

For the case {when $\smallestred$ is small}, we do not have conditions
(\ref{c1:q1fv}) and (\ref{c1:q2fv}), and
the linear program $P_w^{k,\text{sml}}$ looks as follows.
Here we denote the set of patterns simply by $\Q$ since it is the same for all values of $k$
for which $\redspace_k\le1/3$. In this setting, $q^1,q^2,q^3$ do not have a special meaning.

\LPblocktag{$P_w^{k,\text{sml}}$}{\label{linprog:lp2}}%
\begin{minipage}{\linewidth-2cm}
	\begin{flalign}
	& \max & \textstyle\sum_{i=1}^{|\Q|} \chi_i w_{ik} 
	\label{sh1:max}
	\\
	& \text{s.t.} & \textstyle\sum_{i=1}^{|\Q|} \chi_i\left(w_{ik}-v_{ik}\right) & \le  0
	&&  && \label{sh1:q1fv}
	\\
	& & \textstyle\sum_{i=1}^{|\Q|} \chi_i & \leq 1 
	&&  && \label{sh1:chifv}
	\\
	& & \chi(q) & \ge 0 
	&& \forall q\in\Q && \label{sh1:chi0fv}
	\end{flalign}~
\end{minipage}
%\begin{align}
%%\label{sh1:max}
%\max &&\quad\sum_{i=1}^{|\Q|} \chi_i w_{ik} \\
%%\label{sh1:q1fv}
%s.t.&&\quad
%\sum_{i=1}^{|\Q|} \chi_i\left(w_{ik}-v_{ik}\right)&\le 0 &&\\
%%\label{sh1:chifv}
%&&\quad\sum_{i=1}^{|\Q|} \chi_i&\le1&&\\
%%\label{sh1:chi0fv}
%&&\quad\chi(q)&\ge0\ &\forall q\in\Q&&\ 
%\end{align}
%
\vspace{-20pt}
\paragraph{Intermezzo}
It is useful to consider the value of $w_{1k}$ (etc.).
We have not discussed the values of the parameters yet. However, as an example, 
for the algorithm {\Hpp}, two of the types are $(341/512,1]$ and $(1/3,171/512]$ (types 1 and 18).
Let us consider the case where at the end of the input, an item of type
18 is alone in a bin, and no smaller items are alone in bins.
For this case, for {\Hpp}, the two weighting functions for the pattern which contains types 1 and 18 both evaluate to
\[ 1 + \frac{1-0.176247}{2} + \frac{0.176247}{1} + \frac{1}{1-\frac{1}{50}}\cdot\frac{1}{1536} \approx 1.58879.\]
In other words, a distribution $\chi$ consisting only of this one pattern immediately gives a lower bound of 1.58879 on the competitive ratio of {\Hpp}. 

Our improved packing of the medium items and our marking of them ensures that this distribution, where the optimal solution uses critical bins exclusively, can no longer be used, since it is not a feasible solution to $P_w^k$. This is the key to our improvement over \Hpp.

\subsection{Dual program}
\label{sec:dual}

Our general idea is as follows: We consider the duals of the linear programs given above. These dual LPs have variables $y_1,\ldots,y_4$ or $y_3,y_4$, respectively. Any feasible solution for the dual (which is a minimization problem) is an upper bound on the competitive ratio of our algorithm by duality and by (\ref{eq:compratio}). We are interested in feasible dual solutions with objective value $c$, where $c$ is our target competitive ratio. %Our goal is then to find feasible values for $y_1, y_2$ and $y_3$ (or only $y_3$) such that the dual becomes feasible.

\paragraph{Case 1: $\smallestred$  is small}
%{First consider the case that $\smallestred$  is small.}
The dual of \ref{linprog:lp2} is the following. %For consistency with the dual program we will present next, we have named the dual variables $y_3$ and $y_4$.

\LPblocktag{$D_w^{k,\text{sml}}$}{\label{linprog:dp2}}%
\begin{minipage}{\linewidth-2cm}
	\begin{flalign}
	& \nonumber   \min & y_4%\label{sh:dual} 
	\\
	& \text{s.t.} & (w_{ik}-v_{ik})y_3 + y_4   & \geq  w_{ik}                 && i=1,\dots,|\Q| && \label{sh:d3fv}\\
	& \nonumber            & y_i & \geq 0 && i=3,4 %\label{sh:d0fv}
	\end{flalign}~
\end{minipage}
%\begin{align}
%\label{sh:dual}
%\min &&y_4 \\
%\label{sh:d3fv}
%s.t.&&(w_{ik}-v_{ik})y_3 + y_4 &\ge w_{ik} & i=1,\dots,|\Q|\\
%\label{sh:d0fv}
%&&y_i &\ge0 & i=3,4
%\end{align}
%
%We additionally restrict our solutions to those where $y_3^* \in [0,1]$, the reason for which will become clear shortly. 
%Since we are not interested in the optimal solution for the dual, adding constraints is valid, just as long as we find a
%feasible solution with objective value $c$ at the end.	
If the constraint (\ref{sh:d3fv}) does not hold for pattern $q_i$ and a given
dual solution $y^*$, we have
\begin{equation}
\label{eq:heavy}
(1-y_3^*)w_{ik} + y_3^* v_{ik} > y_4^*
\end{equation}
We need to determine if there is a pattern such that (\ref{eq:heavy}) holds.
For $y_3^* \in [0,1]$, the left hand side of (\ref{eq:heavy}) represents a weighted average of the weights $w_{ik}$ and $v_{ik}$.
We add the condition $y_3\le1$ to $D_w^{k,\text{sml}}$.
A feasible solution with objective value $c$ and $y_3\le1$ exists for $D_w^{k,\text{sml}}$ if and only if
a feasible solution with objective value $c$ and $y_3\le1$ exists for $D_v^{k,\text{sml}}$, as (\ref{eq:heavy}) is
now symmetric in $w$ and $v$.
This means that feasibility of $D^{k,\text{sml}}_w$ and $D^{k,\text{med}}_v$ with  $y_3\le1$ can be checked at the same time.
Again, note that it is sufficient for our purposes to find a feasible solution.

We define $\omega_k(\p)= (1-y_3^*)w_{k}(\p) + y_3^* v_{k}(\p)$ for each item $\p$. 
Since $\smallestred$ is small, there are no marked items of type $t(\smallestred)$, so $\omega_k(\p)$ depends only on the type
and size of $\p$.
The problem of determining $W=\max_{q\in\Q}\omega_k(q)$ for a given value of $y_3^*$ is a simple knapsack problem, which is straightforward to solve using dynamic programming. 

All that remains to be done is to determine a value for $y_3^*$ for given $k$ such that $W\le c$.
{In order to do this, we use a binary search in the interval $[0,1]$.}
We start by setting $y_3^*=1/2$ and compute $W$. If $W\le y_4^*$, $D_w^{k,\text{sml}}$ and $D_v^{k,\text{sml}}$ have objective value at most $y_4^*$ and we are done. Else, the dynamic program returns a pattern $q$ such that $\omega_k(q)>y_4^*$. For this pattern $q$, we compare its weights according to $w$ and $v$.
If $w_{ik} > v_{ik}$, we increase $y_3^*$, else we decrease it (halving
the size of the interval we are considering).
If after 20 iterations we still have no feasible solution, we return
infeasible. This may be incorrect (it depends on how long we search), 
but our claimed competitive ratio depends only on the correctness of 
\emph{feasible} solutions.

Summarizing the above discussion, if {$\smallestred$  is small},
proving that an \EHarm{} algorithm is $c$-competitive can be done by 
running the binary search for $k=\needs({t^*})$ using $y_4^*=c$.
If $(y_3^*, y_4^*)$ is a feasible solution for $D_w^{k,\text{sml}}$, then $(1-y_3^*,y_4^*)$ is a feasible solution for $D_v^{k,\text{sml}}$.
%For our algorithm \SonofH{}, the values for $y_3^*$ can be found in Table \ref{tab:sh-y3} in Appendix \ref{sec:sh-all-params}.

\paragraph{Case 2: $\smallestred$  is medium}

For the more interesting case {when $\smallestred$  is medium}, the dual $D_w^{k,\text{med}}$ of the program $P_w^{k,\text{med}}$ is the following.

\LPblocktag{$D_w^{k,\text{med}}$}{\label{linprog:dp1}}%
\begin{minipage}{\linewidth-2cm}
	\begin{flalign}
	& \nonumber \min & y_4 \\
	& \text{s.t.} & y_1 + y_4
	& \ge w_{1k}
	&&  && \label{c:d1fv}
	\\
	& & %-\frac{1-\redfrac_{t(\smallestred)}}{1+\redfrac_{t(\smallestred)}} y_1 + 
	\frac1%{1-\redfrac_{t(\smallestred)}}
	2 y_2 + y_4 
	& \ge w_{2k} 
	&& && \label{c:d2fv}
	\\
	& & \omit\rlap{$\displaystyle -\frac{1-\redfrac_{t(\smallestred)}}{1+\redfrac_{t(\smallestred)}} m_i y_1 
		- y_2 \sum_{j:0<\needs(j)\le \leaves(t(\smallestred))} \frac{\redfrac_j}{\redfit_j} n_{ij}$}
	\notag
	\\
	& & + (w_{ik}-v_{ik})y_3 + y_4
	& \ge w_{ik}
	&& i=3,\dots,|\Q_k| && \label{c:d3fv}
	\\
	& & y_i
	& \ge 0
	&& i=1,2,3,4 && \label{c:d4fv}
	\end{flalign}~
\end{minipage}
%\begin{align}
%\min && y_4 \\
%%\label{c:d1fv}
%s.t.&& y_1 + y_4 &\ge w_{1k}&\\
%%\label{c:d2fv}
%&&\quad\quad -\frac{1-\redfrac_{t(\smallestred)}}{1+\redfrac_{t(\smallestred)}} y_1 + 
%\frac{1-\redfrac_{t(\smallestred)}}2 y_2 + y_4 &\ge w_{2k}&\\
%%\label{c:d3fv}
%&&\omit\rlap{$\displaystyle -\frac{1-\redfrac_{t(\smallestred)}}{1+\redfrac_{t(\smallestred)}} m_i y_1 
%- y_2 \sum_{j:\needs(j)\le \leaves(t(\smallestred))} \frac{\redfrac_j}{\redfit_j} n_{ij}$} \notag\\
%&&+ (w_{ik}-v_{ik})y_3 + y_4 & \ge w_{ik} &\quad i=3,\dots,|\Q_k| \\
%%\label{c:d4fv}
%&&y_i&\ge 0 & i=1,2,3,4
%\end{align}

{Again we restrict ourselves to solutions with $y_3^* \in [0,1]$.}
%
%If the objective value of $D^k_w$ as well as that of $D^k_v$ is at most some value $y^*_4$ (or if one is infeasible), then $y^*_4$ upper bounds the competitive ratio of our algorithm for this value of $k$ by duality and by (\ref{eq:compratio}). 
If the value $y_4^* = c \ge w_{1k}=w_{2k}$, the conditions (\ref{c:d1fv}) and (\ref{c:d2fv}) are automatically satisfied by (\ref{c:d4fv}). In this case we can set $y_1^*=0$ and $y_2^*=0$. In effect, this reduces $D_w^{k,\text{med}}$ to \ref{linprog:dp2}, for which we already know how to {find a feasible value for $y_3^*$}. We therefore ignore the entire marking done by the algorithm and set the weight for each item to be the weight for the case that its mark is not $\R$. Then weights again do not depend on marks and we apply the method from Case 1. 
%Note that in this case, the weight of the critical pattern $q^1$ is lower than the competitive ratio we want to prove; that is why we can ignore our markings and basically use Seiden's simpler method.

Let us now consider the case 
%\[
$y_4^* = c < w_{1k}.
$ %\]
For given $y_4^*$ we need to determine if $D^{k,\text{med}}_w$ and $D^{k,\text{med}}_v$ are feasible; this requires finding suitable values for $y_1, y_2$ and $y_3$. If a solution vector $y^*$ is feasible for $D^{k,\text{med}}_w$ (or $D^{k,\text{med}}_v$), 
$y_4^*< w_{1k}=v_{1k}=w_{2k}=v_{2k}$, and constraint (\ref{c:d1fv}) or (\ref{c:d2fv}) is not tight, then we can decrease $y_1^*$ and/or $y_2^*$ and still have a feasible solution.
%(We first decrease $y_1^*$ if needed and then $y_2^*$.)
We therefore restrict our search to
solutions for which (\ref{c:d1fv}) and (\ref{c:d2fv}) are tight, and $y_4^*< w_{1k}$.
Then 
\begin{align}
\label{eq:y1}
y_1^* &= w_{1k} - y_4^* &&> 0\\
\label{eq:y2}
y_2^* &= 2%\frac{2}{1-\redfrac_{t(\smallestred)}}
%\left(1+\frac{1-\redfrac_{t(\smallestred)}}{1+\redfrac_{t(\smallestred)}}\right)
(w_{1k}-y_4^*) &&>   0.
\end{align}

This means that given $y_4^* < w_{1k}$, we know the values of $y_1^*$ and $y_2^*$. 
We can therefore prove $y^*_4$ is a feasible objective value for $D^{k,\text{med}}_w$ {by giving $y_3^*$-values that make the linear program feasible.} 
If constraint (\ref{c:d3fv}) does not hold for pattern $q_i$ ($i\ge3$) and a given dual solution $y^*$, we have the
following by some simple rewriting:
\begin{equation}
\label{eq:y3}
(1-y_3^*)w_{ik} + y_3^* v_{ik} + \frac{1-\redfrac_{t(\smallestred)}}{1+\redfrac_{t(\smallestred)}} m_i y_1^* + y_2^* \sum_{j:0<\needs(j)\le \leaves({t(\smallestred)})} \frac{\redfrac_j }{\redfit_j} n_{ij} > y_4^*
\end{equation}
If this holds for some pattern $q$ that contains an $\R$-item, then it obviously also holds 
if we replace that $\R$-item by an $\N$-item of the same type.
This gives a pattern with the same values $m_i = q^{t(\smallestred)}_{-\B}(q^i)$ and $n_{ij}=q_j(q^i)$ but a higher value for $w_{ik}$.
It is therefore sufficient to check the patterns with $\N$-items.
The only exception to this is if replacing the $\R$-item by an $\N$-item would give pattern $q^1$, which does have weight
larger than $y_4^*$ and therefore violates (\ref{eq:y3}) (but constraint (\ref{c:d3fv}) does not involve pattern $q^1$).
We therefore check pattern $q^3$ separately.

We have $w_{3k} = 1 + \frac{1-\redfrac_{t(\smallestred)}}2 + \frac{1}{1-\eps} \eps$, 
$v_{3k} = 1 + \frac{1+\redfrac_{t(\smallestred)}}2 + \frac{1}{1-\eps} \eps$,
$m_3=1$, $n_{3j}=0$ for all $j$.
%By Property  \ref{prop:redfrac}, $v_{3k} < \frac53 + \frac{1}{1-\eps} \eps$.
Hence the left hand side of (\ref{eq:y3}) is at most $\frac32 + \frac{1}{1-\eps} \eps + y_1^* \le 1.516$ for $\eps<0.01,
y_1^*\le0.005$, and $y_3^*\le1/2$
using Property \ref{prop:q1unique}.
All our solutions will satisfy these constraints and thus we can ignore $\R$-items in the knapsack problem.
(For completeness, we check pattern $q^3$ separately in our program.)
%We can therefore ignore $\R$-items from this point on. As a result, the function $w_k(\p)$ (and $w_{ik}$) does not depend on the
%mark of $\p$ anymore (since we assume there are no items of mark $\R$).

We define a new weighting function $\omega(\p)$ for the items as given in Table \ref{tab:weights_knapsack},
which depends only on types and sizes (and not on marks).
%Note that $m_i = n_\N(q^i)$ counts only blue items.
%Let $w = w_{1k} = w_{2k}$.

\begin{table}[h]
	\caption{Weighting function $\omega(\p)$.}
	\label{tab:weights_knapsack}
	\centering
	\begin{tabular}{l|l}
		$t(\p)$ & $\omega(\p)$\\
		\hline
		$t(\smallestred)$ &
		$\displaystyle
		(1-y_3^*)\left(\frac{1-\redfrac_{t(\smallestred)}}{\bluefit_{t(\smallestred)}} + \frac{\redfrac_{t(\smallestred)}}{\redfit_{t(\smallestred)}}\right) + y_3^* v_{k}(\p) 
		+ \frac{1-\redfrac_{t(\smallestred)}}{1+\redfrac_{t(\smallestred)}} y_1^*$\\
		$j, 0<\needs(j)\le \leaves(t(\smallestred))$ & 
		$\displaystyle
		(1-y_3^*)w_{k}(\p) + y_3^* v_{k}(\p) + \frac{\redfrac_j}{\redfit_j} y_2^*$
		\\
		else & $(1-y_3^*)w_{k}(\p) + y_3^* v_{k}(\p)$
	\end{tabular}
\end{table}

%
%The problem of determining $W=\max_{q\in\Q}\omega(q)$ for a given value of $y_3^*$ and values $y_1^*$ and $y_2^*$ calculated from the given $y_4^*$ is a simple knapsack problem, which is straightforward to solve using dynamic programming. If $W\le y_4^*$, $D_w^k$ and $D_v^k$ have objective value at
%most $y_4^*$ and we are done. Else, the dynamic program returns a pattern
%$q$ such that $\omega(q) > y_4^*$. This gives us our separation oracle.
%
%For given $y_4^* \le w_{1k}$, we can now determine feasibility of (\ref{c:d1fv})--(\ref{c:d4fv}) by using the ellipsoid method, fortunately for only one dimension: also known as binary search.
%We are only interested in feasible solutions in which $y_3^*\in[0,1]$, since only those represent weighted averages of $w_k$ and $v_k$.
%So we start the binary search for $y_3^*=1/2$. If $W\le y_4^*$, $D_w^k$ and $D_v^k$ are feasible. Else, for the found pattern $q_i$ with $\omega(q_i)>y_4^*$, we compare its weights according to $w$ and $v$.
%If $w_{ik} > v_{ik}$, we increase $y_3^*$, else we decrease it (halving
%the size of the interval we are considering).
%If after 20 iterations we still have no feasible solution, we return
%infeasible. This may be incorrect (it depends on how long we search), 
%but our claimed competitive ratio depends only on the correctness of 
%\emph{feasible} solutions.
%
%Summarizing the above discussion, if $s(\smallestred)>1/3$,
%proving that an algorithm is $c$-competitive can be done by 
%running the binary search for $k=\needs(t(\smallestred))$ using $y_4^*=\min(c,w_{1k})$.

In order to prove that an \EHarm{} algorithm is $c$-competitive if {$\smallestred$  is medium} and $c<w_{1k}$,
it is sufficient to verify that there exists a value
$y_3^*\in[0,1]$ such that $\max_{q\in Q_k} \omega(q)\le c$. {Values for $y_3^*$ that satisfy this can again be found in Appendix \ref{sec:sh-all-params}. Finding these values was done again }
by a binary search for each value of $k$ for which $\redspace_k>1/3$, each time setting $y_4^*=c$ and using (\ref{eq:y1}) and (\ref{eq:y2}). 
%In case $c>w_{1k}$, we run the binary search for (\ref{linprog:dp2}).

\paragraph{Summary}
{
	Overall, our approach is as follows: We first fix a target competitive ratio $c$. We do the following for every value of $k\in\{1,\ldots,K\}.$ Consider the value for $y_3^*$ (for our algorithm \SonofH{}, these values are specified in Table \ref{tab:sonofharm-y3}). If $\smallestred$  is small, we check that \ref{linprog:dp2} is feasible for $y_4^*=c$ and this $y_3^*$. If $\smallestred$  is medium, we compute $w_{1k}$ and check whether $w_{1k} \le c$ or $w_{1k} > c$. In the latter case, we again check that \ref{linprog:dp2} is feasible for $y_4^*=c$ and the given value of $y_3^*$. In the former case, we check that \ref{linprog:dp1} is feasible for $y_4^*=c$ and the given value of $y_3^*$.} Finally, for $k=K+1$, it is sufficient to count blue bins, and we solve a single knapsack problem based on $w_k$ alone, {checking that the heaviest pattern is not heavier than $y_4^*=c$.
}

%{In order to find the $y_3$-values,} we run $K$ binary searches {(no binary search is needed for $k=K+1$)}.

\subsection{Solving the knapsack problems}\label{sec:implementation-knapsack}

In order to prove our competitive ratio $c=\finalratio$, we prove feasibility of the discussed dual linear programs, which amounts to solving knapsack problems and comparing the maximum weight of a pattern to our target competitive ratio. We will now describe how our implementation of this knapsack solving works, given a set of item types as described at the beginning of Section \ref{sec:offline-solution} and a corresponding weight function $w$ (one weight per type).

We use two main heuristics to speed up the computation.
First, for each type $i$, we define the expansion $\text{exp}_i$ of type $i$ as the weight according to function $w$ divided by $t_{i+1}$. Now we sort the types in decreasing order of expansion; call this permutation of types $\pi$.
When constructing a pattern with high weight, we try to add items in the order of this permutation. Note that $\pi$ will not contain types that have expansion below that of sand: Such types will not be part of a maximum weight pattern, as the pattern with sand instead of these items has no smaller weight.

Second, we use branch and bound. We use a variable $\text{maxFound}$ that will store the maximum weight of a pattern found so far,
and give this the initial value $c-1/1000$. Whenever the current pattern cannot be extended to a pattern with weight more than
$\text{maxFound}$ (based on the expansion of the next item in the ordering $\pi$ that still fits), we stop the calculation for
this branch. Initializing $\text{maxFound}$ with a value close to $c$ immediately eliminates many patterns.

The process works as follows.
We start with type $t=\pi(1)$ (i.e., the type with the largest expansion) and an empty pattern.
For current type $t=\pi(j)$ and current pattern $q$ that contains items of total size $S$ and total weight $w(q)$ (counting only the non-sand items in the calculation of the weight), we compute an upper bound on the weight that this pattern $q$ can at most get by adding items of types $\pi(j), \pi(j+1), ... \pi(\rTypes)$, as follows. We find the first type $i$ in this order that still fits with the items of $q$ and compute $u = w(q) + (1-S)\text{exp}_i$. This is an upper bound for the weight of any bin which contains the items from $q$. If this upper bound is already smaller than $\text{maxFound}$, we immediately cancel the further exploration of this pattern $q$.

Otherwise, if we have no more types to add (i.e. we reached the end of list of types in $\pi$), set $\text{maxFound}$ to the weight of $q$ (now including the sand) and store $q$ as the heaviest pattern so far. If we still have more types to explore, find out how many items of the next type can fit maximally into $q$; call this number $m$ (if adding an item of the next type would create pattern $q^1$ or $q^2$ and we are considering the dual program $D^{k,med}_w$, we set $m=0$ as we do not need to consider these patterns). Now recursively call this procedure with type $\pi(j+1)$ and patterns $q_0, \ldots, q_m$ where $q_i$ is obtained from $q$ by adding $i$ items of type $t$.

The heuristics described in this section are still not enough to be able to examine all possible patterns in reasonable time. We explain in the next section how to reduce the set of patterns further (by reducing the number of small types) and how to ensure that larger items are more important than smaller items (by making sure the expansion of small items is monotonically nondecreasing in the size, that is, larger (but still small) items do not have smaller expansions than smaller items).

\section{The algorithm \SonofH}

For our algorithm {\SonofH} we have set initial values as follows.
The right part of Table \ref{tab:sonofharm} below contains item sizes and corresponding $\redfrac_i$
values that were set manually.
Some numbers of the form $1/i$ until the value $t_\rTypes$ are added automatically
by our program if they are not listed below (see below for details on how these are selected).

\begin{table}[h]
	\caption{Parameters and item types used for \SonofH{}.}
	\label{tab:sonofharm}
	\centering
	\subfloat[Parameters]{
		\begin{tabular}{|r|l|}\hline
			Parameter & Value\\ \hline
			$c$ & $\frac{15813}{10000}$\\
			$t_\rTypes$ & $\frac{1}{4000}$\\
			$\Gamma$ & $\frac{2}{7}$ (starting from $\frac1{14}$)\\
			$\mathcal{T}$ & $\frac{1}{50}$\\\hline
		\end{tabular}
	}
	\subfloat[Size lower bounds and values $\redfrac_i$]{
		\begin{tabular}{|r|l|}\hline
			Item size & $\redfrac_i$\\ \hline
			33345/100000 & 0\\
			33340/100000&0\\
			33336/100000&0\\
			33334/100000&0\\
			5/18&2/100\\
			7/27&105/1000\\
			1/4&1061/10000\\
			\hline
		\end{tabular}
		\begin{tabular}{|r|l|}\hline
			Item size & $\redfrac_i$\\ \hline
			8/39&8/100\\
			1/5&93/1000\\ 
			3/17&3/100\\
			1/6&8/100\\
			3/20&0\\
			29/200&0\\
			1/7&16/100\\ 
			\hline
		\end{tabular}
	}
\end{table}

The remaining values $\redfrac_i$ are set automatically  
using heuristics designed to speed up the search and minimize the
resulting upper bound. In the range $(1/3,1/2]$, we automatically
generate item sizes (with corresponding values $\redfrac_i$ and $\redspace_i$) that are less than $t_\rTypes$ apart to ensure uniqueness of $q^1$
and $q^2$: no non-sand item can be packed into any bin of pattern $q^1$ or $q^2$. 
The value $\Gamma$ specifies an upper bound on how much room is used by
red items of size at most $1/14$; larger items ($\le1/3$) use at most $1/3$ room.
Since we have this bound $\Gamma$, we also add size thresholds of the form $\Gamma/i$ for $i=1,2,3,4$, to ensure that items just below this threshold can be packed without leaving much space unused.

The last parameter is some item size $\mathcal{T}=t_j$.
Above this size, 
we generate all item sizes of the form $1/i$ for $i>3$. 
Below this size, we skip some item sizes as described below.
%using a heuristic that ensures that
%the performance of the algorithm is not affected (basically, we make sure that
%the expansion (i.e., the weight of the item divided by its size) does not become larger than that of items above this threshold).

Our program uses an exact representation of fractions, with numerators and denominators of potentially unbounded size, in order to avoid rounding errors. 
The source code and the full list of all types and parameters as determined by the program can be found at  \url{https://sheydrich.github.io/ExtremeHarmonic/}. 
In Appendix \ref{sec:sh-all-params}, we provide an alternative set of parameters, which give a competitive ratio of 1.583 with a much smaller set of knapsack problems to check.

Additionally, in Table \ref{tab:sonofharm-y3} we provide the $y_3^*$-values that certify the competitive ratio of our algorithm. Note that only two different values for $y_3^*$ were used.

\begin{table}[h]
	\caption{$y_3^*$-values used to certify that \SonofH{} is $\finalratio$-competitive.}
	\label{tab:sonofharm-y3}
	\centering
		\begin{tabular}{|r|l||r|l|}\hline
			$y_3^*$ & range of $k$ & $y_3^*$ & range of $k$ \\\hline
			$\frac{9}{32}$ & $k\le 4,$ & $\frac{3}{16}$ & $k=5,$\\
			& $6\le k \le 7, $ && $8 \le k \le 43, $\\
			& $44\le k\le 49$ && $k > 49$\\\hline
%			$k=5$ & $\frac{3}{16}$\\\hline
%			$k=6, 7$ & $\frac{9}{32}$\\\hline
%			$8 \le k le 43$ & $\frac{3}{16}$\\\hline
%			$44 \le k \le 49$ & $\frac{9}{32}$ \\\hline
%			$k > 49$ & $\frac{3}{16}$\\\hline
		\end{tabular}
\end{table}

\paragraph{Automatic generation of item sizes}\label{sec:automatic-size-generation}
We start by generating all item sizes of the form $1/i$ for $i$ between 2 and $\mathcal{T}$ (if they are not already present in the parameter file). After that, we generate types above $1/3$ in steps of size $t_\rTypes$. By choosing this step size, we make sure that no non-sand items can be added to the patterns $q^1, q^2, q^3$. 
The value $\redfrac_j$ for such a type $j$ is chosen such that the pattern containing an item $\smallestred'$ of type $j$ and a large item $\mathfrak{L}$ of type $2$ (i.e., $t_{i+1}=1/2$) has as weight exactly our target competitive ratio if $k=K+1$.
That is, we consider the weighting function $w_{K+1}$. We have $w_{K+1}(\smallestred') = \frac{1-\redfrac_j}{2}$, $w_{K+1}(\mathfrak{L}) = 1$, and an upper bound for the amount of sand that fits with these items is 
$1/2-t_{j+1}$. Therefore, $\redfrac_j$ is defined as the solution of the equation
\begin{equation}
\label{eq:calcred}
1 + \frac{1-\redfrac_j}{2} +\frac{1}{1-\eps}\left(\frac12 - t_{j+1}\right) = c = \finalratio,
\end{equation}
as long as this value is positive. We stop generating types as soon as it becomes negative.
To be precise, our highest value $t_{j+1}$ is defined by taking $\redfrac_j=0$ in (\ref{eq:calcred}).

We have now generated all item sizes above $\mathcal{T}$. 
%Note that there is no apparent need to generate any item sizes above $1/2$, as these are treated as a single type by {\EHarm} algorithms. Nevertheless, 
%In order to represent the patterns $q^1$ and $q^2$ in the code,
We generate large types as described in Section \ref{sec:implementation-knapsack}.
In the range $(\mathcal{T},t_\rTypes)$, we do not generate all $1/i$ types, but we skip some (to speed up the knapsack search) if this can be done without a deterioriation in the competitive ratio. We do this by considering the expansion of such items, that is, the weight divided by the infimum size. We will ensure that the expansion of smaller items is smaller than that of larger items, so that they are irrelevant (or less relevant) for the knapsack problem.

Let us consider how we test whether a certain type $(1/j,x]$ is required (where $x$ is the next larger type, i.e. either the last type generated before we started this last phase or the last type generated in this phase), and which $\redfrac_i$ we should choose. Denote by $s_i := 1/j$ the value we want to check. We compute a lower and upper bound $\underline{\redfrac}_i, \overline{\redfrac}_i$ for the $\redfrac_i$-value of this type as follows: We can compute $\bluefit_i$ and $\redfit_i$ only depending on the upper bound of the size of items of this type, i.e. depending on $x$, the lower bound of the next larger item size. First, we require $\frac{1-\redfrac_i}{\bluefit_it_{i+1}} \le 1$, which gives $\redfrac_i \ge 1-s_i\cdot \bluefit_i =: \underline{\redfrac}_i$. Second, we want to make sure that the maximum expansion of the current type is not larger than the expansion of the previous (next larger) type (since that might slow down the search), $exp_{i-1}$: $\frac{1-\redfrac_i}{\bluefit_is_i} + \frac{\redfrac_i}{\redfit_is_i} \le exp_{i-1} \Leftrightarrow \redfrac_i \le \frac{\bluefit_i\cdot\redfit_i}{\bluefit_i-\redfit_i} \left(exp_{i-1}s_i - 1/\bluefit_i\right) =: \overline{\redfrac}_i$. If $\underline{\redfrac}_i \le \overline{\redfrac}_i$, we continue to test $(1/(j+1),1/j]$; if not, we know that the \textit{previously tested} type is necessary to ensure the two constraints. Hence, we add this previous type to the list of types, together with the value $\underline{\redfrac}_{i-1}$ computed in the previous iteration.

\paragraph{Computation of $\redspace$-values}
The $\redspace$-values are completely auto-generated, in contrast to Seiden's paper, where these values are given by hand. For every type $i$ such that $t_{i+1} \in [1/6, 1/3]$, $t_{i+1}$ is added as a $\redspace$-value and for every type $i$ such that $2\cdot t_{i+1} \in [1/6, 1/3]$, $2\cdot t_{i+1}$ is added as a $\redspace$-value.
Additionally, we make sure that for each medium type we have a $\redspace$-value equal to $x$ and one equal to $1-2x$ where $x$ is the lower bound of the size of items of this type.

After computing the functions $\leaves$ and $\needs$, we then eliminate $\redspace$-values that are unused and less than $1/3$, i.e., if there is no pair of types
$i,j$ such that $\needs(i) = \leaves(j) = l, \redspace_{l}< 1/3$, then $\redspace_{l}$ is removed from the
list. This reduces the number of knapsack problems that need to be solved.

\paragraph{Computation and adjustment of values $\redfrac_i$}
\label{sec:alpha-computation}

For each item type $i$ that has size at most $1/6$ and at least $\mathcal{T}$, 
we adjust the value $\redfrac_i$ such that 
$
\frac{1-\redfrac_i}{\bluefit_i \cdot t_{i+1}} \ge f
$
where $f=\frac{95}{100}$ if $t_{i+1}\le\frac1{13}$ and $t_{i+1}>\mathcal{T}$ and $f=1$ otherwise.
To be precise, we set $\redfrac_i = 1-t_{i+1} \bluefit_i$. The reason for this is that it ensures that the
``small expansion'' of these items, where we count only the blue items of this type, is at least $f$.
This is a heuristic; it does not seem to help to make $\redfrac_i$ larger than this.

\section{Super Harmonic revisited}
\label{sec:sh}

%We revisit the \SuperH{} framework in this section. 

Seiden used the following weighting functions, but presented them in a different way.
Define $k$ and $\smallestred$ as in Definition \ref{def:6}.
The two weight functions of an item of type $i$ are given by Table \ref{tab:weights_seiden}.

\begin{table}[h]
	\caption{Weighting functions used by Seiden for \SuperH{}.}
	\label{tab:weights_seiden}
	\centering
	\begin{tabular}{ll|rl}
		\multicolumn{2}{c|} {
			$w_k(i)$} & 
		\multicolumn{2}{c} {
			$v_k(i)$}\\
		\hline
		$\frac{1-\redfrac_i}{\bluefit_i} + \frac{\redfrac_i}{\redfit_i}$ & if
		$\needs(i)\ge k$ or $\needs(i)=0$
		& $\frac{1-\redfrac_i}{\bluefit_i} + \frac{\redfrac_i}{\redfit_i}$ & if
		$\leaves(i)< k$\\
		$\frac{1-\redfrac_i}{\bluefit_i}$ & if
		$\needs(i)< k$
		& $\frac{\redfrac_i}{\redfit_i}$ & if
		$\leaves(i)\ge k$\\
	\end{tabular}
\end{table}

Using these weight functions, he shows that (\ref{eq:compratio}) with $c=1.58889$ holds for
\SuperH{} algorithms. Instead of the mathematical program that Seiden considers,
we use $P^{k,\text{sml}}_w$ and its dual $D^{k,\text{sml}}_w$.
We use the method described in Section \ref{sec:dual} (a binary search for a weighted average of weights)
to check for feasibility of the dual linear programs for all values of $k$, including the cases where {$\smallestred$  is medium}.
This is a significantly easier method than the one Seiden used, since it is based on solving standard knapsack problems.

A small modification of our computer program can be used to verify Seiden's result. Surprisingly, it shows that {\Hpp} is in fact \superhratio-competitive. In contrast to Seiden's heuristic program, which took 36 hours to prove {\Hpp}'s competitive ratio, our program terminates in a few seconds. Of course, this was over fifteen years ago, but we believe the algorithmic improvement explains a significant part of the speedup.
The fast running time of our approach also allowed us to improve upon {\Hpp} within the {\SuperH} framework (at least as long as we allow multiple red items per bin): Using improved $\redfrac_i$ values, we can show a \newsuperhratio-competitive {\SuperH}-algorithm. Our values are also simpler than the ones Seiden used (which were optimized up to precision $1/2\cdot 10^{-7}$); they can be found in the appendix.

\section{Lower bound}
We prove a lower bound for any \EHarm{} algorithm.
%Each bin may contain items of only two different types. 
We will consider inputs consisting of essentially four different item sizes: $1/2+\eps$, $1/3+\eps$, $1/4+\eps$, and $1/7+\eps$ (we also speak of types 1 through 4). Here $\eps$ is a very small number. However, there will be many different item sizes in the range
$(1/3,1/3+\eps]$.
The value of $\eps$ is chosen small enough that the algorithm puts all these sizes in the same type.
%Usually, these items will have these sizes for $\eps \rightarrow \infty$, except for one case explained in detail below, where we need to distinguish item sizes more carefully.
Note that the algorithm has not much choice about how many red items of types 2 and 3 can be packed in one bin: only one such item can be packed, else larger blue items could not be added anymore. For type 4, between 1 and 3 red items could be packed in one bin, and we will give lower bound constructions for each of these three cases.

Consider the case that the algorithm packs red type 4 items pairwise into bins. In Table \ref{tab:lb_redfit2}, we give four different inputs that together will prove a lower bound of \finallb{} for this case. A pattern $(a,b,c,d)$ denotes a set of items containing $a$ items of type 1, $b$ items of type 2 and so on. 
Note that our types defined here do not necessarily correspond to size thresholds used by the algorithm; nevertheless, each item gets a single type assigned by the algorithm, and if we use notation such as $\redfit_i$ for type $i$ as defined here, we mean the $\redfit$-value of the item type the algorithm assigns to such an item.
The other two columns of the table are explained below.

\begin{table}[t]
	\caption{Inputs for lower bound $1.5762$ in case $\redfit_4=2$.}
	\label{tab:lb_redfit2}
	\centering
	\subfloat{
		\begin{tabular}{ccc}
			Pattern & Space for $\eps \rightarrow 0$ & Distribution $\chi$\\ \hline
			$0\ 0\ 3\ 1$ & $\frac{31}{28}+2{\redfrac_3} + \frac{1-\redfrac_4}{6}
			+\frac{\redfrac_4}{2}$ & $1$\\
			\hline
			$1\ 1\ 0\ 0$ & $1 + \frac{1-\redfrac_2}2 + \frac16$ & $1$\\
			\hline
			$1\ 1\ 0\ 1$ & $1 + \frac{1-\redfrac_2}2 + \frac{1-\redfrac_4}{6} + \frac1{42}$ & $1$\\
			\hline
			$0\ 2\ 1\ 0$ & $2\cdot \frac{1+\redfrac_2}2 + \frac{1-\redfrac_3}3 + \frac1{12}$
			& $1$ (scaled) \\
			$q^1$ & $1 + \frac{1-\redfrac_2}2 + \redfrac_2$ & $\frac{2(1-\redfrac_2-\redfrac_2\cdot\redfrac_3)}{1+\redfrac_2}$\\
			$q^2$ & $1 + \frac{1-\redfrac_2}2 + \redfrac_2$ & $2\redfrac_3$\\
		\end{tabular}
	}
\end{table}

The first three lines of the table represent three different inputs to the algorithm, and the last three lines together represent the final input used in the lower bound.
We construct the first three inputs as follows. For each pattern in the table, 
items arrive in order from small to large. Each item in the pattern arrives $N$ times. In addition, we get $N$ times some amount of sand per bin, that fills up the bin completely. 
Based on each pattern and the values $\redfrac_i$ and $\redfit_i$, we can calculate exactly how much space (represented as fractions of bins) the online algorithm needs to pack each item in the pattern \emph{on average}. 
To do this, we assume that if red small items can be packed with larger 
blue ones, the algorithm will always do this (this is a worst-case assumption).
The result of this calculation is shown in the column Space.

%Hence, the space of a pattern gives a lower bound on the competitive ratio of the algorithm. 
%Formally, if $S(L)$ denotes the total space of an input sequence $L$ as defined here, 
%and $\alg(L)$ denotes the cost of an online algorithm $\alg$ for the input $L$,
%we find $\alg(L)\ge S(L)-O(1)$. 
%This is an opposite relation to what we saw in the rest of the paper.

To illustrate this approach, let us consider an input based on the pattern $(0,0,3,1)$ in the manner described above. %, which is given space 
%$\frac{31}{28} + 2\redfrac_3+\frac{1-\redfrac_4}{6}+\frac{\redfrac_4}{2}$ in the table. 
As we assumed that $\redfit_4=2$, we know that items of types $3$ and $4$ will not be 
combined by the algorithm, as $3/4+2/7>1$. Thus, the algorithm will not be able to combine the red items 
of both types with any other items. 
The number of bins used for blue type $3$ items is at least $3\cdot(\frac{1-\redfrac_3}{3})N$, 
the number of bins for red type $3$ items is at least $3\cdot\redfrac_3N$. 
Analogously, we need at least $\frac{1-\redfrac_4}{6}N$ bins for blue type $4$ items and at least 
$\frac{\redfrac_4}{2}N$ bins for red type $4$ items. Finally, sand of total volume arbitrarily close to $(1-3/4-1/7)N=\frac{3}{28}N$ 
arrives, which is packed in at least as many bins by the online algorithm. 
Thus, on average the items in this pattern need $3\cdot(\frac{1-\redfrac_3}{3}) + 3\cdot\redfrac_3 + \frac{1-\redfrac_4}{6} + \frac{\redfrac_4}{2} + \frac{3}{28} 
= \frac{31}{28} + 2\redfrac_3+\frac{1-\redfrac_4}{6}+\frac{\redfrac_4}{2}$ bins to be packed.
The space needed for the second and third patterns can be calculated in the same way.

\begin{figure}[h]
	\includegraphics[width=\textwidth]{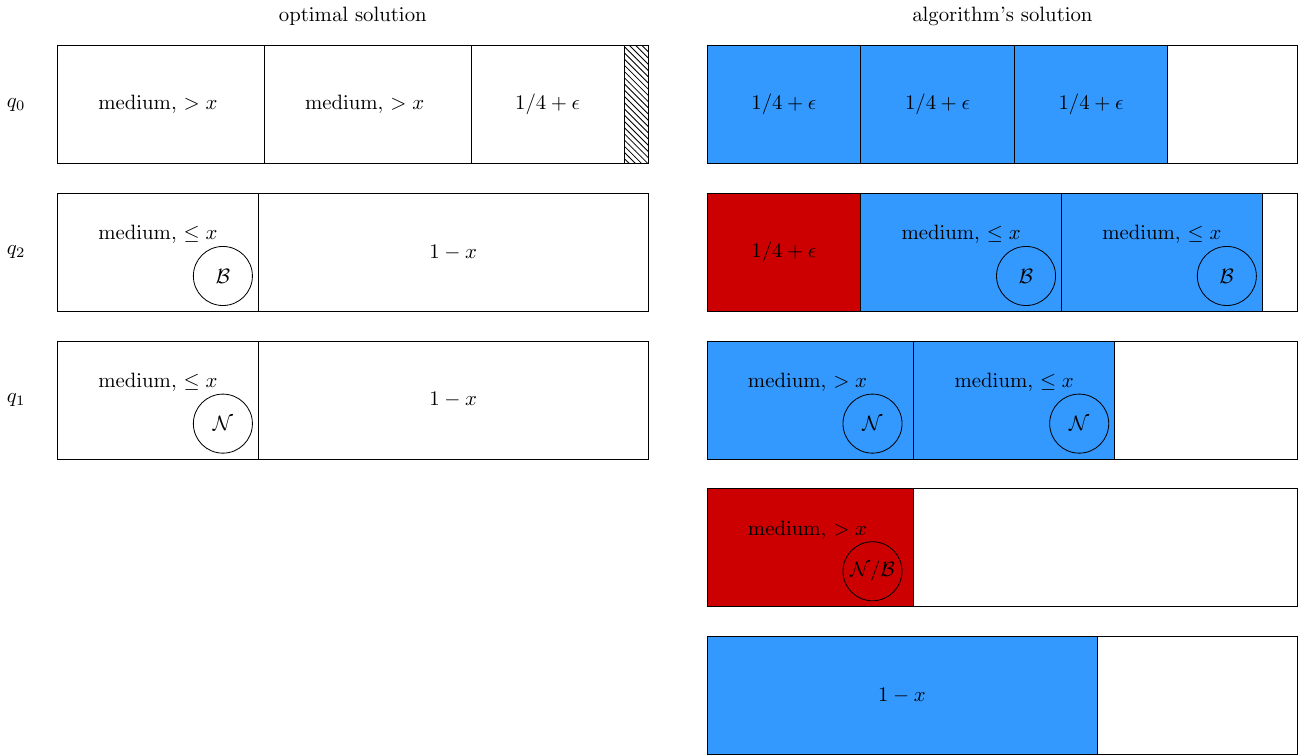}
	\caption{Fourth input for our lower bound construction. The three patterns used in the optimal solution are depicted on the left. The shaded area in the first pattern denotes sand. The algorithm produces the five types of bins depicted on the right, plus bins that only contain sand (not depicted here).}
	\label{fig:lower-bound}
\end{figure}

The fourth input (based on pattern (0,2,1,0)) requires more explanation; see also Fig. \ref{fig:lower-bound}.
For this input, we consider a combination of three patterns that arrive in the distribution given in the last column of the table.
Items of type 2 have size $1/3+\eps$ (according to the
table above) and some of them end up
alone in bins. We extend the input in this case by a number of items of size
almost $2/3$, where this number is 
calculated as explained below.
All these large items will be placed in new bins by the online algorithm.
In order for this to hold, the items of type 2 must have slightly different
sizes - not all exactly $1/3+\eps$.
We therefore pick $\eps$ small enough so that the interval
$(1/3,1/3+\eps]$ is contained in a single type according to the classification
done by the algorithm. The first item of this type
will have size $1/3+\eps/2$. The sizes of later items depend on how it is packed:
\begin{itemize}
	\item If the item is packed in a new bin, all future items
	will be smaller (in the interval $(1/3,1/3+\eps/2]$)
	\item If the item is packed into a bin with an existing item of type 2 or 3,
	all future items will be larger (in the interval $(1/3+\eps/2,1/3+\eps]$)
\end{itemize}
We use the same method for all later items of the same type, each time dividing the
remaining interval in two equal halves. 
By induction, it follows that whenever an item is placed in a new bin,
all previous items that were packed first into their bins are larger, and
all previous items that were packed into existing bins are smaller.
Therefore,
after all items of this type have arrived, let $x$ be the size of the last item
that was placed into a new bin. 
(Since the algorithm maintains a fixed fraction of red items of type 2,
there can be only constantly many items that arrived after this item; we ignore such
items.)
We have the following.
\begin{itemize}
	\item All items of size more than $x$ are packed either alone into bins or are the first item in a bin with two medium but no small red items; and 
	\item All items of size less than $x$ are in bins with items of type 3 or were
	packed as the second item of their type in an existing bin.
\end{itemize}
We now let items of size exactly $1-x$ arrive. 
For every bin with red type 3 items and blue type 2 items, two such items arrive, which will be packed in $q^2$-bins.
Assume that we have $N$ bins with pattern $q^0=(0,2,1,0)$, then we create exactly $\redfrac_3N$ such bins, i.e., we let $2\redfrac_3N$ large items arrive for these.
For every bin with a pair of blue medium items but no red items, one such $1-x$ item arrives.
The number of these bins is harder to calculate. Let $M$ be the total number of medium items in the input. 
Then the number of such bins is $\frac{1-\redfrac_2}{2}M-\redfrac_3N$. Now, we want to express $M$ in terms of $N$: Observe that $N$ is half the number of medium items larger than $x$ (as only these end up in $q^0$-bins). The number of those items is equal to the number of bins with red medium items (which is $\redfrac_2M$) plus the number of bins with two blue medium but no red items (which is $\frac{1-\redfrac_2}{2}M-\redfrac_3N$). Thus, $N$ is equal to $\frac{1}{2}\left( \redfrac_2M + \frac{1-\redfrac_2}{2}M-\redfrac_3N \right)$. This shows that $M=\frac{4+2\redfrac_3}{1+\redfrac_2}N$. Finally, we conclude that we can send  $\frac{1-\redfrac_2}{2}M-\redfrac_3N = \frac{1-\redfrac_2}{2}\cdot\frac{4+2\redfrac_3}{1+\redfrac_2}N-\redfrac_3N=\frac{2(1-\redfrac_2-\redfrac_2\cdot\redfrac_3)}{1+\redfrac_2}N$ many large items and thus get this many $q^1$-bins. % as well.

To pack $N$ copies of a given pattern, the online algorithm needs $N$
times the space calculated in Table \ref{tab:lb_redfit2}, while the optimal solution needs exactly $N$ bins. 
In order to calculate the final lower bound, for each of the four inputs, we simply calculate the space of the pattern(s), in the last case the weighted (in proportion to the distribution) sum of the three patterns' spaces. All four cases yield a lower bound of at least {\finallb}, which is achieved if $\redfrac_1=0,\redfrac_2=0.1800,\redfrac_3=0.1276,\redfrac_4=0.1428$.
Whenever an algorithm has a smaller or larger value for some $\redfrac_i$ value, the space needed by one of the patterns (or the weighted sum of the spaces needed by the three patterns of the last case) increases and thus gives a lower bound above \finallb{}.

Constructions for the other two cases $\redfit_4=1$ and $\redfit_4=3$ can be found below in Tables \ref{tab:lb_redfit3} and \ref{tab:lb_redfit1}.
The analysis is completely analogous to the first case.
For the case $\redfit_4=1$, the best values the online algorithm can use are $\redfrac_1=0,\redfrac_1=0.19,\redfrac_2=0.0872$.
The analysis for the case $\redfit_4=3$ is particularly simple, as the given distribution requires
$100/63$ bins on average (independent of $\redfrac_2$ and $\redfrac_3$), implying a lower bound of $100/63\approx1.5873$.

\begin{table}[h]
	\caption{Inputs for lower bound $1.5788$ in case $\redfit_4=1$.}
	\label{tab:lb_redfit1}
	\centering
	\subfloat{
		\begin{tabular}{ccc}
			Pattern & Space for $\eps\to0$ & Distribution $\chi$\\
			\hline 
			$1\ 1\ 1$ & $1 + \frac{1-\redfrac_2}2 + \frac{1-\redfrac_3}{6} + \frac1{42}$ & $1$\\
			\hline
			$0\ 0\ 6$ & $6\cdot \frac{1-\redfrac_3}6 + 6\redfrac_3 + \frac1{7}$ & $1$\\
			\hline
			$0\ 2\ 2$ & $2\cdot \frac{1+\redfrac_2}2 + 2\cdot \frac{1-\redfrac_3}6 + \frac1{21}$
			& $1$ (scaled) \\
			$q^1$ & $1 + \frac{1-\redfrac_2}2 + \redfrac_2$ & $\frac{4(1-\redfrac_2-\redfrac_2\cdot\redfrac_3)}{1+\redfrac_2}$\\
			$q^2$ & $1 + \frac{1-\redfrac_2}2 + \redfrac_2$ & $4\redfrac_3$
		\end{tabular}
	}
\end{table}
\vspace{-20pt}
\begin{table}[h]
	\caption{Inputs for lower bound $1.5872$ in case $\redfit_4=3$.}
	\label{tab:lb_redfit3}
	\centering
	\subfloat{
		\begin{tabular}{ccc}
			Pattern & Space for $\eps\to0$ & Distribution $\chi$\\
			\hline 
			$1\ 1\ 1$ & $1 + \frac{1-\redfrac_2}2 + \frac{1-\redfrac_3}{6} + \frac1{42}$ & $2/3$\\
			$0\ 2\ 2$ & $2\cdot \frac{1+\redfrac_2}2 + 2\cdot \frac{1+\redfrac_3}6 + \frac1{21}$ & $1/3$
		\end{tabular}
	}
\end{table}

%\begin{acknowledgements}
\vspace{-30pt}
\paragraph{Acknowledgements}
	We thank the anonymous referees for their useful comments, and Leah Epstein
	for interesting discussions.
%\end{acknowledgements}

%\vspace{-10pt}

\bibliographystyle{plain}
\bibliography{bibliography}

\appendix

\section{Alternative parameters for a competitive ratio of 1.583}\label{sec:sh-all-params}

\begin{table}%[h]
	\centering
	\subfloat[Parameters]{
		\begin{tabular}{|r|l|}\hline
			Parameter & Value\\ \hline
			$c$ & $\frac{1583}{1000}$\\
			$t_N$ & $\frac{1}{100}$\\
			$\Gamma$ & $\frac{2}{7}$ (starting from $\frac1{12}$)\\
			$\mathcal{T}$ & $\frac{1}{30}$\\\hline
		\end{tabular}
	}
	\subfloat[Size lower bounds and initial values $\redfrac_i$]{
		\begin{tabular}{|r|l|}\hline
			Item size & $\redfrac_i$\\ \hline
			335/1000 & 0\\
			334/1000&0\\
			5/18&2/100\\
			7/27&105/1000\\
			1/4&106/1000\\
			8/39&8/100\\
			1/5&93/1000\\ \hline
		\end{tabular}
		\begin{tabular}{|r|l|}\hline
			Item size & $\redfrac_i$\\ \hline
			3/17&3/100\\
			1/6&8/100\\
			3/20&0\\
			29/200&0\\
			1/7&135/1000\\
			1/13&1/10\\
			1/14&1/13 \\ \hline
		\end{tabular}
	}
	\caption{Parameters and item types.}
\label{tab:sonofharm2}
\end{table}

We give a list of item types together with their parameters in Table \ref{tab:sonofh_all_parameters}. 
Please note that type 2 is only defined for the definition of the knapsack problem in case {$\smallestred$  is medium}.
{\EHarm} algorithms, in contrast to {\SuperH} algorithms, treat all items larger than $1/2$ as a single type
(thus it sees types 1 and 2 as a single type).
Between type 6 and type 12, the values $t_i$ are $1/100$ apart. Between type 39 and type 101, the types are of the form $1/i$ for some values $i\in\{14, \ldots, 100\}$ (below $1/30$, we skip some values). The values $\redfrac_i$ for these types are computed as described in Sections \ref{sec:automatic-size-generation}. % and \ref{sec:alpha-computation}.
The paramters are auto-generated from the input in Table \ref{tab:sonofharm2}.

We give a list of all $\redspace$-values that are at most $1/3$ in Table \ref{tab:sh-deltas}. The $\redspace$-values above $1/3$ are equal to the $t_i$-values above $1/3$.
\begin{table}
	\centering
\begin{tabular}{|c|c||c|c||c|c|}\hline
	index $i$ & $\redspace_i$ & index $i$ & $\redspace_i$ & index $i$ & $\redspace_i$ \\ \hline
	$0$ & $0$ & $4$ & $11 / 50$ & $8$ & $7 / 25$ \\ \hline
	$1$ & $1 / 6$ & $5$ & $2 / 9$ & $9$ & $3 / 10$ \\ \hline
	$2$ & $3 / 17$ & $6$ & $6 / 25$ & $10$ & $8 / 25$ \\ \hline
	$3$ & $1 / 5$ & $7$ & $13 / 50$ & $11$ & $33 / 100$ \\ \hline
\end{tabular}
	\caption{$\redspace$-values below $1/3$ in the 1.583-competitive algorithm.\label{tab:sh-deltas}}
\end{table}
Finally, there were only two different $y_3^*$-values used to establish the feasibility of the dual LPs: $9/32$ for the cases $k=2,3,4,6,7,9,10,11$ and $3/16$ in all other cases.

{\small
	\newpage
	\begin{longtable}{|c|c|c|c|c|c|c|}
		\hline
		Type $i$ & $t_i$ & $\redfrac_i$ & $\bluefit_i$ & $\redfit_i$ & $\needs(i)$ & $\leaves(i)$ \\ 
		\hline \endhead 
		$1$ & $1$ & $0$ & $2$ & $0$ & $0$ & $0$ \\ 
		\hline
		$2$ & $41783 / 100000$ & $87 / 5500\approx 0.0158$ & $2$ & $1$ & $23$ & $0$ \\  \hline
		$3$ & $41 / 100$ & $1783 / 49500\approx 0.0360$ & $2$ & $1$ & $22$ & $2$ \\ \hline
		$4$ & $2 / 5$ & $253 / 4500\approx 0.0562$ & $2$ & $1$ & $21$ & $3$ \\ \hline
		$5$ & $39 / 100$ & $1261 / 16500\approx 0.0764$ & $2$ & $1$ & $20$ & $4$ \\  \hline
		$6$ & $19 / 50$ & $4783 / 49500\approx 0.0966$ & $2$ & $1$ & $19$ & $6$ \\ \hline
		$7$ & $37 / 100$ & $5783 / 49500\approx 0.1168$ & $2$ & $1$ & $18$ & $7$ \\ \hline
		$8$ & $9 / 25$ & $2261 / 16500\approx 0.1370$ & $2$ & $1$ & $17$ & $8$ \\ \hline
		$9$ & $7 / 20$ & $7783 / 49500\approx 0.1572$ & $2$ & $1$ & $16$ & $9$ \\ \hline
		$10$ & $17 / 50$ & $251 / 1500\approx 0.1673$ & $2$ & $1$ & $15$ & $10$ \\ \hline
		$11$ & $67 / 200$ & $8383 / 49500\approx 0.1694$ & $2$ & $1$ & $14$ & $11$ \\ \hline
		$12$ & $167 / 500$ & $25349 / 148500\approx 0.1707$ & $2$ & $1$ & $13$ & $11$ \\ \hline
		$13$ & $1 / 3$ & $0$ & $3$ & $0$ & $0$ & $0$ \\ \hline
		$14$ & $29 / 90$ & $0$ & $3$ & $0$ & $0$ & $0$ \\ \hline
		$15$ & $11 / 36$ & $1 / 50\approx 0.0200$ & $3$ & $1$ & $10$ & $0$ \\  \hline
		$16$ & $5 / 18$ & $21 / 200\approx 0.1050$ & $3$ & $1$ & $8$ & $1$ \\ \hline
		$17$ & $7 / 27$ & $53 / 500\approx 0.1060$ & $3$ & $1$ & $7$ & $5$ \\ \hline
		$18$ & $1 / 4$ & $2 / 25\approx 0.0800$ & $4$ & $1$ & $7$ & $0$ \\ \hline
		$19$ & $8 / 39$ & $93 / 1000\approx 0.0930$ & $4$ & $1$ & $4$ & $2$ \\ \hline
		$20$ & $1 / 5$ & $3 / 100\approx 0.0300$ & $5$ & $1$ & $3$ & $0$ \\ \hline
		$21$ & $3 / 17$ & $2 / 25\approx 0.0800$ & $5$ & $1$ & $2$ & $0$ \\ \hline
		$22$ & $1 / 6$ & $1 / 30\approx 0.0333$ & $6$ & $1$ & $1$ & $0$ \\ \hline
		$23$ & $29 / 180$ & $1 / 12\approx 0.0833$ & $6$ & $2$ & $11$ & $0$ \\ \hline
		$24$ & $11 / 72$ & $1 / 10 \approx 0.1000$ & $6$ & $2$ & $10$ & $0$\\ \hline
		$25$ & $3 / 20$ & $13 / 100\approx 0.1300$ & $6$ & $2$ & $9$ & $0$ \\ \hline
		$26$ & $29 / 200$ & $1 / 7\approx 0.1429$ & $6$ & $2$ & $9$ & $0$ \\ \hline
		$27$ & $1 / 7$ & $1 / 8\approx 0.1250$ & $7$ & $2$ & $9$ & $0$ \\ \hline
		$28$ & $1 / 8$ & $1 / 9\approx 0.1111$ & $8$ & $2$ & $7$ & $0$ \\ \hline
		$29$ & $1 / 9$ & $1 / 30\approx 0.0333$ & $9$ & $2$ & $5$ & $0$ \\ \hline
		$30$ & $29 / 270$ & $1 / 12\approx 0.0833$ & $9$ & $3$ & $11$ & $0$ \\ \hline
		$31$ & $11 / 108$ & $1 / 10\approx 0.1000$ & $9$ & $3$ & $10$ & $0$ \\ \hline
		$32$ & $1 / 10$ & $1 / 11\approx 0.0909$ & $10$ & $3$ & $9$ & $0$ \\ \hline
		$33$ & $1 / 11$ & $1 / 12\approx 0.0833$ & $11$ & $3$ & $8$ & $0$ \\ \hline
		$34$ & $1 / 12$ & $1 / 30\approx 0.0333$ & $12$ & $3$ & $7$ & $0$ \\ \hline
		$35$ & $29 / 360$ & $8 / 65\approx 0.1231$ & $12$ & $3$ & $7$ & $0$ \\ \hline
		$36$ & $1 / 13$ & $163 / 2880\approx 0.0566$ & $13$ & $3$ & $6$ & $0$ \\ \hline
		$37$ & $11 / 144$ & $33 / 280\approx 0.1179$ & $13$ & $3$ & $6$ & $0$ \\ \hline
		$38$ & $1 / 14$ & $17 / 150\approx 0.1133$ & $14$ & $4$ & $9$ & $0$ \\ \hline
		\vdots & \vdots & \vdots & \vdots & \vdots & \vdots & \vdots \\ \hline
		$122$ & $1 / 98$ & $7 / 660\approx 0.0106$ & $98$ & $28$ & $9$ & $0$ \\ \hline
		$123$ & $1 / 99$ & $21 / 2000\approx 0.0105$ & $99$ & $28$ & $9$ & $0$ \\ \hline
		\caption{Parameters used by 1.583-algorithm. The values $t_4$ to $t_{13}$ are in $\D$.}
		\label{tab:sonofh_all_parameters}
	\end{longtable}
}

\section{Parameters for an improved {\SuperH} algorithm}

With the parameters listed below in Table %s \ref{tab:improved_superh_2} and
\ref{tab:improved_superh}, using the same item types used by Seiden, we are able to achieve a {\SuperH} algorithm (with more than one red item per bin) with competitive ratio \newsuperhratio.

Note that the $\redfit_i$-values are computed differently than in \Hpp{}. For types $i$ such that $\redfrac_i=0$ or $t_i > \redspace_K$, we have $\redfit_i=0$. Otherwise, we have $\redfit_i=\lfloor \frac{24/83}{t_i} \rfloor$. The value $24/83$ in this expression is related to the item threshold $12/83$. By using this bound, two items of size slightly larger than $1/7$ can be packed together in one bin.

\iffalse 
\begin{table}[h]
	\caption{Redspace values used for our improvement of \Hpp{}.}
	\label{tab:improved_superh_2}
	\centering
	\subfloat{
	\begin{tabular}{|c|c|c|c|c|c|c|c|}
		\hline 
		$i$ & $\redspace_i$ & $i$ & $\redspace_i$ & $i$ & $\redspace_i$ & $i$ & $\redspace_i$\\
		\hline 
		1 & 6 / 83 &  6 & 1 / 4   & 11 & 127 / 384 & 16 & 65 / 192\\
		2 & 11 / 83 & 7 & 7 / 24  & 12 & 85 / 256  & 17 & 11 / 32\\
		3 & 13 / 88 & 8 & 5 / 16  & 13 & 171 / 512 & 18 & 17 / 48\\
		4 & 1 / 6 &   9 & 31 / 96 & 14 & 257 / 768 & 19 &  3 / 8\\
 		5 & 11 / 63 &10 & 21 / 64 & 15 & 43 / 128  & 20 &  5 / 12\\ \hline 
	\end{tabular}
	}
\end{table}
\fi 
\begin{table}[h]
%	\small
	\caption{Parameters used for our improvement of \Hpp{}.}
	\label{tab:improved_superh}
	\centering
	\subfloat{
		\begin{tabular}{|c|c|c|c|}
			\hline 
			$i$ & $t_i$ & $\redfrac_i$ & $\redfit_i$ \\ 
			\hline 
			1 & $1$ & 0 & 0 \\ 
			\hline 
			2 & $341/512$ & 0 & 0 \\ 
			\hline 
			3 & $511/768$ & 0 & 0 \\ 
			\hline 
			4 & $85/128$ & 0 & 0 \\ 
			\hline 
			5 & $127/192$ & 0 & 0 \\ 
			\hline 
			6 & $21/32$ & 0 & 0 \\ 
			\hline 
			7 & $31/48$ & 0 & 0 \\ 
			\hline 
			8 & $5/8$ & 0 & 0 \\ 
			\hline 
			9 & $7/12$ & 0 & 0 \\ 
			\hline 
			10 & $1/2$ & 0 & 0 \\ 
			\hline 
			11 & $5/12 $& $9/100 =0.09 $ & 1 \\ 
			\hline 
			12 & $3/8$ & $267/2000=0.1335$ & 1 \\ 
			\hline 
			13 & $17/48$ & $311/2000=0.1555$ & 1 \\ 
			\hline 
			14 & $11/32$ & $829/5000=0.1658$ & 1 \\ 
			\hline 
			15 & $65/192$ & $107/625=0.1712$ & 1 \\ 
			\hline 
			16 & $43/128$ & $87/500=0.174$ & 1 \\ 
			\hline 
			17 & $257/768$ & $7/40=0.175$ & 1 \\ 
			\hline 
			18 & $171/512$ & $877/5000=0.1754$ & 1 \\ 
			\hline 
			19 & $1/3$ & 0 & 0\\ 
			\hline 
			20 & $13/48$ & $37/400=0.0925$ & 1 \\ 
			\hline 
			21 & $1/4$ & $46/625=0.0736$ & 1 \\ 
			\hline 
		\end{tabular}
	}
	\subfloat{
		\begin{tabular}{|c|c|c|c|}
			\hline 
			$i$ & $t_i$ & $\redfrac_i$ & $\redfit_i$ \\ 
			\hline 
			22 & $13/63$ & $1/10=0.1$ & 1 \\ 
			\hline 
			23 & $1/5$ & $7/200=0.035$ & 1 \\ 
			\hline 
			24 & $15/88$ & $83/1000=0.083$ & 1 \\ 
			\hline 
			25 & $1/6$ & $789/10000=0.0789$ & 1 \\ 
			\hline 
			26 & $12/83$ & $13/100=0.13$ & 2 \\ 
			\hline 
			27 & $1/7$ & $29/2000=0.0145$ & 2 \\ 
			\hline 
			28 & $11/83$ & $71/1000=0.071$ & 2 \\ 
			\hline 
			29 & $1/8$ & $1191/20000=0.05955$ & 2 \\ 
			\hline 
			30 & $1/9$ & $1/20=0.05$ & 2 \\ 
			\hline 
			31 & $1/10$ & $9/200=0.045$ & 2 \\ 
			\hline 
			32 & $1/11$ & $4/125=0.032$ & 3 \\ 
			\hline 
			33 & $1/12$ & $11/500=0.022$ & 3 \\ 
			\hline 
			34 & $1/13$ & $71/2000=0.0355$ & 3 \\ 
			\hline 
			35 & $1/14$ & $17/2000=0.0085$ & 4 \\ 
			\hline 
			36 & $1/15$ & $1/100=0.01$ & 4 \\ 
			\hline 
			37 & $1/16$ & $1/100=0.01$ & 4 \\ 
			\hline 
			38 & $1/17$ & $1/100=0.01$ & 4 \\ 
			\hline 
			39 & $1/18$ & $0$ & 0 \\ 
			\hline 
			\vdots & \vdots & \vdots & \vdots \\ 
			\hline
			70 & $1/49$ & 0 & 0 \\
			\hline 
		\end{tabular}
	}
\end{table} %

\end{document}